\documentclass[draftcls, onecolumn]{IEEEtran}

\usepackage{hhline}
\usepackage{cite}
\usepackage[T1]{fontenc}
\usepackage{bbm}
\usepackage{amsmath}
\usepackage{amssymb}
\usepackage{amsfonts}
\usepackage{epsfig}
\usepackage{calc}
\usepackage{color}
\usepackage[all]{xy}
\usepackage{bm}
\usepackage{algorithmicx}
\usepackage{algpseudocode}
\usepackage{algorithm}
\usepackage{enumerate}
\usepackage{pdfsync}
\usepackage{framed}
\usepackage{caption}
\usepackage{amsthm}
\usepackage{blkarray}
\usepackage{subcaption}
\usepackage{enumitem}
\usepackage{hyperref}
\usepackage{mathrsfs}
\usepackage{caption}
\usepackage{float}
\newtheorem{defn}{Definition}[section]

\newtheorem{thm}{Theorem}[section]
\newtheorem{cor}[thm]{Corollary}
\newtheorem{prop}{Proposition}

\newtheorem{lem}[thm]{Lemma}
\newtheorem{conj}[thm]{Conjecture}
\newtheorem{constr}[thm]{Construction}
\usepackage[normalem]{ulem}

\newtheorem{remark}{Remark}[section]
\newtheorem{example}{Example}[section]

\newcounter{definition}[section]

\newcommand{\bit}{\begin{itemize}}
	\newcommand{\eit}{\end{itemize}}
\newcommand{\bcor}{\begin{cor}}
	\newcommand{\ecor}{\end{cor}}
\newcommand{\beq}{\begin{equation}}
\newcommand{\eeq}{\end{equation}}
\newcommand{\beqn}{\begin{equation*}}
\newcommand{\eeqn}{\end{equation*}}
\newcommand{\bea}{\begin{eqnarray}}
\newcommand{\eea}{\end{eqnarray}}
\newcommand{\bean}{\begin{eqnarray*}}
	\newcommand{\eean}{\end{eqnarray*}}
\newcommand{\ben}{\begin{enumerate}}
	\newcommand{\een}{\end{enumerate}}
\newcommand{\bdefn}{\begin{defn}}
	\newcommand{\edefn}{\end{defn}}
\newcommand{\bnote}{\begin{remark}}
	\newcommand{\enote}{\end{remark}}
\newcommand{\bprop}{\begin{prop}}
	\newcommand{\eprop}{\end{prop}}
\newcommand{\blem}{\begin{lem}}
	\newcommand{\elem}{\end{lem}}
\newcommand{\bthm}{\begin{thm}}
	\newcommand{\ethm}{\end{thm}}
\newcommand{\bconj}{\begin{conj}}
	\newcommand{\econj}{\end{conj}}
\newcommand{\bconstr}{\begin{constr}}
	\newcommand{\econstr}{\end{constr}}
\newcommand{\bpf}{\begin{proof}}
	\newcommand{\epf}{\end{proof}}

	\newcommand{\cstr}{\mbox{$\mathcal{C}_{\text{str}}$}}

\usepackage{url}

\usepackage{caption} 
\newcommand{\fq}{\mbox{$\mathbb{F}_q$}}
\newcommand{\calc}{\mbox{$\mathcal{C}$}}
\newcommand{\calp}{\mbox{$\mathcal{P}$}}
\newcommand{\calpc}{\mbox{$\mathcal{P}^c$}}

\newcommand{\calec}{\mbox{$\mathcal{E}^c$}}

\newcommand{\hrowl}{\mbox{$\underline{h}_{\text{\tiny row},\ell}$}}
\newcommand{\vij}{\mbox{$\{v_{i,j}\}$}}
\newcommand{\params}{\mbox{$\{a,b,\tau\}$}} 
\newcommand{\pl}{\mbox{$P_\ell$}}
\newcommand{\plset}{\mbox{$\{P_\ell\}$}} 
\newcommand{\rl}{\mbox{$R_\ell$}}
\newcommand{\hsupl}{\mbox{$H^{(\ell)}$}}
\newcommand{\al}{\mbox{$A_\ell$}}
\newcommand{\fqs}{\mbox{$\mathbb{F}_{q^2}$}}
\newcommand{\ul}{\mbox{$U_\ell$}}

\newcommand{\ml}{\mbox{$M_\ell$}}
\newcommand{\bl}{\mbox{$B_\ell$}}
\newcommand{\gmds}{\mbox{$G_\text{\tiny MDS}$}}
\newcommand{\cmds}{\mbox{$\mathcal{C}_\text{\tiny MDS}$}}

\begin{document}
	
	\title{Low Field-size, Rate-Optimal Streaming Codes for Channels With Burst and Random Erasures}
		\author{\IEEEauthorblockN{M. Nikhil Krishnan, Deeptanshu Shukla and P. Vijay Kumar, {\em Fellow, IEEE} }\\
		\IEEEauthorblockA{Electrical Communication Engineering, Indian Institute of Science, Bangalore - 560012 \\
			email: \{nikhilkrishnan.m, deeptanshukla, pvk1729\}@gmail.com}
		\thanks{P. Vijay Kumar is also a Visiting Professor at the University of Southern California. This research is supported in part by the National Science Foundation under Grant 1421848 and in part by an India-Israel UGC-ISF joint research program grant. M. Nikhil Krishnan would like to acknowledge the support of Visvesvaraya PhD Scheme for Electronics \& IT awarded by Department of Electronics and Information Technology, Government of India. Construction A presented in this paper is submitted in part to 2019 IEEE Int. Symp. Inf. Theory (ISIT) for possible publication. The remaining three constructions presented here are new.}}
	\maketitle
	\maketitle
	
%	\tableofcontents
	\begin{abstract}
	In this paper, we design erasure-correcting codes for channels with burst and random erasures, when a strict decoding delay constraint is in place. We consider the sliding-window-based packet erasure model proposed by Badr et al., where any time-window of width $w$ contains either up to $a$ random erasures or an erasure burst of length at most $b$. One needs to recover any erased packet, where erasures are as per the channel model, with a strict decoding delay deadline of $\tau$ time slots. Presently existing rate-optimal constructions in the literature require, in general, a field-size which grows exponential in $\tau$, for a constant $\frac{a}{\tau}$. In this work, we present a new rate-optimal code construction covering all channel and delay parameters, which requires an $O(\tau^2)$ field-size. As a special case, when $(b-a)=1$, we have a field-size linear in $\tau$. We also present three other constructions having linear field-size, under certain constraints on channel and decoding delay parameters. As a corollary, we obtain low field-size, rate-optimal convolutional codes for any given column distance and column span. Simulations indicate that the newly proposed streaming code constructions offer lower packet-loss probabilities compared to existing schemes, for selected instances of Gilbert-Elliott and Fritchman channels.
	\end{abstract}
	\section{Introduction}
	
	%Why is low latency important
	Reliable communication at low-latency often comes up as an important requirement in the design of next-generation communication systems, including 5G, augmented reality and IoT. Low latency is particularly crucial for real-time multimedia applications, autonomous navigation and V2X (vehicle-to-everything) communications, `working and playing' in the cloud, automation and remote management, tele-medicine and several other mission-critical scenarios \cite{series2015imt}. A recent study \cite{cisco} estimates that IP video traffic, a single use case of low-latency communication, will account for 82 percent of all consumer Internet traffic by 2021, up from 73 percent in 2016. The challenge of enabling delay-constrained communication is further exacerbated by issues arising out of noise, interference, fading, routing, mobility and reliability. In order to ensure robust performance under such a wide range of operating conditions, networks provide for error detection, concealment and correction schemes at multiple layers. These error control strategies can be classified under two broad heads;  re-transmission strategies, like Automatic Repeat Request (ARQ) protocols, and channel coding or forward error correction (FEC). Choosing one of these error control strategies or a suitable hybrid of both of them, is a critical design decision for any communication system. 
	
	\subsection{ARQ vs. FEC}
	Re-transmission based strategies, in general, add lower amount of redundancy compared to FEC, but incur an additional round-trip delay per re-transmission. This might be acceptable for error control on a per hop basis, as in the link layer, but can significantly exceed latency requirements for long-distance communication. Re-transmission also leads to more complicated protocols as the transmitter needs an acknowledgment from the receiver. If the message is received but its acknowledgment is lost, the sender will have to re-transmit, wasting time and bandwidth. Re-transmission based error control is also not amenable to multicasting, a common data streaming scenario. Each client may miss different packets and re-transmitting all of them may lead to a feedback implosion.
	
	On the other hand, FEC is a more natural fit for low-latency applications. It incurs no round-trip delays, no acknowledgment issues and no feedback implosion during multicasting.  %mention some references for this. Blocking of window in TCP. vs. 	FEC. Congestion control issues. 
	Even in re-transmission based schemes like TCP, it is shown in \cite{fec_tl,lt} that introducing FEC can lead to performance gains. 
	But these advantages of FEC come at the cost of injecting redundancy. Hence, the channel model and FEC parameters must be carefully calibrated to achieve optimal latency-redundancy tradeoff. 
	
	\subsection{Models for Handling Burst and Random Erasures}
	In end-to-end layers of the network, error control is mostly in the form of integrity checks such as checksums. These error detection features help the receiver infer if a packet has been received without any error. This can be naturally modeled by an erasure channel. This model also incorporates packet drops due to other factors such as congestion, mis-routing and buffer overflows.
	
	Coding-theoretic literature on erasure channels has focused on either random isolated erasures, such as the binary erasure channel, or on burst erasures. However, measurements on real-world systems \cite{seeme} indicate that erasures occur as isolated entities as well as in bursts. One means of modeling them is by using probabilistic channel models like Gilbert-Elliott and Fritchman Channels. However, such models are hard to analyze and even closed-form expressions for their capacities are not known. Thus, there is need for models rich enough to capture both isolated and burst erasures but simple enough to be tractable.

	%\cite{seeme} demonstrates that current forms of FEC, even in proprietary systems like Skype and FaceTime, are inadequate in recovering from burst erasures. Systems that treat burst erasures arising from packet drops in a congested network and random isolated erasures as the same can run into congestion cycles.  --  %strongly indicate that mobile video call quality is highly vulnerable to burst erasures.%Burst erasure occuring from weak signals or high throughput in mobile and cellular systems. 
	
	A second important consideration is whether to inject redundancy by introducing more packets per unit time (bandwidth expansion) or by increasing the packet size by adding redundancy within the packets (symbol expansion). In \cite{MartSunTIT04}, the authors argue in favor of symbol expansion, as burst erasures often occur due to congestion in a network. Introducing more packets under such circumstances may lead to a congestion cycle \cite{seeme} and degrade performance. Introducing new packets may also increase channel contention overhead \cite{SPMag}. Hence symbol expansion is often the preferred option. This leads to the question of what parities have to be added, i.e., what error-correcting code to be used. In response to these requirements, a new class of codes dedicated to transmitting packets over erasure channels under stringent decoding-delay constraints, named \textit{streaming codes}, has emerged in recent years.

	\subsection{A Brief History of Streaming Codes}
		While burst erasure correction has been studied for a very long time (for instance, see \cite{forney_burst,massey1965burst,lincostello,BurstOld1,firecode,BurstOld2,BurstErasureLDPC}), the problem of burst erasure correction under decoding delay constraints is relatively new and was first studied in \cite{MartSunTIT04}. Prior to this systematic study, off-the-shelf codes like Reed-Solomon combined with heuristics like interleaving, mean burst loss length (MBL) and mean inter-loss distance (MILD) were employed for combating burst erasures in latency-critical applications \cite{fec_wireless}. In their model,	Martinian and Sundberg \cite{MartSunTIT04} consider a channel which can introduce a burst erasure of length at most $b$. They incorporate latency-criticality in the model as a decoding delay constraint of $\tau$ packets, i.e., a packet transmitted at time $t$ must be recovered at the decoder by time $(t+\tau)$. The authors derive an upper bound on the rate of codes that can tolerate an erasure burst with delay at most $\tau$ and also obtain a family of rate-optimal codes for a wide range of parameters. The paper \cite{MartTrotISIT07} provides a code construction which achieves the rate upper bound in \cite{MartSunTIT04} for all parameters $\{b,\tau\}$. The authors of \cite{MartTrotISIT07} also introduce a {\it diagonal embedding technique} to design streaming codes using block codes as building blocks. In \cite{BadrPatilKhistiTIT17}, a richer sliding-window-based erasure channel is proposed and analyzed, wherein  any sliding-window of size $w$ can have either up to $a$ random erasures or an erasure burst of length at most $b$ (see Section \ref{sec:setting} for a detailed explanation). The authors of \cite{BadrPatilKhistiTIT17} also derive an upper bound on the rate of streaming codes which can tolerate all the erasure patterns of the sliding-window channel model, with a delay of at most $\tau$. The works \cite{FongKhisti,NikPVK} provide the first-known streaming codes that achieve the rate upper-bound in \cite{BadrPatilKhistiTIT17}, for all feasible parameters. However except for  a small range of parameters, the field-size requirements here are large; $>2({\tau+1 \choose a}+\tau-b+2)$ in \cite{FongKhisti} and $\sim(b-a)^{\tau+1}$ in \cite{NikPVK}. In Table ~\ref{tab:constructions}, we provide a summary of streaming code constructions  existing in the literature for these burst and sliding-window-based erasure channel models (including constructions from the present paper). Streaming codes have also been constructed for channels with unequal source-channel inter-arrival rates \cite{BadrPatilKhistiTIT17}, multiplicative-matrix channels \cite{robin} and  multiplexed communication scenarios with different decoding delays for different streams \cite{KhistiMultiplex}. In \cite{RudowRashmi18}, the authors consider a setting for variable-size arrivals. Several other models for delay-constrained communication have been proposed and analyzed in works such as  \cite{HoLeong,HoQureshiLeong,HoJaggi}. A comprehensive survey on streaming codes can be found in \cite{SPMag}.

	\begin{table}[t]
	\centering
	\begin{tabular}{||l|c|c|c|c|c||}
		%\hhline{||=#|=|=|=|=|#=||}
		\hhline{#======#} 
		
		%\hline\hline
		Streaming Code & Channel Model & Rate  & Field Size & Rate-Optimal? & Explicit?\\
		\hhline{#=|=|=|=|=|=#}%\hline\hline
		Maximally Short (MS) Codes\cite{MartSunTIT04} & Burst & $\frac{ms+1}{ms+s+1}$\footnotemark & $O(b)$ & For $ \frac{\tau}{b} = \frac{ms+1}{s}$ & Yes \\\hline 
		Delay-Optimal Burst Erasure Codes\cite{MartTrotISIT07} & Burst & $\frac{\tau}{\tau+b}$ & $O(b)$ & Yes & Yes\\ %\hline\hline
		\hhline{#=|=|=|=|=|=#}
		MiDAS-m-MDS Codes \cite{BadrPatilKhistiTIT17} & Sliding-Window & $\frac{\tau-a}{\tau-a+b}$ & $\exp(\tau)$ & Near-optimal & Yes\\\hline
		MiDAS-interleaved  Codes\cite{BadrPatilKhistiTIT17} & Sliding-Window & $\frac{\tau-a}{\tau-a+b}$ & $O(\tau^3)$ & Near-optimal & Yes\\\hline
		Embedded-Random Linear Codes \cite{BadrRateHalfINFOCOM13}  & Sliding-Window & $\frac{\tau-a+1}{\tau+b-a+1}$ & $\exp(\tau)$ & For $\frac{\tau-a+1}{\tau+b-a+1}=\frac{1}{2}$ & No\\\hline
		
		%\hline
		Random Convolutional Codes \cite{FongKhisti} & Sliding-Window & \textquotedbl & $\sim2{\tau+1 \choose a}$ & Yes & No\\\hline
		Construction A \cite{NikPVK} & Sliding-Window  & \textquotedbl & $O(\tau^2)$ & For $\tau \mod b \ge (b-a)$ or $b\mid \tau$ & Yes\\ \hline
		Construction B \cite{NikPVK} & Sliding-Window  & \textquotedbl & $\sim(b-a)^{\tau+1}$ & Yes & Yes\\ %\hline
		\hhline{#=|=|=|=|=|=#}
		Construction A  (present paper)& Sliding-Window & \textquotedbl & $O(\tau^2)$ & Yes & No\\\hline
		Construction A  (present paper)& Sliding-Window & \textquotedbl & $O(\tau)$ & For $(b-a)=1$ & No\\\hline
		Construction B (present paper)& Sliding-Window & \textquotedbl & $O(\tau)$ & For $(\tau+a+1) \ge 2b \ge 4a$ & Yes\\\hline
		Construction C (present paper)& Sliding-Window & \textquotedbl & $O(\frac{a}{b}\tau)$ & For $a\mid b\mid (\tau+1+b-a)$ & Yes\\ \hline
		Construction D (present paper)& Sliding-Window & \textquotedbl & $O(\tau)$ & For $b = 2a-1; \;\tau=a+ \gamma b-2, \gamma \in \mathbb{N}$ & Yes\\ %\hline	\hline
		\hhline{#=|=|=|=|=|=#}
	\end{tabular}
	\caption{Streaming code constructions for burst/sliding-window-based erasure channel models and their operating regimes.}
	\label{tab:constructions}
\end{table}

\addtocounter{footnote}{-1} 
\stepcounter{footnote}
\footnotetext{$m \ge 0$ and $s\ge 1$ are integer parameters.}
%\footnotetext{$t = \eta b + \xi$}
%\vspace*{-5pt}

	\subsection{Contributions of the Present Paper}
 We employ the diagonal embedding technique introduced in \cite{MartTrotISIT07} to reduce the problem of designing streaming codes to that of constructing linear block codes with certain properties. We then translate these properties as some requirements on the parity-check (p-c) matrix (which is of size $b\times (\tau+b-a+1)$) of the block code to be used for diagonal embedding. We provide four different families of code (p-c matrix) constructions, which will be referred to as Constructions A, B, C and D.  All these block codes when used in conjunction with diagonal embedding, will result in rate-optimal streaming codes for the sliding-window-based erasure channel model. 
 
 Construction A works for all parameters $\{a,b,\tau,w\}$. The p-c matrix that we obtain in Construction A is not completely explicit. There are $(\tau+1-b)(b-a)$ entries of the p-c matrix which are not explicitly specified. An application of Combinatorial Nullstellensatz \cite{AlonCombNul99} guarantees that there exist an assignment of values to these entries so that the resultant p-c matrix satisfies the required properties. These entries however, can be easily determined via a greedy algorithm \cite[Algorithm 1]{KoetterMedardNetworkCoding}. The remaining three constructions; Constructions B, C, D are explicit and cover a wide range of parameters. In terms of the field-size requirement, Construction A requires an $O(\tau^2)$ field-size, whereas the other three need a linear field-size. This is in contrast to the field-size requirements of currently existing rate-optimal streaming code constructions for the sliding-window based channel model, which grow in general, exponential in $\tau$, once we fix the ratio $\frac{a}{\tau}$. In Fig. \ref{fig:3d_plot}, we show all the valid parameters for the four constructions, when $\tau\leq 20$.

The rest of this paper is organized as follows; basic notation and some preliminary results regarding MDS codes are given in Section \ref{sec:MDS}. Section \ref{sec:setting} describes the coding theoretic framework employed in this paper, the channel model and the technique of diagonal embedding, which reduces the design of streaming codes to that of block codes satisfying some specific conditions. In Sections \ref{sec:construction_A}--\ref{sec:construction_D}, we present the four code constructions. In Section \ref{sec:conv_codes_col_distance_col_span}, by invoking results from \cite{BadrPatilKhistiTIT17}, we discuss how the new code constructions imply the existence of rate-optimal convolutional codes for given column distance ($d_{\tau}$) and column span ($c_{\tau}$), which require a lower field-size, when compared to other rate-optimal convolutional code constructions in the literature. Section \ref{sec:simulations} presents simulation results which indicate that the new streaming code constructions outperform existing streaming code constructions for some instances of GE and Fritchman channels.

\begin{figure}[!htb]
	\centering
	\captionsetup{justification=centering}
	\includegraphics[scale=0.75]{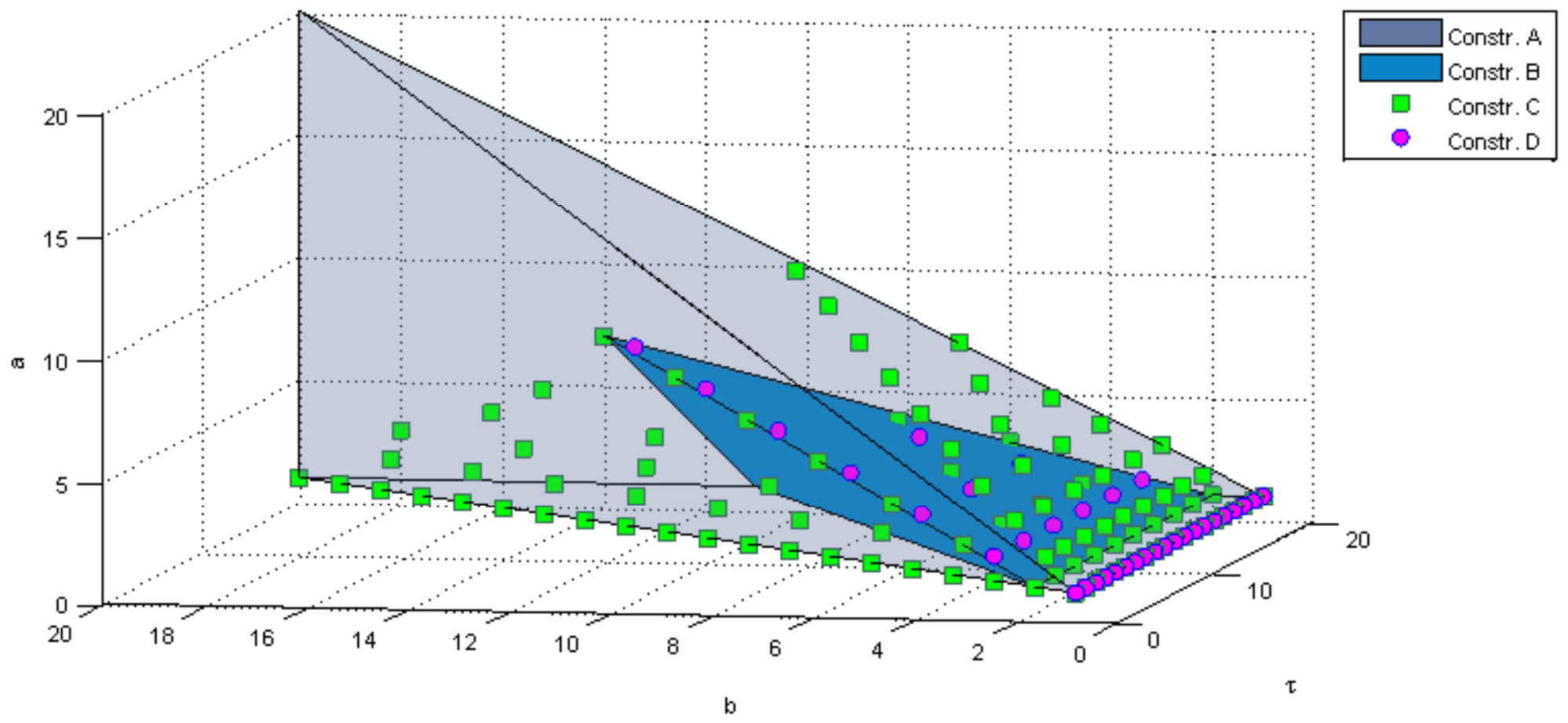}
	\caption{Let $w=(\tau+1)$. In the figure, we show all the valid parameter sets $\{a,b,\tau\}$ for Constructions A, B, C and D, where $a\leq b\leq \tau$, $\tau\leq 20$. Construction A covers the entire parameter range.}
	\label{fig:3d_plot}
\end{figure}	
	
	%\newpage
	
	%===============SECTION II ===================================
	%===============SECTION II ===================================
	%===============SECTION II ===================================
	\section{Punctured and Shortened Subcodes of an MDS Code}\label{sec:MDS}
	
	\subsection{Notation}
	
	For $m,n\in \mathbb{Z}$, let $[n]\triangleq \{i: 1\leq i\leq n\}$ and $[m:n]\triangleq\{i:m\leq i\leq n\}$. The $(n\times n)$ identity matrix will be denoted by $I_n$. For a row vector $\underline{v}=[v_0\ v_1\ \ldots\ v_{n-1}]\in\mathbb{F}_q^n$, $\text{supp}(\underline{v})\triangleq \{i: v_i\neq 0\}$. We write $m\mid n$ if $m$ divides $n$. Let $A\in \mathbb{F}_q^{m\times n}$, $\mathcal{I}\subseteq [0:m-1]$, $\mathcal{J}\subseteq [0:n-1]$. By $A(\mathcal{I},\mathcal{J})$ we mean the $(|\mathcal{I}|\times |\mathcal{J}|)$ submatrix of $A$ obtained by selecting the rows indexed by $\mathcal{I}$ and columns indexed by $\mathcal{J}$. For $i\in [0:m-1], j\in [0:n-1]$, $A(i,:)$ and $A(:,j)$ will denote the $i$-th row and $j$-th column, respectively.  Similarly, $A(\mathcal{I},:)$ and $A(:,\mathcal{J})$ will denote the sub-matrices of $A$ obtained by selecting the rows in $\mathcal{I}$ and columns in $\mathcal{J}$, respectively.  We will often use the alternate notation, $\underline{a}_{\text{\tiny row},i}$ and $\underline{a}_i$ to denote the $i$-th row and $i$-th column of a matrix $A$, respectively. 
	If ${\cal A} \subseteq [0:n-1]$, then we will use ${\cal A}^c$ to denote the complement of ${\cal A}$ in $[0:n-1]$ given by 
	\bean
	{\cal A}^c & = & [0:n-1] \setminus {\cal A}. 
	\eean
	An $m\times n$ matrix $A$ is said to be {\it Cauchy-like}, if every square submatrix of $A$ is non-singular. The dual of an $[n,k]$ code $\mathcal{C}$ over $\mathbb{F}_q$, will be denoted by $\mathcal{C}^\perp$. 
	
	\subsection{Preliminaries}
	%=======================
	\begin{lem}[Combinatorial Nullstellensatz \cite{AlonCombNul99}] \label{lem:comb_null}
		\normalfont Consider a non-zero multivariate polynomial 
		\bean 
		f(x_1,x_2,\ldots,x_m)\in \mathbb{F}_q[x_1,x_2,\ldots,x_m]. 
		\eean 
		Let the degree of the polynomial in the variable $x_i$ be  $d_i$, for $1\leq i\leq m$. If $|\mathbb{F}_q|>d_i$ for all $i\in [1:m]$, then there exists $(s_1, s_2,\ldots,s_m) \in \mathbb{F}^m_q$ such that  $f(s_1,s_2,\ldots,s_m)\neq0$.
	\end{lem}

	\begin{defn}[Punctured Codes]\normalfont
		Let $\mathcal{C}$ be an $[n,k]$ linear code over $\mathbb{F}_q$. Given a subset \calp\ of $[0:n-1]$, the code $\mathcal{C}$ {\it punctured on} the coordinates in \calp, is the linear code of length $|\calpc | = (n-|\mathcal{P}|)$ obtained from $\mathcal{C}$ by deleting all the coordinates in $\mathcal{P}$.   Equivalently, the code $\mathcal{C}$ punctured on the coordinates in $\mathcal{P}$ is the restriction $\mathcal{C}|_{\calpc}$ of \calc\ to the coordinates in \calpc.  The punctured code (or restriction) $\mathcal{C}|_{\calpc}$ will also be referred to as $\mathcal{C}$ {\it punctured to} the coordinates in \calpc.
	\end{defn} 
	\begin{defn}[Shortened Codes]\normalfont
		Let $\mathcal{C}$ be an $[n,k]$ linear code over $\mathbb{F}_q$.  Given a subset $\mathcal{P}$ of $[0:n-1]$, consider first the subcode $\mathcal{C}^*$ given by: 
		\begin{equation*}
		\mathcal{C}^*=\{\underline{c}=(c_0\ c_1\ \ldots\ c_{n-1})\in \mathcal{C}: c_i=0\ \forall i\in \mathcal{P}\}.
		\end{equation*}
		Then by the phrase $\mathcal{C}$ {\it shortened on} the coordinates in $\mathcal{P}$, denoted by $\mathcal{C}^{\calpc}$, we will mean the linear code of length $(n-|\mathcal{P}|)$ obtained from $\mathcal{C}^*$ after puncturing on the coordinates given by $\mathcal{P}$.   The code $\mathcal{C}^{\calpc}$ will also be referred to as the code $\mathcal{C}$ {\it shortened to} the coordinates in \calpc.
	\end{defn}

	%====================
	
	\begin{lem}[{\cite[p.~17]{HufPle}}]\label{lem:punc_short_duality}\normalfont
		Let $\mathcal{P}\subseteq [0:n-1]$. Then \bean (\mathcal{C}|_{\mathcal{P}})^\perp=(\mathcal{C}^\perp)^\mathcal{P}. \eean 
	\end{lem}
	
	An $m\times n$ matrix $A$ over a finite field \fq, with $m\leq n$, will be referred to as an {\it MDS matrix} if any  $m$ distinct columns of $A$ form a linearly independent set. Clearly, an MDS matrix $A$ can serve as the generator matrix of an $[n,m]$ MDS code.
	
%	\subsection{Properties of a Zero-Band Generator Matrix for an MDS Code}

	\begin{lem}\label{lem:shortened_mds}\normalfont
		Let $\mathcal{C}_\text{\tiny MDS}$ denote an $[n,k]$ MDS code.   For $1\leq l\leq k$, consider an $(l\times n)$ matrix $P$ whose rows $\{\underline{c}_i\}_{i=0}^{l-1}$ correspond to a set of linearly independent codewords drawn from $\mathcal{C}_\text{\tiny MDS}$.   Then if $|\cup_{i=0}^{l-1}\text{supp}(\underline{c}_i)|=(n-k+l)$, any choice of $\leq l$ non-zero columns of $P$ forms a linearly independent set.
	\end{lem}
	\begin{proof}
		The basic idea here is to show that if a collection of $l$ independent codewords drawn from an $[n,k]$ MDS code share $(k-l)$ zeros in common, then these $l$ codewords comprise a shortened MDS code (after the $(k-l)$ coordinates corresponding to the common zeros are deleted). 
		More formally, it is known that shortening an $[n,k]$ MDS code on a set $\mathcal{S}$ of coordinates, where $\mathcal{S}\subseteq [0:n-1]$, $|\mathcal{S}|=s$, $0\leq s\leq (k-1)$, results in an $[n-s,k-s]$ MDS code.  Let $\mathcal{A} \triangleq \cup_{i=0}^{l-1}\text{supp}(\underline{c}_i)$. Clearly $\text{span}<\underline{c}_1,\underline{c}_1,\cdots, \underline{c}_l >$, after removing the $(k-l)$ trivial zero coordinates corresponding to $\mathcal{A}^c$, is a subspace of the $[(n-k+l),l]$ MDS code $\mathcal{C}_\text{\tiny MDS}^{{\cal A}}$ obtained by shortening $\mathcal{C}_\text{\tiny MDS}$ to $\mathcal{A}$. As $\text{rank}(P)=l$, the matrix $P$ (after removing the zero columns $\mathcal{A}^c$) is indeed a generator matrix for the shortened MDS code $\mathcal{C}_\text{\tiny MDS}^{{\cal A}}$ of dimension $l$. The lemma then follows.
	\end{proof}
	
	\begin{defn}{(Zero-band generator matrix of an MDS code)}
		\normalfont Consider an $[n,k]$ MDS code $\mathcal{C}_\text{\tiny MDS}$ over $\mathbb{F}_q$. A zero-band generator matrix (ZB generator matrix), say $Z$, corresponding to $\mathcal{C}_\text{\tiny MDS}$ is a $(k\times n)$ generator matrix of $\mathcal{C}_\text{\tiny MDS}$ that contains a diagonal band of $(k-1)$ consecutive zeros as shown below: 
		\begin{equation*}
		Z=\left[
		%\mathbf{Z}=\left[ 
		{\begin{array}{ccccccccccccccc}
			* & 0 & 0 & \ldots &\ldots & 0 & * & * & * & \ldots & * & * &\ldots&*\\
			* & * & 0 & \ldots &\ldots &  0 & 0 & * & * & \ldots & * & *&\ldots&*\\
			* & * & * & 0 & \ldots & 0 & 0 & 0 & * & \ldots & * & *&\ldots&* \\ 
			&  &  & \ddots &  &  &  & &   & \ddots &  & &\ldots&*\\
			* & * & * & \ldots & * & 0 & 0 & 0 & \ldots & 0 & * & *&\ldots&*\\
			* & * & * & \ldots & * & * & 0 & 0 & \ldots & 0 & 0 & * &\ldots&*
			\end{array}}\right].
		\end{equation*}
		More precisely, the $i$-th row of $Z$, $\underline{z}_{\text{\tiny row},i}$ for $0\leq i\leq (k-1)$,  has a run of $(k-1)$ zeros spanning the coordinates $[i+1:i+k-1]\pmod n$. Here each $*$ indicates a non-zero element in $\mathbb{F}_q$.
	\end{defn}
	
	\begin{lem}\normalfont\label{lem:zbgen_existence}
		Given an $[n,k]$ MDS code there always exists a corresponding ZB generator matrix $Z$.
	\end{lem}
	\begin{proof}
		Let us choose the $i$-th row of $Z$ to be the non-zero codeword $\underline{c}_i\triangleq(c_{i,0}\ c_{i,1}\ \ldots \ c_{i,n-1})\in \mathcal{C}_\text{\tiny MDS}$ such that $c_{i,j}=0$ for $(i+1)\leq j\leq (i+k-1)\pmod n$.  This can always be done since a codeword in an MDS code is uniquely specified by any set of $k$ coordinates and $(k-1)$ of them can be chosen to be zeros.  Note that all the remaining $n-(k-1)$ coordinates of $\underline{c}_i$ are forced to be non-zero as $\mathcal{C}_\text{\tiny MDS}$ has minimum-distance and hence minimum Hamming weight, equal to $(n-k+1)$.  The first $k$ columns of $Z$ then form a lower triangular matrix with non-zero entries along the diagonal.  Hence $\text{rank}(Z)=k$ and $Z$ is a generator matrix for the MDS code.
		%
		%
		%For $0\leq i\leq (k-1)$, $\underline{z}_{\text{\tiny row},i}$ is chosen as a non-zero codeword $\underline{c}_i\triangleq(c_{i,1}\ c_{i,2}\ \ldots \ c_{i,n})\in \mathcal{C}_\text{\tiny MDS}$ such that $c_{i,j}=0$ for $(i+1)\leq j\leq (i+k-1)\mod n$. Such a codeword always exists since the dimension of $\mathcal{C}_\text{\tiny MDS}$ is $k$ and any choice of $(k-1)$ coordinates can be specified to have $0$'s to arrive at a non-zero codeword. Note that all the remaining $n-(k-1)$ coordinates of $\underline{c}_i$ are forced to be non-zero as $\mathcal{C}_\text{\tiny MDS}$ has a minimum-distance of $(n-k+1)$. The first $k$ columns of $Z$ form a lower triangular matrix with non-zero entries  all along the diagonal and hence $\text{rank}(Z)=k$. Therefore $Z$ is a generator matrix.
	\end{proof}
	\begin{lem}\normalfont\label{lem:mds_consec_rows}
		Let $Z$ be a ZB generator matrix corresponding to an $[n,k]$ MDS code $\mathcal{C}_\text{\tiny MDS}$. Fix $i,j$ such that $0\leq i\leq (k-1)$ and $1\leq j\leq(k-i)$. Any choice of $\leq j$ non-zero columns of the $j\times n$ matrix $Z(i:i+j-1,:)$ forms a linearly independent set.
	\end{lem}
	\begin{proof}  The result follows from Lemma \ref{lem:shortened_mds} by noting that the $j$ consecutive rows are linearly independent and that they have common zeros of size $(k-j)$ on the set of coordinates: $[i+j:i+k-1]\mod n$. 
	\end{proof}
	
	\begin{lem}\label{lem:recovery_parity_col_span}\normalfont
		Consider an $[n,k]$ code $\mathcal{C}$ having p-c matrix $H$. Let the coordinates indexed by $\mathcal{E}\subseteq [0:n-1]$ be erased from $\mathcal{C}$ and let $i\in \mathcal{E}$. Then the $i$-th code symbol in a codeword can be recovered from the code symbols of the same codeword corresponding to the coordinates in \calec\ iff: 
		\bean
		\underline{h}_i & \notin & \text{span} \left \langle \ \{\underline{h}_{j}\}_{j\in \mathcal{E}\setminus\{i\}} \ \right \rangle,
		\eean
		where $\underline{h}_{j}$ denotes the $j$-th column of $H$.
	\end{lem}
	\begin{proof}
		We skip the proof as it is a well-known result.	
	\end{proof}
	%===============SECTION III Sliding-Window Channel Model and a Coding Framework==================
	%===============SECTION III Sliding-Window Channel Model and a Coding Framework==================
	%===============SECTION III Sliding-Window Channel Model and a Coding Framework==================
	
	\section{A Coding Framework for Streaming Codes and the Sliding-Window Channel Model} \label{sec:setting}
	
\subsection{A Coding Framework for Streaming Codes}

In this paper, we follow the framework introduced by Martinian and Sundberg \cite{MartSunTIT04}. Let $k,n$ be integers such that $k< n$. The encoder $\mathbf{E}$ receives a {\it message packet} $\underline{s}(t) \in \mathbb{F}_q^{k}$, 
\bean
\underline{s}(t) & \triangleq & [s_0(t)\ s_1(t)\ \ldots\ s_{k-1}(t)]^T, \ \ t\in\{0,1,2,\ldots\},  
\eean
at time-$t$. The causal encoder $\mathbf{E}$ produces a {\it coded packet}: 
\bean
\underline{x}(t)& \triangleq & \left[ \begin{array}{c} \underline{s}(t) \\ \underline{p}(t) \end{array} \right] 
\eean at time-$t$, where the {\it parity packet} $\underline{p}(t)\triangleq[p_0(t)\ p_1(t)\ \ldots\ p_{n-k-1}(t)]^T \in \mathbb{F}_q^{n-k}$ is a function of message packets received till time-$t$, i.e., $\{\underline{s}(0), \underline{s}(1),\ldots,\underline{s}(t)\}$. Between the encoder $\mathbf{E}$ and the decoder $\mathbf{D}$, there exists a channel which erases some of the transmitted coded packets. Let $\underline{y}(t)$ denote the received packet at time-$t$. We have:
\bean
\underline{y}(t) & = & \left\{ \begin{array}{cl} *, & \text{if $\underline{x}(t)$ is erased}, \\
	\underline{x}(t), & \text{otherwise}. \end{array} \right.
\eean
The delay-constrained decoder $\mathbf{D}$ with delay-parameter $\tau$ outputs the decoded message packet $\hat{\underline{s}}(t)$ by time $(t+\tau)$. Here $\hat{\underline{s}}(t)$, which is an estimate of the message packet $\underline{s}(t)$, is a function of received packets till time $(t+\tau)$, i.e., $\{\underline{y}(0),\underline{y}(1),\ldots,\underline{y}(t+\tau)\}$.  In an ideal scenario, we would have $\hat{\underline{s}}(t)=\underline{s}(t)$. As $k$-length message packets are mapped to $n$-length coded packets at each time-$t$, rate of the code $R$ is defined as $\frac{k}{n}$.

%	Clearly, an encoding fraemwork such as this that oeprates with fiAn encoding scheme introduced in \cite{MartSunTIT04} falls within the framework of convolutional codes.
%{\pvkred the framework is convolutional if it is linear time invariant with finite memory and governed by  a convolutional I/O relationship, so not yet ready to declare this as being convolutional.}

\subsection{Erasure-Channel Models}

In \cite{MartTrotISIT07}, the authors propose a family of streaming codes that can tolerate a burst erasure of length $b$ with delay $\tau$. i.e., the code permits recovery of message packet $\underline{s}(t)$ by time $(t+\tau)$, even when $\underline{y}(l)=*$ for $l\in[j:(j+b-1)]$, where $t\in[j:(j+b-1)], j\in \{0,1,2,\ldots\}$.  At first glance this model might appear restricted to handling just a single erasure burst of length $b$ over all time. However because of the delay, this model forces the decoder to tolerate any number of bursts, as long as they are spaced apart by at least $\tau$ time units. 
%	For instance, the code can recover all the message packets with delay $t$, when coded packets indexed by coordinates $\cup_{j=\{0,1,\ldots\}}[j(t+b):j(t+b)+b-1]$ are erased. 
Badr et al. \cite{BadrPatilKhistiTIT17} extend the burst-loss based channel model introduced in \cite{MartSunTIT04} to include random erasures as well. The paper \cite{BadrPatilKhistiTIT17} introduces a sliding-window (SW) channel model  with parameters $a$, $b$, $w$ in addition to the decoding delay parameter $\tau$. The model is as follows. Given any sliding window of width $w$, the channel introduces at most one of the following patterns of erasures (i) a burst erasure of length $\leq b$ (ii) a total of $\leq  a$ random erasures (see Fig. \ref{fig:s-w_channel_model}).   As explained below, the parameters are subject to certain constraints.  We must have: 
\ben
\item $a\leq b$, since if $a>b$, the burst-error requirement would be subsumed by the random-erasure requirement and rendered redundant,
\item $b\leq \tau$ in order to have non-zero rate when operating with a causal encoder, 
\item $b<w$ to avoid admitting within the model, a never-ending continuous stream of erasures.
\een
These constraints can be summarized in the form:
\bea 
a \leq b \leq \min \{\tau,w-1\}.  \label{eq:constraints} 
\eea
We will refer to $\{a,b,\tau,w\}$ as {\em the parameter set} of the Delay-Constrained SW (DC-SW) channel model \footnote{The terminology ``DC-SW channel model'' clearly is an abuse of notation since the delay constraint is not part of the channel model, but rather a constraint placed on the decoder.   We employ this terminology for the sake of convenience.  This allows us to refer for example, to the set $\{a,b,\tau,w\}$ simply as the parameter set of the DC-SW channel.}.  We will assume throughout the remainder of this paper that the parameter set satisfies the constraints laid out in \eqref{eq:constraints}. 
\begin{figure}[!htb]
	\centering
	\captionsetup{justification=centering}
	\includegraphics[scale=0.45]{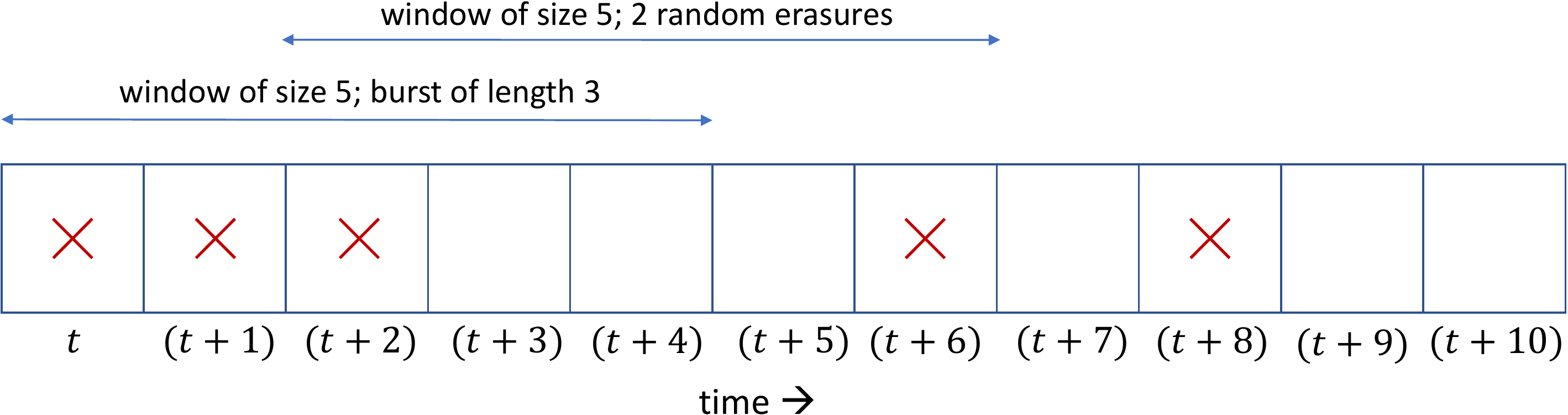}
	\caption{An example channel realization under the SW channel model, for parameters $a=2, b=3, w=5$. Here each $\color{red} \times$ indicates an erasure.}
	\label{fig:s-w_channel_model}
\end{figure}

Set $\delta\triangleq(b-a)$ and define the {\em effective time delay} parameter $\tau_\text{eff}\triangleq \min\{\tau,w-1\}$. The rate $R$ of a {\em streaming code} \cstr\ which can faithfully recover all data with a delay not exceeding $\tau$, from all of the erasure patterns permitted under the DC-SW channel model, was shown in \cite{BadrPatilKhistiTIT17} to be upper bounded as below:
%\begin{equation}\label{eq:rate_upper_bound}
%	R\leq \frac{(t_\text{eff}-a+1)}{(t_\text{eff}+\delta+1)}.
%\end{equation}

\begin{equation}\label{eq:rate_upper_bound}
R\leq \frac{(\tau_\text{eff}-a+1)}{(\tau_\text{eff}+\delta+1)}.
\end{equation}
where $\tau_\text{eff} = \min\{\tau,w-1\}$. We give a short proof for \eqref{eq:rate_upper_bound} as follows.

\bpf
\bit
\item[]
\item[]
Consider the DC-SW channel with parameters $\{a,b,\tau,w\}$.
\item(\textit{Case I:} $w\geq(\tau+1)$) Let \cstr\ be a code which can tolerate all the erasure patterns of the DC-SW channel. Consider a periodic erasure channel  with a period of $\nu = (\tau+\delta +1)$, as shown in Figure \ref{fig:periodic_erasure_channel}.  We show that this code can also tolerate the erasure patterns occurring in the periodic erasure channel. 
\ben
\item The first $(b-a+1)$ packets can be recovered within a delay of $\tau$, as every packet in this set encounters a burst of length at most $b$.
\item Consider the starting index as $0$ in Figure \ref{fig:periodic_erasure_channel}. A packet indexed $\ell$ in $[b-a+1:b-1]$ encounters the erasure pattern $[\ell:b-1]\cup [\delta+1+\tau:\ell+\tau]$ within the decoding window of $[\ell:\ell+\tau]$. The length of this erasure pattern is $a$, irrespective of $\ell$. 
\een
Hence the code can recover from the first erasure burst spanning coordinates $[0:b-1]$ with a delay of $\tau$. This observation can be easily extended to show  recoverability from the other erasure bursts as well. So, the rate of \cstr\ must be less than the capacity of the periodic erasure channel: 
\begin{equation}\label{eq:rate_upper_bound_1}
R \le \frac{\tau-a+1}{\tau+\delta+1}.
\end{equation} 
\item (\textit{Case II:} $w<(\tau+1)$) Consider a periodic erasure channel with a period $\nu = w+\delta$, where the first $b$ packets at the beginning of each period are erased, similar to the previous case. The erasure pattern as set by this periodic erasure channel is a permitted pattern for the SW channel with parameters $\{a,b,w\}$. Thus we have: 
\begin{equation}\label{eq:rate_upper_bound_2}
R \le \frac{w-a}{w+\delta}.
\end{equation} 
\eit
Combining equations \ref{eq:rate_upper_bound_1} and \ref{eq:rate_upper_bound_2}, we obtain \eqref{eq:rate_upper_bound}.
\epf
\begin{figure}[!htb]
	\centering
	\captionsetup{justification=centering}
	\includegraphics[scale=0.7]{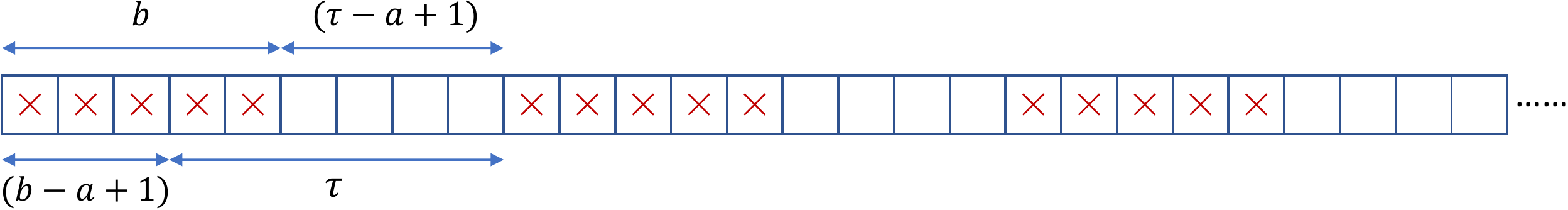}
	\caption{A periodic erasure channel with period $(\tau+\delta +1)$ for parameter set $\{a=3,b=5,\tau=6,w=9\}$. Here each $\color{red} \times$ indicates an erasure. }
	\label{fig:periodic_erasure_channel}
\end{figure}

\begin{defn}[Rate-Optimal Streaming Code] \normalfont Given a DC-SW channel model with parameter set $\{a,b,\tau,w\}$, a streaming code \cstr\ is said to be {\it rate-optimal} if the rate of the code meets the upper bound in \eqref{eq:rate_upper_bound} with equality and \cstr\ permits recovery from all the erasure patterns as set by the DC-SW model.  
\end{defn} 

\subsection{Relative Sizes of Window Length and Delay Parameters}

We show here that one can without loss of generality assume that: 
\bea
\tau_\text{eff} \ = \ \tau \ = \ (w-1). \label{eq:teqwm1}
\eea
In order to see this, let \cstr\ be a streaming code that handles all the erasure patterns permitted by the DC-SW having parameter set $\{a,b,w,\tau\}$, in which $w>(\tau+1)$.  It is not hard to see that \cstr\ can handle all the erasure patterns permitted by a second DC-SW model where the parameters $a,b,\tau$ are unchanged, but where the window-size parameter $w$ has the smaller value $w=(\tau+1)$. As reducing the window size represents in general, a more stringent constraint, we can assume without loss of generality, that: 
\bea
w & \leq & (\tau+1), \text{  or equivalently} \\
\tau & \geq & (w-1).
\label{eq:wt1} 
\eea
%	\ \ \ \text{ leading to }  \\ 
%	 \min\{t,w-1\} = t_{\text{eff}}  & = &  w-1.
Next, consider the situation when in the DC-SW channel, for some parameter set $\{a,b,\tau,w\}$, we have that $\tau>(w-1)$.  It turns out that in this case, we can construct for any given parameter set, rate-optimal codes having finite field size $q$ that is $O(\tau^2)$ and which incur no rate penalty when $\tau$ is reduced to equal $(w-1)$.  This code is constructed in Section \ref{sec:construction_A}.   Thus regarding $O(\tau^2)$ as an acceptable field-size, which is reasonable given that $\tau$ is typically small in the applications envisaged here, then we can without loss in performance, set $\tau=w-1$, as stated in \eqref{eq:teqwm1}.   We should note however, that of the four constructions presented in this paper, apart from Construction A which has a $O(\tau^2)$ field-size requirement, all the others have a field-size requirement that is $O(\tau)$, while maintaining $\tau = w-1$.   We will from here on assume that \eqref{eq:teqwm1} holds and hence will replace $\tau_{\text{eff}}$ with $\tau$.   Thus the parameter sets of the DC-SW channel model that we will focus on from here onward are of the form $\{a,b,\tau,w=\tau+1\}$.
	We note that the authors of \cite{BadrPatilKhistiTIT17} also set $w\geq(\tau+1)$, although employing a slightly different argument.
%	
%	Hence from here onwards we always assume that the delay parameter $t$ and window-size parameter $w$ are related as: $w=(t+1)$ and the subscript will be dropped from $t_\text{eff}$ ($t_\text{eff}$ will always be $t$ when $w=(t+1)$).  	

	\subsection{An Equivalent Set of Conditions for Erasure Recovery}\label{sec:streaming_equivalent_conditions}
	Given a DC-SW channel with parameter set $\{a,b,\tau,w\}$ and $w=(\tau+1)$, it can be easily shown that a streaming code \cstr\  recovers from all erasure patterns permissible under the DC-SW model iff the following conditions are true, for any $t\in \{0,1,2,\ldots\}$: 
	\begin{enumerate}
		\item[J1.] [{\em Burst Erasure Requirement}] \cstr\ can guarantee recovery of $\underline{x}(t)$ with a delay of at most $\tau$ in the presence of an erasure-burst of length at most $b$, involving coded packet $\underline{x}(t)$,  
		\item[J2.] [{\em Random Erasure Requirement]} \cstr\ can guarantee recovery of $\underline{x}(t)$ with a delay of at most $\tau$ in the presence of t most `$a$' erasures involving coded packet $\underline{x}(t)$ and $\leq(a-1)$ other coded packets.
		\een
		
		\subsection{Convolutional Codes Derived from the Diagonal Embedding of a Block Code}\label{subsec:Diag_Embedding}
		
		The particular streaming code, which is a convolutional code, employed in the Martinian-Trott~\cite{MartTrotISIT07}  scheme is constructed by embedding a block code in diagonal fashion (see Fig. \ref{fig:diag_interleaving_example}).   The same scheme has since been employed in \cite{RobustStreamingCodes,FongKhisti,NikPVK} as it reduces the problem of designing a streaming code to that of carefully designing a block code that satisfies multiple constraints.  We adopt the same approach in the present paper as well. Formally, the diagonal embedding scheme can be described as follows.  
		
		Let $\mathcal{C}$ be an $[n,k]$ code having a $(k\times n)$ systematic generator matrix $G=[I_k\ P]$ and set $r\triangleq (n-k)$. The parity symbols of the resultant systematic convolutional code after diagonally embedding $\mathcal{C}$ are given by:
		\begin{equation}
		\big[p_0(t)\ p_1(t+1)\ \ldots\ p_{r-1}(t+r-1)\big]=\big[s_{0}(t-k)\ s_{1}(t-k+1)\ \ldots\  s_{k-1}(t-1)\big]P,
		\end{equation}
		for $t\in \mathbb{Z}$. Set $s_j(t)\triangleq0$ when $t<0$, for $0\leq j\leq (k-1)$. We make the diagonal embedding technique explicit with the help of an example illustrated in Fig. \ref{fig:diag_interleaving_example}, where $[n=6,k=3]$ and the submatrix $P$ is given by:
		\begin{equation*}
		P=\left[
		{\begin{array}{ccc}
			p_{0,0} & p_{1,0} & p_{2,0} \\
			p_{0,1} & p_{1,1} & p_{2,1} \\
			p_{0,2} & p_{1,2} & p_{2,2}
			\end{array}}\right].
		\end{equation*} 
		In order to bring out the convolutional nature of this encoder, we represent the encoding process in a more traditional manner using shift registers in Fig. \ref{fig:diag_interleaving_example_2}.
		\begin{figure}[!htb]
			\centering
			\captionsetup{justification=centering}
			\includegraphics[scale=0.6]{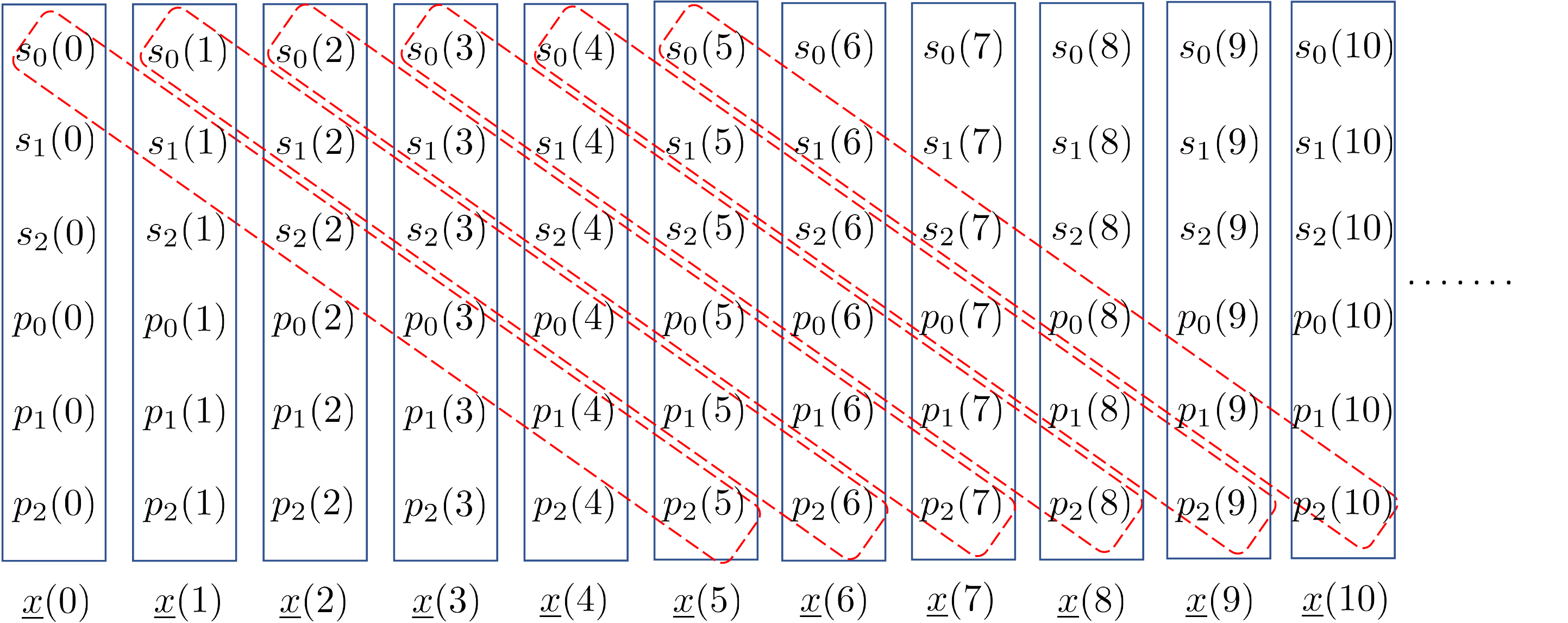}
			\caption{The streaming code \cstr\ obtained by diagonally embedding a $[6,3]$ systematic block code $\mathcal{C}$. Each diagonal of the form $[s_0(t)\ s_1(t+1)\ s_2(t+2)\ p_0(t+3)\ p_1(t+4)\ p_2(t+5)]$ is a codeword in $\mathcal{C}$, where $t\in\mathbb{Z}$. The symbols $\{s_i(t)\}$ are raw message symbols belonging to the message packet $\underline{s}(t)$, whereas $\{p_i(t)\}$ are parity symbols.}
			\label{fig:diag_interleaving_example}
		\end{figure}

		\begin{figure}[!htb]
			\centering
			\captionsetup{justification=centering}
			\includegraphics[scale=0.6]{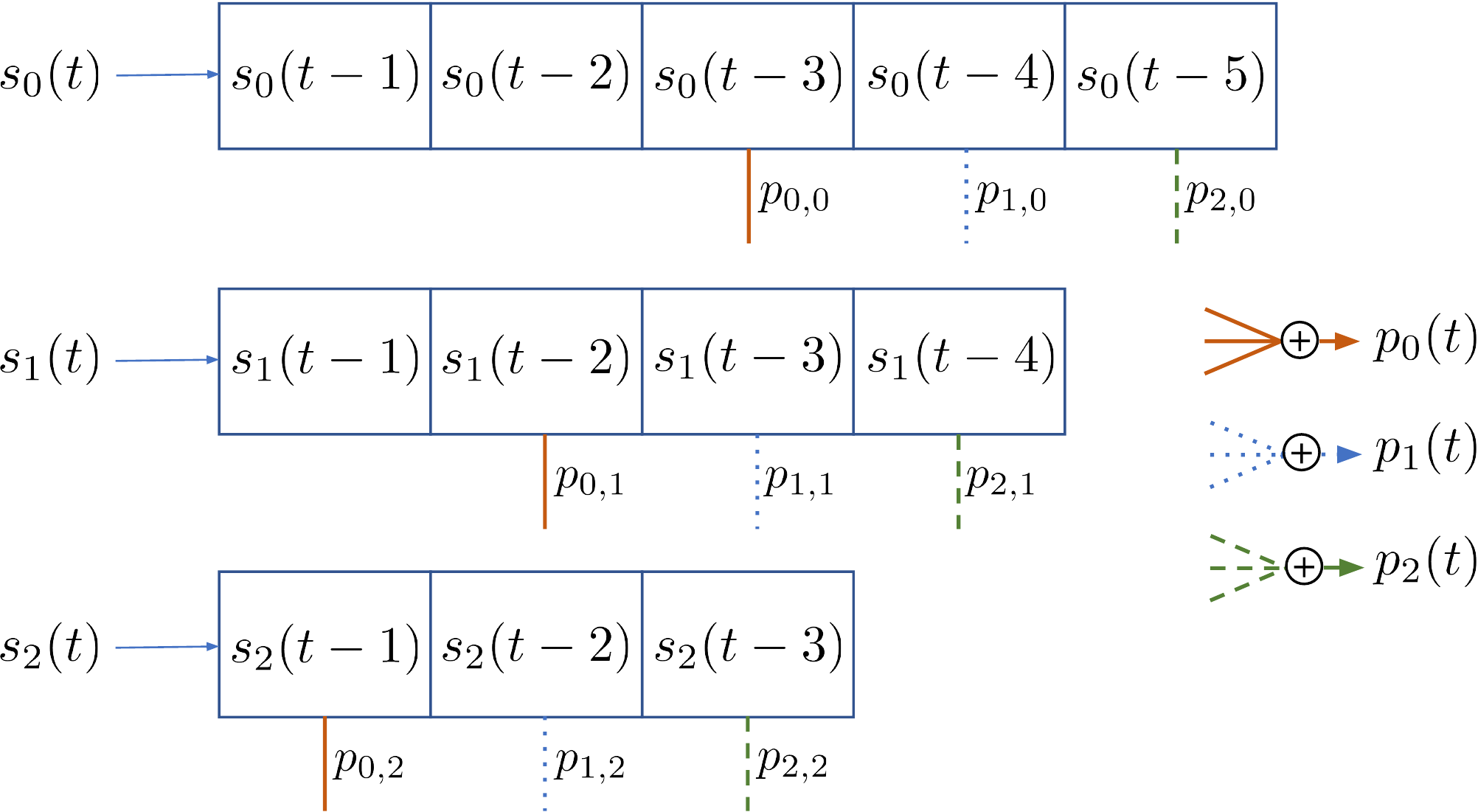}
			\caption{An alternate representation of the streaming code \cstr\ obtained by diagonally embedding a $[6,3]$ block code $\mathcal{C}$ (see Fig. \ref{fig:diag_interleaving_example}) using shift registers, to more clearly bring out its convolutional nature. Here $p_0(t)=p_{0,0}s_0(t-3)+p_{0,1}s_1(t-2)+p_{0,2}s_2(t-1)$ and so on.}
			\label{fig:diag_interleaving_example_2}
		\end{figure}
		
In this paper, we follow the approach of building \cstr\ via diagonal embedding a block code $\mathcal{C}$. Thus our aim is to construct, using the diagonal embedding technique, a streaming code \cstr\ that can recover from any erasure pattern permissible under the DC-SW channel model. 
% , i.e.,  any packet $\underline{x}[i]$ can recover from erasure with a delay of at most $t$, when there is either (i) an erasure-burst of  length at most $b$ affecting the coordinate $i$ or (ii) there are at most $a$ erasures affecting $i$ and $\leq (a-1)$ other randomly placed coordinates, translates to certain requirements on the block code. 
In the following, we discuss the requirements imposed on the block code $\mathcal{C}$.

%\newpage

\begin{table}  
	\begin{center}
		\begin{tabular}{||c|c|c|c|c||} \hhline{#=|=|=|=|=#}%\hline \hline 
			Symbols to be recovered & Past known symbols & Erasure pattern & Symbols available & Unavailable symbols \\ 
			& & & (to the decoder) & (due to delay constraint) \\  \hhline{#=|=|=|=|=#}%\hline \hline 
			$s_0(5)$ & - & $s_0(5),s_1(6),s_2(7)$ & $p_0(8)$ & $p_1(9),p_2(10)$ 
			\\ \hline 
			$s_1(5)$ & $s_0(4)$  & $s_1(5),s_2(6),p_0(7)$ & $s_0(4),p_1(8)$ & $p_2(9)$ 
			\\ \hline 
			$s_2(5)$ & $s_0(3),s_1(4)$ & $s_2(5),p_0(6),p_1(7)$ & $s_0(3),s_1(4),p_2(8)$ & - 
			\\ \hline 
			$p_0(5)$ & $s_0(2),s_1(3),s_2(4)$ & $p_0(5),p_1(6),p_2(7) $ & $s_0(2),s_1(3),s_2(4)$ & - 
			\\ \hline  
			$p_1(5)$ & $s_0(1),s_1(2),s_2(3),p_0(4)$ & $p_1(5),p_2(6) $ & $s_0(1),s_1(2),s_2(3),p_0(4)$ & - 
			\\ \hline 
			$p_2(5)$ & $s_0(0),s_1(1),s_2(2),p_0(3),p_1(4)$ & $p_2(5)$ & $s_0(0),s_1(1),s_2(2),p_0(3),p_1(4)$ & - 
			\\  \hhline{#=|=|=|=|=#} % \hline  \hline 
			%  $s_2(5),p_0(6),p_1(7),p_2(8)$ & a burst of $3$ erasures,$s_0(3),s_1(5)$ & $s_2(5),p_0(6),p_1(7)$ & $p_2(8)$ & -  
			% \\ \hline 
		\end{tabular}
	\end{center}
    \caption{Assume received packets $\underline{y}(5),\underline{y}(6),\underline{y}(7)$ are erased in Fig.~\ref{fig:diag_interleaving_example}, i.e., a burst erasure of length $3$.  Here we state requirements to recover $\underline{x}(5)$ with a delay of at most $\tau=3$, in terms of the symbols appearing in the streaming code \cstr.}
	\label{tab:c_str_rec_requirements}
\end{table}

\begin{table} 
	\begin{center}
		\begin{tabular}{||c|c|c|c|c||} \hhline{#=|=|=|=|=#}%\hline \hline 
			Symbols to be recovered & Past known symbols & Erasure pattern & Symbols available & Unavailable symbols \\ 
			& & & (to the decoder) & (due to delay constraint) \\ \hhline{#=|=|=|=|=#}%\hline \hline 
			$c_0$ & - & $c_0,c_1,c_2$ & $c_3$ & $c_4,c_5$ 
			\\ \hline 
			$c_1$ & $c_0$  & $c_1,c_2,c_3$ & $c_0,c_4$ & $c_5$ 
			\\ \hline 
			$c_2$ & $c_0,c_1$ & $c_2,c_3,c_4$ & $c_0,c_1,c_5$ & - 
			\\ \hline 
			$c_3$ & $c_0,c_1,c_2$ & $c_3,c_4,c_5$ & $c_0,c_1,c_2$ & - 
			\\ \hline  
			$c_4$ & $c_0,c_1,c_2,c_3$ & $c_4,c_5$ & $c_0,c_1,c_2,c_3$ & - 
			\\ \hline 
			$c_5$ & $c_0,c_1,c_2,c_3,c_4$ & $c_5$ & $c_0,c_1,c_2,c_3,c_4$ & - 
			\\\hhline{#=|=|=|=|=#}% \hline  \hline 
			%  $s_2(5),p_0(6),p_1(7),p_2(8)$ & a burst of $3$ erasures,$s_0(3),s_1(5)$ & $s_2(5),p_0(6),p_1(7)$ & $p_2(8)$ & -  
			% \\ \hline 
		\end{tabular}
	\end{center}
    \caption{The requirements imposed on the code symbols of the diagonally-embedded block code ${\cal C}$ used to build the streaming code \cstr.}
	\label{tab:diag_embedding_rec_requirements}
\end{table}

%\newpage		
		
		\begin{example}\normalfont
			Consider the parameter set $\{a=1,b=3, \tau=3,w=(\tau+1)=4\}$. Our interest is in constructing a code \cstr\ which can recover from all the erasure patterns permitted by the DC-SW channel model given these values of the parameters. Consider an $[n=6,k=3]$ block code $\mathcal{C}$  which is diagonally embedded to obtain \cstr. Note that the rate $R$ of the streaming code \cstr\ is always equal to the rate of the block code used for diagonal embedding. In the present case, the rate $R=0.5$, which meets the \eqref{eq:rate_upper_bound} since here, $\delta=(b-a)=1$ and hence: 
			\bean
			R \leq \frac{(\tau_\text{eff}-a+1)}{(\tau_\text{eff}+\delta+1)} \ = \ \frac{3}{6} \ = \ \frac{1}{2}.
			\eean
			In Fig. \ref{fig:diag_interleaving_example} which corresponds to diagonally embedding a $[6,3]$ code, assume coded packets  $\underline{x}(5),\underline{x}(6), \underline{x}(7)$ are erased, i.e., a burst erasure of length $b=3$ starting at time $5$. Consider the recovery of  $\underline{x}(5)$ with a delay of at most $\tau=3$. In Table \ref{tab:c_str_rec_requirements}, we list down various requirements to recover $\underline{x}(5)$, in terms of the symbols appearing in the streaming code $\cstr$ formed via diagonally embedding $\mathcal{C}$. Let $\underline{c}^T=(c_0\ c_1\ \ldots\ c_5)$ denote an arbitrary codeword in $\mathcal{C}$. In Table \ref{tab:diag_embedding_rec_requirements}, we translate the requirements  presented in Table \ref{tab:c_str_rec_requirements} to that on the code symbols $\{c_i\}_{i=0}^5$. To summarize, we have the condition on $\mathcal{C}$ that for all codewords $\underline{c}^T\in \mathcal{C}$, $c_i$ must be recoverable from $\{c_j:j\in[0:i-1]\cup [i+3:\min\{5,i+3\}]\}$. Note that here we considered the case where coded packet is erased as part of a burst erasure (i.e., condition J1 in Section \ref{sec:streaming_equivalent_conditions}). An analogous condition can be imposed on the block code $\mathcal{C}$ even for the case of random erasures (corresponds to condition J2 in Section \ref{sec:streaming_equivalent_conditions}). We formally summarize in the following subsection the requirements on $\mathcal{C}$ so that the streaming code \cstr\ ( formed via diagonally embedding $\mathcal{C}$) can recover any coded packet $\underline{x}(t)$  from a burst erasure of length  $b$ or $a$ random erasures, with a delay constraint of $\tau$.
		
%			\begin{figure}[!htb]
%				\centering
%				\captionsetup{justification=centering}
%				\includegraphics[scale=0.6]{fig_diagonal_interleaving_example_3_v2}
%				\caption{The list of all the $6$ codewords of block code $\mathcal{C}$ that intersect with the coded packet $\underline{x}[5]^\mathbf{t}$. }
%				\label{fig:diag_interleaving_example_3}
%			\end{figure}
%			\begin{table}[ht]
%				\centering
%				\begin{tabular}{ |c|c|c| } 
%					
%					\hline
%					Layer & Symbol to be & Symbols used 
%					\\ & recovered& in recovery\\  
%					
%					%	& & &(real additions+comparisons) & & \\		
%					\hline \hline
%					0& $c^{(5)}_{0,5}$ & $c^{(5)}_{0,0},c^{(5)}_{0,1},c^{(5)}_{0,2},c^{(5)}_{0,3},c^{(5)}_{0,4}$\\ 
%					\hline
%					1&$c^{(5)}_{1,4}$ & $c^{(5)}_{1,0},c^{(5)}_{1,1},c^{(5)}_{1,2},c^{(5)}_{1,3}$\\ 
%					\hline
%					2&$c^{(5)}_{2,3}$ & $c^{(5)}_{2,0},c^{(5)}_{2,1},c^{(5)}_{2,2}$\\ 
%					\hline
%					3&$c^{(5)}_{3,2}$ & $c^{(5)}_{3,0},c^{(5)}_{3,1},c^{(5)}_{3,5}$\\
%					\hline
%					4&$c^{(5)}_{4,1}$ & $c^{(5)}_{4,0},c^{(5)}_{4,4},c^{(5)}_{4,5}$\\
%					\hline
%					5&$c^{(5)}_{5,0}$ & $c^{(5)}_{5,3},c^{(5)}_{5,4}$\\
%					\hline
%				\end{tabular} 
%				\caption{A table listing symbols used in the recovery of coded packet $\underline{x}[5]^\mathbf{t}$, from each layer of codeword.}
%				\label{tab:diag_embedding_rec_sets}
%			\end{table}
		\end{example}
		\subsection{Requirements on the Block Code} \label{sec:bc_req} 
		
		{\it Case (i)}: Let $t\in[0:n-2-\tau]$ denote an erased coordinate. Owing to the delay constraint $\tau$, all the coordinates in $[t+\tau+1:n-1]$ are unavailable to the delay-constrained decoder irrespective of whether some of these coordinates are erased or not. For any erased coordinate $i<t$, as the $i$-th coordinate should be decodable by accessing code symbols up to the coordinate $(i+\tau)<(t+\tau)$. Hence during the decoding of $c_t$, all the symbols $c_0,\ldots,c_{t-1}$ can be assumed to be known. In summary, we have:
		\bean
		\underbrace{c_0,\cdots,c_{t-1} }_{\text{  known  }}, \ \underbrace{c_t}_{\text{  symbol to be recovered  }}, \underbrace{c_{t+1},\cdots,c_{t+\tau} }_{\text{ all the non-erased symbols are  accessible  }},
		\underbrace{c_{t+\tau+1},\cdots,c_{n-1} }_{\text{  inaccessible symbols, beyond delay constraint  }}.
		\eean
		
		Let $\mathcal{K}$ denote the set of coordinates $[0:t-1]$ and $\mathcal{U}$ denote the set of coordinates $[t+\tau+1:n-1]$.
		Thus one is faced with the task of decoding the code symbol $c_t$ when the code symbols $c_i, i \in \mathcal{K}$ are known, $c_i, i \in \mathcal{U}$ are inaccessible and some of the code symbols among $c_i,i\in[t+1:t+\tau]$ are possibly erased. Let $\mathcal{C}_{\mathcal{K},\mathcal{U}}$ denote the code obtained from \calc\ by: 
		\bit
		\item shortening on the coordinates in $\mathcal{K}$ and 
		\item puncturing on coordinates in $\mathcal{U}$.
		\eit 
		Let $H$ be the p-c matrix of \calc\ and set:
		\bean
		\underline{s} & = & - \sum_{j \in \mathcal{K}} c_j \underline{h}_{j}.
		\eean
		Then in order to recover $c_t$, we are faced with decoding the coset of the code $\mathcal{C}_{\mathcal{K},\mathcal{U}}$ corresponding to syndrome $\underline{s}$, when some of the code symbols from the set $\{c_{t+1},\ldots,c_{t+\tau}\}$ are erased. 
		The p-c matrix of $\mathcal{C}_{\mathcal{K},\mathcal{U}}$  can be obtained by first creating a matrix $H_1$ obtained by deleting the columns of $H$ corresponding to coordinates in $\mathcal{K}$.  Then one identifies a basis for the subspace of the rowspace of $H_1$ consisting of all the vectors having zeros in the coordinates making up $\mathcal{U}$ (see Lemma \ref{lem:punc_short_duality}). The required p-c matrix of $\mathcal{C}_{\mathcal{K},\mathcal{U}}$ will be formed using the basis vectors as its rows, after removing the trivial zero coordinates corresponding to $\mathcal{U}$.
		
		%%Let \cala\ denote the set of coordinates $[0:(i-1)]$ and \calb\ the set    
		%%Here one is faced with decoding the codeword when the last $(n-i-t-1)$ code symbols corresponding to coordiantes in \cala\ are unknown.  Thus one is faced with the task of decoding the punctured code.  The p-c matrix $H_1$ of the punctured code is any generator matrix for the dual code shortened on the coordinates in \cala. One can find $H_1$ by identifying a basis for the subspace of the rowspace of $H$, corresponding to vectors with zeros in the coordinates that belong to \cala.
		%
		%Thus one is faced with the task of decoding the code 
		%
		%punctured code.  The p-c matrix $H_1$ of the punctured code is any generator matrix for the dual code shortened on the coordinates: $[i+t+1:n-1]$.   One can find $H_1$ by identifying a basis for the subspace of the rowspace of $H$, corresponding to vectors with zeros in coordinates $[i+t+1:n-1]$.
		%
		%
		%
		
		{\it Case (ii)}:
		Let $t\in[n-1-\tau:n-1]$ denote an erased coordinate. Here one is faced with the task of decoding the code symbol $c_t$ when the code symbols $c_i, i \in \mathcal{K}$ are known and some of the symbols among $c_i,i\in [t+1:n-1]$ are possibly erased. The scenario is summarized as follows:
		
		\bean
		\underbrace{c_0,\cdots,c_{t-1} }_{\text{  known  }}, \ \underbrace{c_t}_{\text{  symbol to be recovered  }}, \underbrace{c_{t+1},\cdots,c_{n-1} }_{\text{all the non-erased symbols are  accessible }}.
		\eean

		Let $\mathcal{C}^{\mathcal{K}^c}$ denote the code obtained from \calc\ by shortening on the coordinates in $\mathcal{K}$. Let $H$ be the p-c matrix of \calc\ and set
		\bean
		\underline{s} & = & - \sum_{j \in \mathcal{K}} c_j \underline{h}_{j}.
		\eean
	In order to recover $c_t$, we are faced with decoding the coset of the code $\mathcal{C}^{\mathcal{K}^c}$ corresponding to the syndrome $\underline{s}$. 
		The p-c matrix of $\mathcal{C}^{\mathcal{K}^c}$ here can be obtained by simply deleting the columns of $H$ corresponding to the coordinates in $\mathcal{K}$.
		
	\subsection{Consequent Requirements on the P-C Matrix $H$} \label{sec:four_conditions} 
		
		Let $H$ be the p-c matrix of the code \calc.  For $0 \leq \ell \leq (n-2-\tau)$, set
		\bean
		H^{(\ell)} & = & \left[\underline{h}^{(\ell)}_{0}\ \ldots\  \underline{h}^{(\ell)}_{\ell+\tau}  \right],
		\eean
		be a p-c matrix for the punctured code $\mathcal{C}|_{[0:\ell+\tau]}$. Recall that $\mathcal{C}|_{[0:\ell+\tau]}$ is the restriction of \calc\ to the coordinates in $[0:\ell+\tau]$.  We will refer to \hsupl\ as a {\em shortened p-c matrix} as the dual of a punctured code is a shortened code and the rowspace of \hsupl\ is precisely the dual code. 
		Then applying Lemma \ref{lem:recovery_parity_col_span} and based on the observations in Section~\ref{sec:bc_req}, we have that the following conditions need to be satisfied by the p-c matrices $\{H^{(\ell)}\}$ and $H$: 
		\bean
		H, \ \ \{H^{(\ell)} \mid 0 \leq \ell \leq n-2-\tau\}.
		\eean 
		\ben
		\item {\bf Condition B1} For $0\leq \ell \leq (n-2-\tau) $, the $\ell$-th column, $\underline{h}^{(\ell)}_{\ell}$ of $H^{(\ell)}$ should be linearly independent of the set of $(b-1)$ columns  
		\bean
		\left\{ \underline{h}^{(\ell)}_{j} \mid \ell+1 \leq j \leq \ell+b-1 \right\}.
		\eean
		\item {\bf Condition B2} For $(n-1-\tau) \leq \ell \leq (n-b)$, the set
		\bean
		\left\{ \underline{h}_{j} \mid \ell \leq j \leq \ell +b-1 \right\}, 
		\eean
		of columns of $H$ should be linearly independent. 
		\item {\bf Condition R1}   
		For $0\leq \ell \leq (n-2-\tau) $, the column $\underline{h}^{(\ell)}_{\ell}$ of $H^{(\ell)}$ should be linearly independent of any set of $(a-1)$ columns drawn from the set 
		\bean
		\left\{\underline{h}^{(\ell)}_{j} \mid \ell+1 \leq j \leq \ell+\tau \right\}. 
		\eean
		\item {\bf Condition R2} Any set of $a$ columns from the set 
		\bean
		\left\{ \underline{h}_{j} \mid n-1-\tau \leq j \leq n-1 \right\} 
		\eean
		should be linearly independent. 
		\een

		\begin{remark}\normalfont\label{rem:dim_length_block_code}
			In this paper, we aim to construct streaming codes for the DC-SW channel whose rates meet \eqref{eq:rate_upper_bound} with equality. Hence throughout the remainder of this paper, we choose the dimension $k$ of the block code $\mathcal{C}$ to be diagonally embedded, as $(\tau-a+1)$ and code length $n$ of $\mathcal{C}$ to be $(\tau+\delta+1)$. This way, the resultant convolutional code after diagonal embedding, \cstr\ has a rate $R$ which meets the bound \eqref{eq:rate_upper_bound}. We note in passing that if $(\tau-a+1,b)=1$, then such a construction will also have least possible coded packet size $n$, i.e., will also be {\it packet-size-optimal}.
		\end{remark}
		
%		\begin{remark}\normalfont\label{rem:delay_sensitive_coords}
%			From Remark \ref{rem:dim_length_block_code}, length $n$ of the code $\mathcal{C}$ is $(t+\delta+1)$. For $i\in[0:\delta-1]$, we have $\Delta_i=t+i<(n-1)$ and hence we cannot make use of code symbols belonging to coordinates $[(t+i+1):n-1]$ for the recovery of an erased coordinate $i$. For this reason, we will refer to the set of coordinates $[0:\delta-1]$ of $\mathcal{C}$ as {\it delay-sensitive} coordinates. We also say that the delay constraint is inactive for the coordinates $[\delta:n-1]$.
%		\end{remark}
        \begin{remark}[Systematicity]\normalfont As we choose dimension $k$ as $(\tau-a+1)$ and code length $n$ as $(\tau+\delta+1)$, we have $(n-k)=b$. Condition B2 ensures that, for the p-c matrix $H$ of size $b\times (\tau+\delta+1)$, the last $b$ columns are independent. This implies that the code \calc\ has a generator matrix of the form $[I_k\ P]$, where the first $k$ columns are independent. This aligns with the description of the diagonal-embedding-based encoder we illustrated in Fig. \ref{fig:diag_interleaving_example}.
        \end{remark}

	In Fig. \ref{fig:our_approach}, we summarize the approach used in this paper to construct streaming codes for the DC-SW channel model.

		\begin{figure}[!htb]
			\centering
			\captionsetup{justification=centering}
			\includegraphics[scale=0.5]{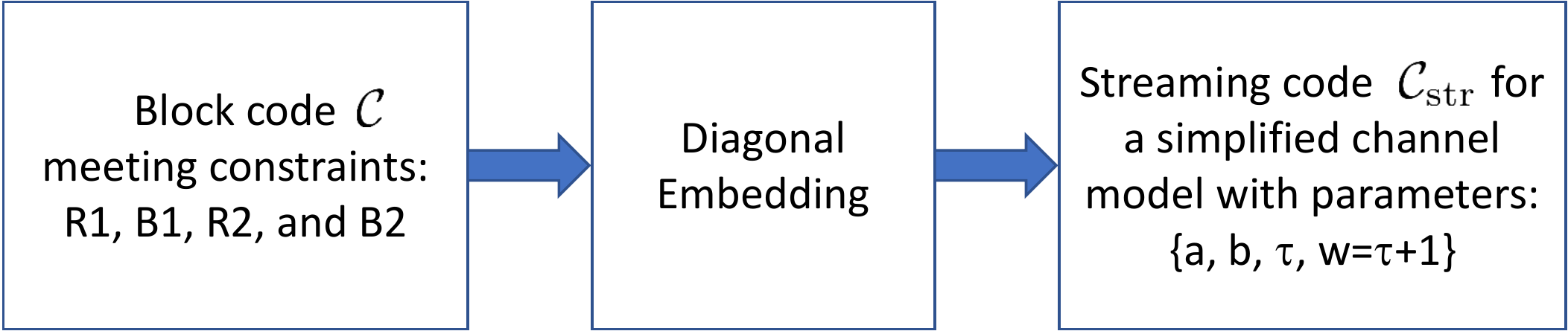}
			\caption{The approach used in this paper for construction of streaming codes for the DC-SW channel model.}
			\label{fig:our_approach}
		\end{figure}
		
	%===============SECTION IV CONSTRUCTION A =======================
	%===============SECTION IV CONSTRUCTION A =======================
	%===============SECTION IV CONSTRUCTION A =======================

	\section{Quadratic Field-Size Construction A: for all Parameter Sets}\label{sec:construction_A}
 
In the present section and the next three sections to follow, we will provide four different constructions for the diagonally-embedded block code \calc\ underlying the streaming code \cstr.  We will simply refer to \calc\ as the block code. As the parameter $w=(\tau+1)$ is redundant, we consider the reduced parameter set \params\ for \calc. As mentioned in Remark \ref{rem:dim_length_block_code}, our approach towards constructing a rate-optimal streaming code \cstr\ whose rate matches with the upper bound in \eqref{eq:rate_upper_bound}, is to construct a block code $\mathcal{C}$ with dimension, $k=(\tau-a+1)$ and code length, $n=(\tau+1+\delta)$. Under these choice of values for parameters $k,n$, since $(n-k)=(\delta+a)=b$, the  p-c matrix will always be of size $(b\times n)$.

The construction in the present section yields block codes for all parameters $a,b,\tau$, and employs a field-size that is $O(\tau^2)$.  The required field-size in this construction, which we will refer to as Construction A, is $q^2$ where $q\geq(\tau+1)$. We describe the construction by successively refining or updating in four steps, our description of the p-c matrix $H$ of the code $\mathcal{C}$ over $\mathbb{F}_{q^2}$. We initialize $H$ to be the $(b\times n)$ all-zero matrix.
\bit 
\item {\it Step-a}: Set the submatrix $H(\delta:b-1,0:\tau)$ of $H$ to be the generator matrix of an MDS code over $\fq\subseteq\fqs$ that is of the form $[I_a\ C]$, where $C$ is an $a\times(\tau+1-a)$ Cauchy-like matrix.  As $a+(\tau+1-a)=(\tau+1)\leq q$, there exists a Cauchy matrix of size $a\times (\tau+1-a)$ \cite[Ch.~11]{macwilliamssloane}, which may be chosen as $C$.
\item {\it Step-b}: Set $H(0:\delta-1,0:\delta-1)= \alpha I_{\delta}$, where $\alpha\in \mathbb{F}_{q^2}\setminus \mathbb{F}_q$.
\item {\it Step-c}: For $i \in[0:\delta-1]$, $j \in [b+i:\tau+i]$, set $H(i,j)=v_{i,j}$.  Here the \vij\ are variables which will be assigned values drawn from \fq.
\item {\it Step-d}: Additionally, if $\delta >0$, set $H(\delta,\tau+\delta)=1$ (any non-zero value in place of $1$ would also work).
\eit

In Fig. \ref{fig:construction_A_example}, we illustrate Construction A for parameters \params\ $=\{5,8,12\}$.

\begin{figure}[h!]
	\centering
	\captionsetup{justification=centering}
	\includegraphics[scale=0.55]{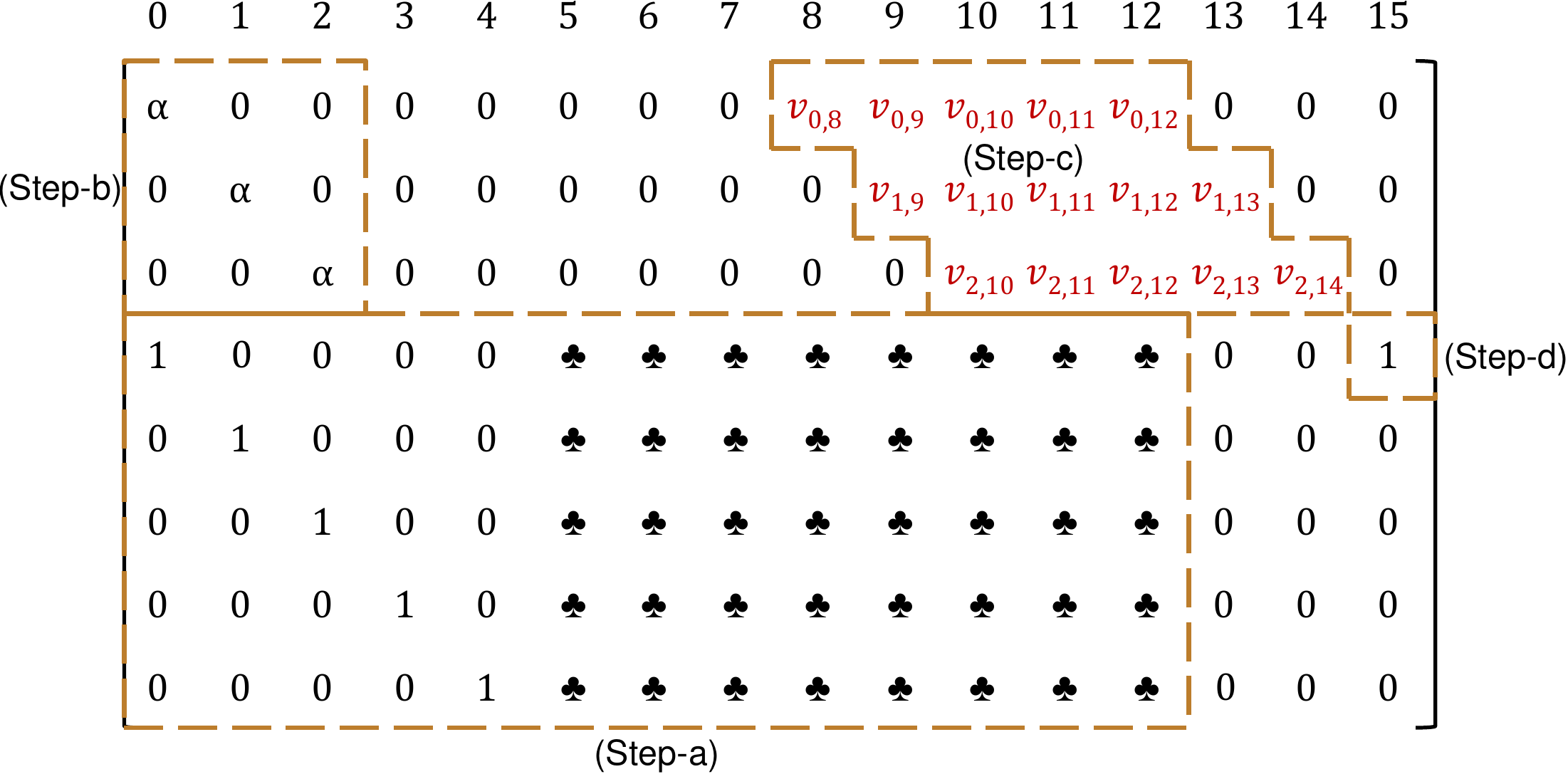}
	\caption{The p-c matrix of an example code constructed using Construction A for the parameter set $\{a=5,b=8,\tau=12\}$. Here we set $q=2^4$ and $\alpha \in \mathbb{F}_{q^2}\setminus\mathbb{F}_{q}$. The submatrix $H(3:7,5:12)$ (elements denoted by $\clubsuit$) is set to be a $(5\times 8)$ Cauchy-like matrix $C$ whose elements belong to $\mathbb{F}_{q}$.  The \vij\ are variables which will also be assigned values drawn from the subfield $\mathbb{F}_{q}$. }
	\label{fig:construction_A_example}
\end{figure}

\begin{remark}\normalfont
	For the trivial case $b=a$ and hence $\delta=0$, the p-c matrix $H$ of $\calc$ takes the form $[I_a\ C]$. Thus $\calc$ is a $[\tau+1,\tau-a+1]$ MDS code over \fq. One can easily show that such an MDS code satisfies properties B1, B2, R1, R2 and therefore, results in a rate-optimal streaming code via diagonal embedding. Hence throughout the paper, we assume $\delta>0$.
\end{remark}

\begin{thm}\label{thm:constr_A}
	If $q\geq (\tau+1)$, there exists an assignment of the variables $\{v_{i,j}\}$ over $\mathbb{F}_q$ such that the code $\mathcal{C}$ over $\mathbb{F}_{q^2}$ having p-c matrix $H$, when employed as the diagonally-embedded block code, will yield a rate-optimal streaming code \cstr.
\end{thm}
\begin{proof}  Clearly, it suffices to show that there exists an assignment of the variables \vij\ over $\mathbb{F}_q$ so that the p-c matrix $H$ satisfies the four conditions laid out in Section~\ref{sec:four_conditions}. The proof is deferred to Appendix \ref{app:proof_constrn_a}. 
	
\end{proof}

\begin{remark} [$O(\tau)$ Field-Size Construction for $\delta=1$] \normalfont
	If $\delta=1$, the property $\alpha\in\fqs\setminus\fq$ will not be used in the proof of Theorem \ref{thm:constr_A}. Hence for the case of $(\delta=1)$, $\alpha$ may be chosen to be a non-zero element drawn from $\mathbb{F}_q$. This results in a code \calc\ over $\mathbb{F}_q$.
\end{remark}
	
	%%%%%%%SECTION-Construction-B	
		%%%%%%%SECTION-Construction-B	
			%%%%%%%SECTION-Construction-B	
				%%%%%%%SECTION-Construction-B	
					%%%%%%%SECTION-Construction-B	

%	\newpage	
		\section{Linear Field-Size Construction B:  for $\delta\geq a$ and $(\tau+1)\geq b+\delta$} \label{sec:construction_B}
	
	In this section, we provide a second construction (which we will refer to as Construction B), which is explicit and furthermore, requires a field-size $q$ that is linear in $\tau$.  More specifically, the construction works for all parameters \params\ satisfying:
	\bit
	\item  $(\tau+1)\geq(b+\delta)$, $\delta\geq a$ and $q\geq(\tau+1)$. 
	\eit
	As in the other three constructions, we assume that the block code length, $n= (\tau+1+\delta)$ and dimension, $k=(\tau-a+1)$.   Again, we describe the construction by successively refining in four steps, our description of the p-c matrix $H$ of the code $\mathcal{C}$ over $\mathbb{F}_q$. We initialize $H$ to be the $(b\times n)$ all-zero matrix. 
	\bit
	\item Step-a: Let \gmds\ be a ZB generator matrix corresponding to a $[(\tau+1),b]$ MDS code, say $\mathcal{C}_\text{\tiny MDS}$, over $\mathbb{F}_q$.  We update $H$ by setting: $H(:,0:\tau)=\gmds$. As $q\geq (\tau+1)$, the existence of such a matrix \gmds\ is guaranteed.
\item Step-b: Next, we update $H$ by setting $H(a:\delta,\tau+a:\tau+\delta)=H(a:\delta,a:\delta)$ (elements of the submatrix on the RHS are already defined in Step-a).
\item {\it Step-c}: For $1\leq j\leq (a-1)$, we set $H(b-j,\tau+j)=1$ (any non-zero value in place of $1$ would also work).
\item {\it Step-d}: In the final update, we replace $H(\delta:b-1,0:b-1)$ with an $(a\times b)$ Cauchy-like matrix $C$. As $(a+b)\leq (\delta+b)\leq (\tau+1)\leq q$, we are guaranteed the existence of such a Cauchy-like matrix $C$.
\eit
	
	In Fig. \ref{fig:delta_geq_a_example_step1}--\ref{fig:delta_geq_a_example_step4}, we illustrate the four steps involved in the construction of $H$ for parameters \params\ $=\{3,8,14\}$.
	
	\begin{figure}[!htb]
		\centering
		\captionsetup{justification=centering}
		\includegraphics[scale=0.6]{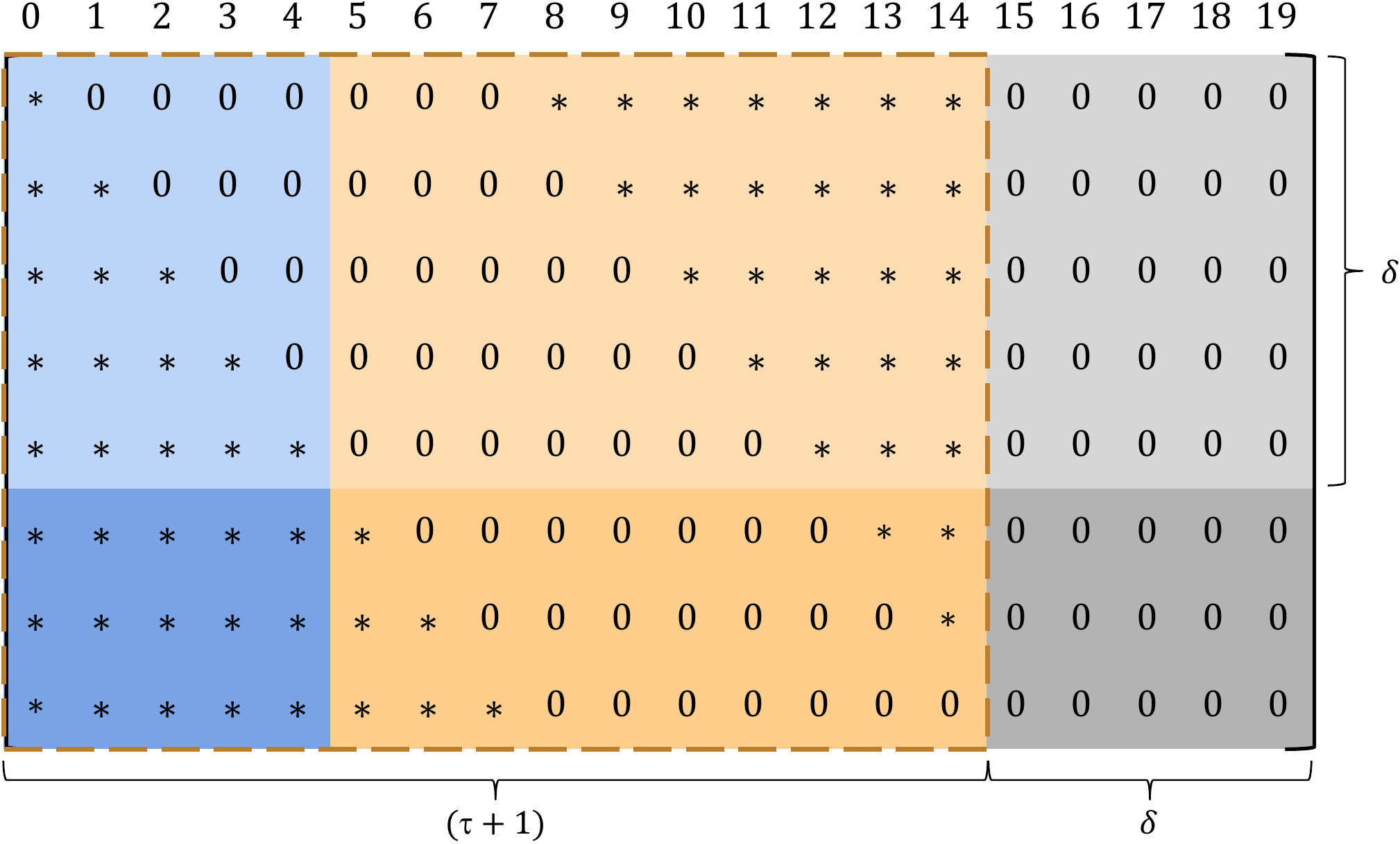}
		\caption{$H$ after Step-a for parameters \params\ $=\{3,8,14\}$. The submatrix $H(:,0:14)$ (demarcated by dashed lines) is a ZB generator matrix \gmds\ corresponding to a $[15,8]$ MDS code. }
		\label{fig:delta_geq_a_example_step1}
	\end{figure}

	\begin{figure}[!htb]
		\centering
		\captionsetup{justification=centering}
		\includegraphics[scale=0.6]{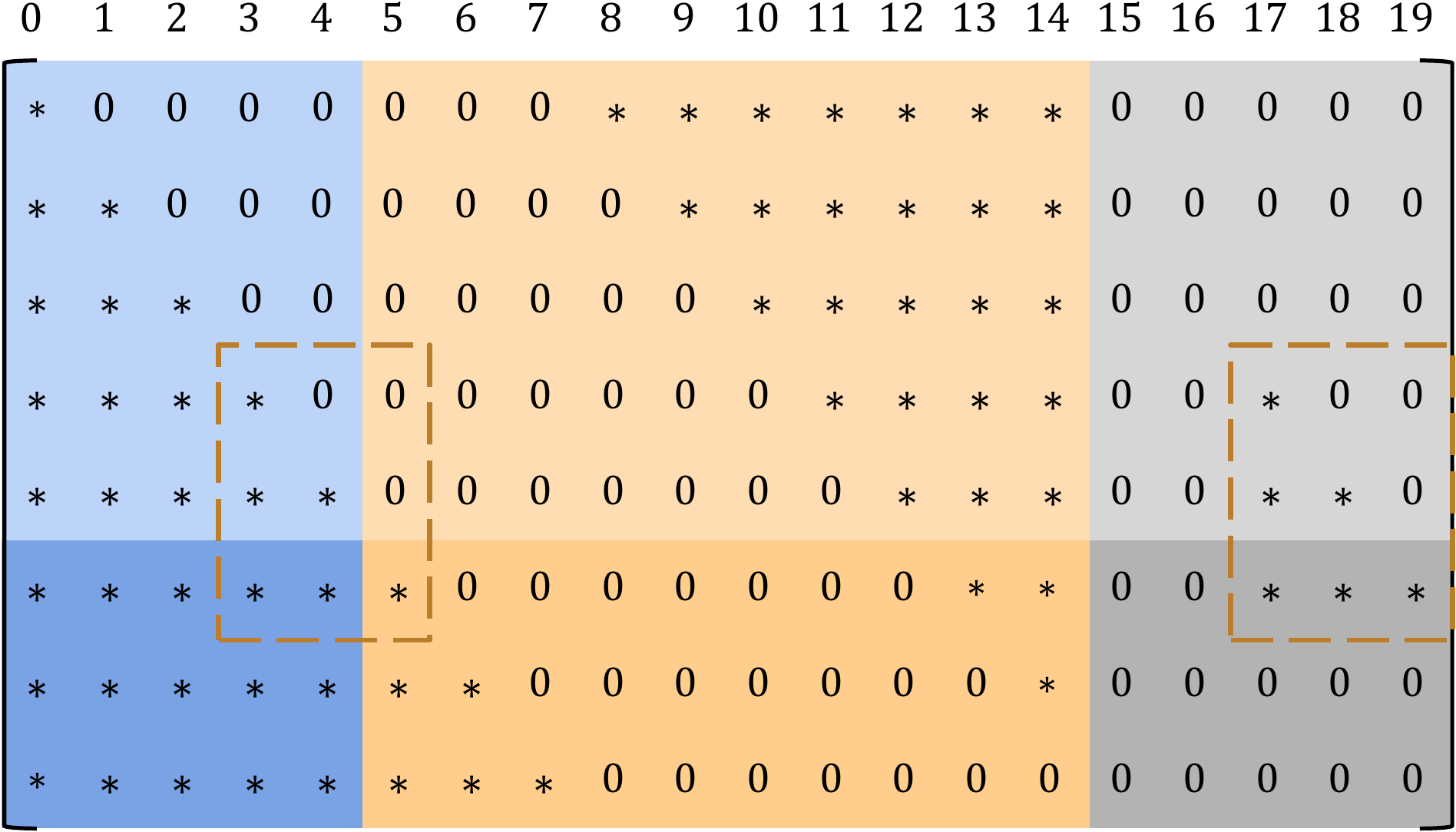}
		\caption{$H$ after Step-b for parameters \params\ $=\{3,8,14\}$. Here we replicate a portion of the p-c matrix.}
		\label{fig:delta_geq_a_example_step2}
	\end{figure}

	\begin{figure}[!htb]
		\centering
		\captionsetup{justification=centering}
		\includegraphics[scale=0.6]{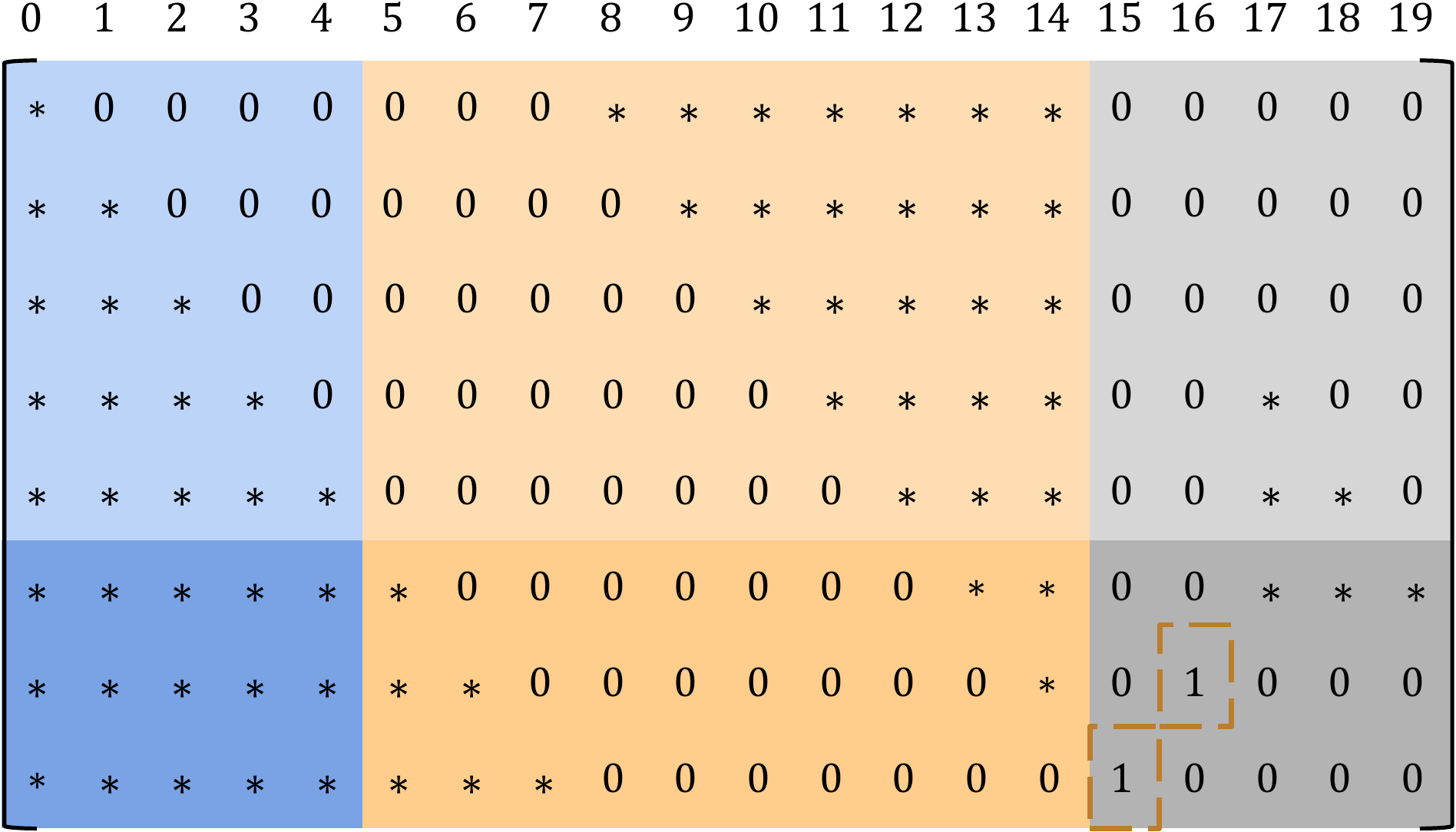}
		\caption{$H$ after Step-c for parameters \params\ $=\{3,8,14\}$. Here we replace some of the $0$'s with $1$'s. }
		\label{fig:delta_geq_a_example_step3}
	\end{figure}

	\begin{figure}[!htb]
		\centering
		\captionsetup{justification=centering}
		\includegraphics[scale=0.6]{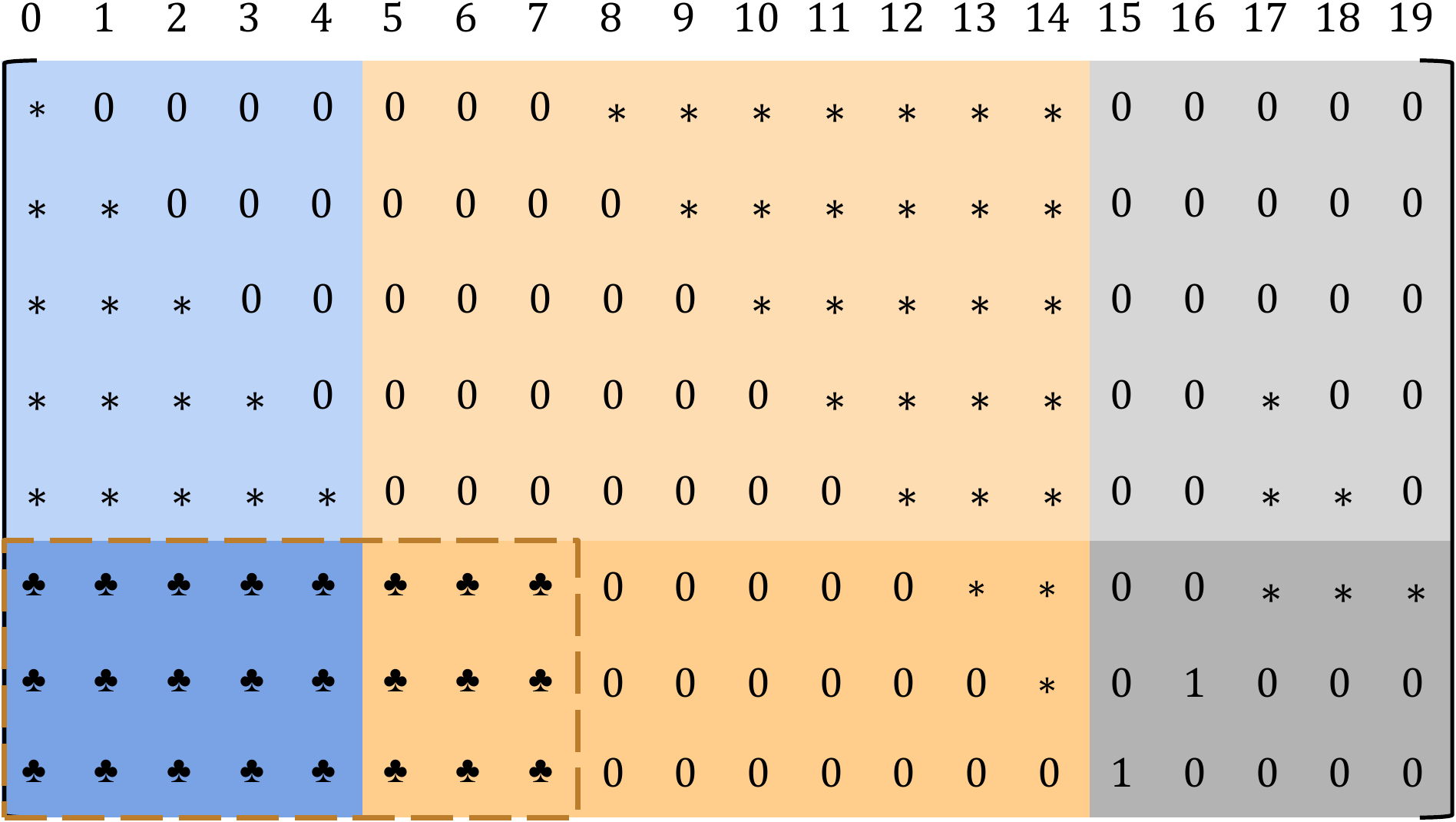}
		\caption{$H$ after Step-d for parameters \params\ $=\{3,8,14\}$. In Step-d, we set the demarcated section of $H$ to be a Cauchy-like matrix. Here $\clubsuit$'s denote elements of a Cauchy-like matrix.}
		\label{fig:delta_geq_a_example_step4}
	\end{figure}
	
		\begin{thm}\label{thm:constr_B}
		The code $\mathcal{C}$ over $\mathbb{F}_{q}$ having p-c matrix $H$ based on Construction B, when employed as the diagonally-embedded block code, will yield a rate-optimal streaming  code \cstr.	
	\end{thm}
	\begin{proof}
		Clearly, it suffices to show that $H$ meets all the four conditions; B1, B2, R1 and R2 described in Section \ref{sec:four_conditions}. The proof is deferred to Appendix \ref{app:proof_constrn_b}.
    \end{proof}

%	\newpage	
		
		\section{Construction C: Interleaving MDS Codes }\label{sec:construction_C}
		%\section{Linear Field-Size Construction C: for $a\mid b \mid(t-a+1)$}
		
		In this section, we present a second linear field-size construction, Construction C, obtained by simply interleaving MDS codes.  The code requires parameters \params\  to satisfy: 
		\bit
		\item $a\mid b \mid(\tau-a+1)$.
		\eit
		As always, $n= (\tau+\delta+1)$, and thus we have $b\mid n$ as well. Let $\alpha\triangleq \frac{b}{a}$, $\beta\triangleq \frac{n}{b}$. In terms of $\alpha,\beta$ and $a$, one can express $\tau$ and $\delta$ as: $\tau=(\beta-1)\alpha a+a-1$ and $\delta=(\alpha-1)a$, respectively. The construction is over a field \fq, of size $q$ satisfying: $q\geq a\beta$.   The description below of Construction C is, again, in terms of the p-c matrix $H$: 
		
	 \bit
	 \item  Let $G_\text{\tiny MDS}$ denote the generator matrix  of an $[a\beta,a]$ MDS code \cmds\ over $\mathbb{F}_q$.  Clearly, as $q\geq a\beta$, the required generator matrix $G_\text{\tiny MDS}$ of the MDS code can be found.
	 
	 \item  For $0\leq i\leq (\beta-1)$, group the $a$ adjacent columns $[ia:(i+1)a-1]$ of $G_\text{\tiny MDS}$ to form the matrix $G^{(i)}_\text{\tiny MDS}$.  Thus we have $G_\text{\tiny MDS}=[G^{(0)}_\text{\tiny MDS} \ G^{(1)}_\text{\tiny MDS}\ \ldots \ G^{(\beta-1)}_\text{\tiny MDS}]$. Here, without loss of generality, one can assume that  $G^{(\beta-1)}_\text{\tiny MDS}=I_a$. 
	 
	 \item The p-c matrix $H$ of the code is then built up of the matrices $\left\{G^{(i)}_\text{\tiny MDS} \right\}$ as shown in Fig. \ref{fig:interleaved_mds_code}. 
	 
	 	\begin{figure}[!htb]
	 	\centering
	 	\captionsetup{justification=centering}
	 	\includegraphics[scale=0.8]{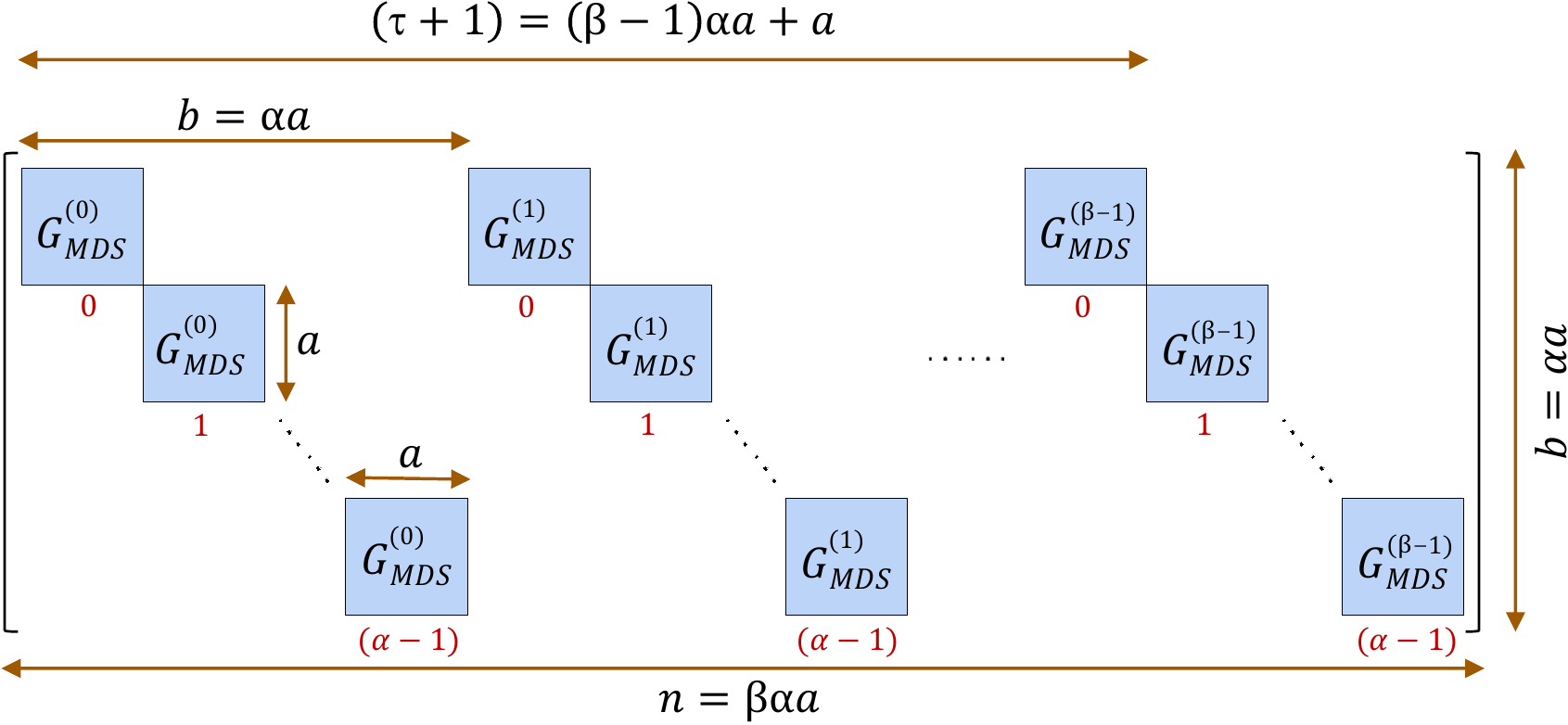}
	 	\caption{The p-c matrix associated to Construction C.   Here $G_\text{\tiny MDS}=[G^{(0)}_\text{\tiny MDS} \ G^{(1)}_\text{\tiny MDS}\ \ldots \ G^{(\beta-1)}_\text{\tiny MDS}]$ is the generator matrix of a $[a\beta,a]$ MDS code \cmds. From this, the interleaved MDS nature of the construction is apparent.}
	 	\label{fig:interleaved_mds_code}
	 \end{figure}
	 \eit

		\begin{example}\label{eg:interleaved_mds_code}\normalfont
		Let parameters $\params=\{2,6,13\}$. Thus we have $\alpha=3,\beta=3,n=18,\delta=4$. Let $G_\text{\tiny MDS}$ be as follows:
		\begin{equation*}
		G_\text{\tiny MDS}=\left[
		%\mathbf{Z}=\left[ 
		{\begin{array}{cc|cc|cc}
			g_{0,0} & 	g_{0,1} & 	g_{0,2} & 	g_{0,3} &	1 & \ \	0 \\
		g_{1,0} & 	g_{1,1} & 	g_{1,2} & 	g_{1,3} &	0 & \ \	1  
			\end{array}}\right].
		\end{equation*}
		The p-c matrix $H$ associated to Construction C is as given in Fig. \ref{fig:interleaved_mds_code_example_1}.		
		
	  \begin{figure}[!htb]
			\centering
			\captionsetup{justification=centering}
			\includegraphics[scale=0.65]{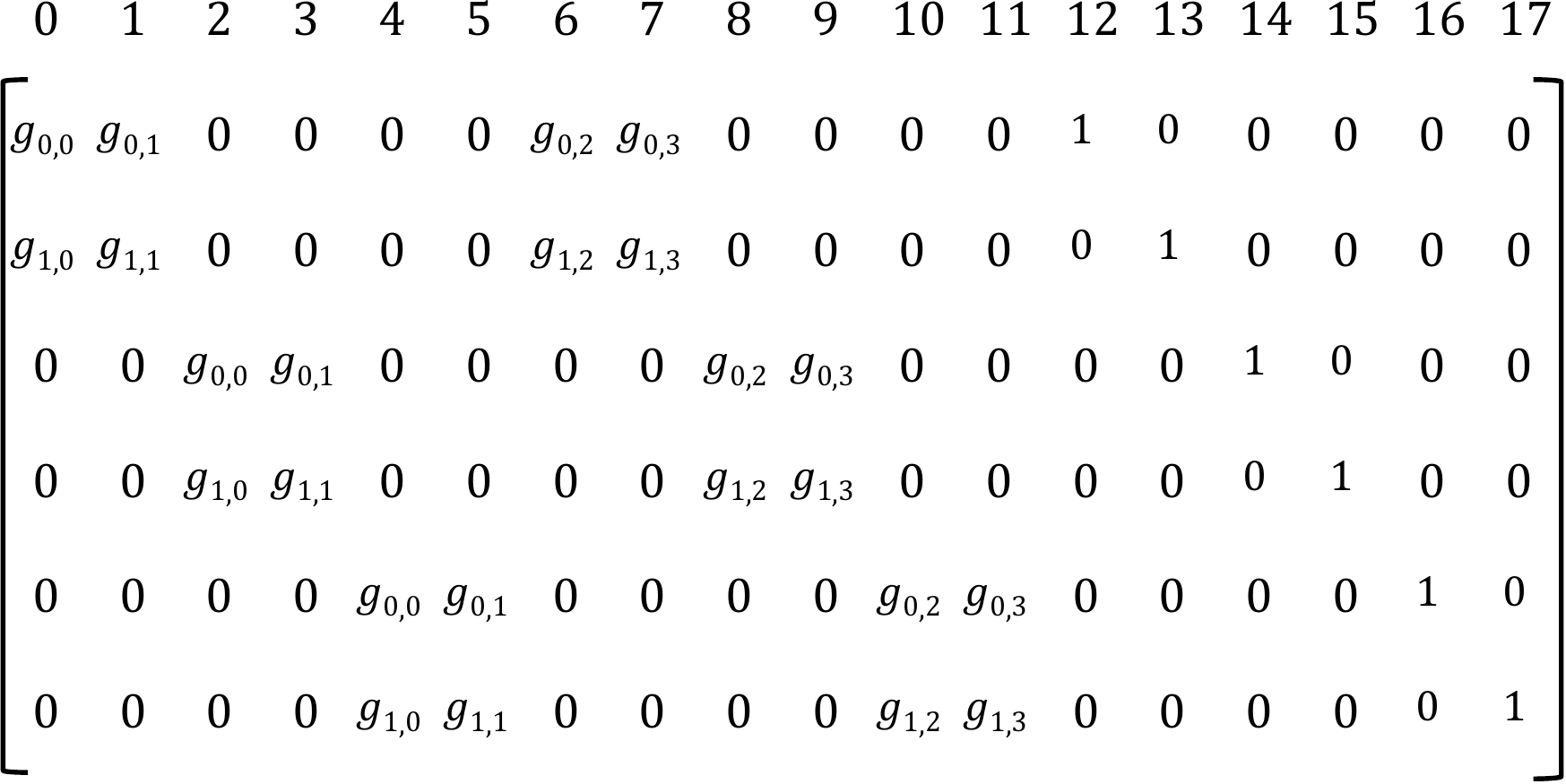}
			\caption{$H$ corresponding to the Example \ref{eg:interleaved_mds_code}. Here $\{a=2,b=6,\tau=13\}$.}
			\label{fig:interleaved_mds_code_example_1}
		\end{figure}
		
\end{example}

%%%%%%% Interleaved MDS code section%%%%%%%%%%%%%		
%%%%%%% Interleaved MDS code section%%%%%%%%%%%%%			
%%%%%%% Interleaved MDS code section%%%%%%%%%%%%%			
%%%%%%% Interleaved MDS code section%%%%%%%%%%%%%			
%%%%%%% Interleaved MDS code section%%%%%%%%%%%%%			
		
		\begin{thm}\label{thm:constr_C}
	The code $\mathcal{C}$ over $\mathbb{F}_{q}$ having p-c matrix $H$ based on Construction C, when employed as the diagonally-embedded block code, will yield a rate-optimal streaming  code \cstr.
	   \end{thm}
       \begin{proof}
       	Clearly, as in the case of previous constructions, it suffices to show that $H$ satisfies all the four conditions; B1, B2, R1 and R2.
       	\bit
       	\item[]
       	\item {\it Recovery from burst erasure of length $\leq b$}:  
       	\bit
       	\item Condition B1: Let $0\leq\ell\leq (\alpha-1)a-1\triangleq(\delta-1)$. Partition the p-c matrix into $\alpha$ submatrices as follows (also, see Fig. \ref{fig:interleaved_mds_code_be_proof_1}):
       	 \bean
       	H=\left[ \begin{array}{c} H_0\\ H_1\\ \vdots\\H_{(\alpha-1)} \end{array} \right]. 	
       	\eean
       	  
       	\begin{figure}[!htb]
       		\centering
       		\captionsetup{justification=centering}
       		\includegraphics[scale=0.8]{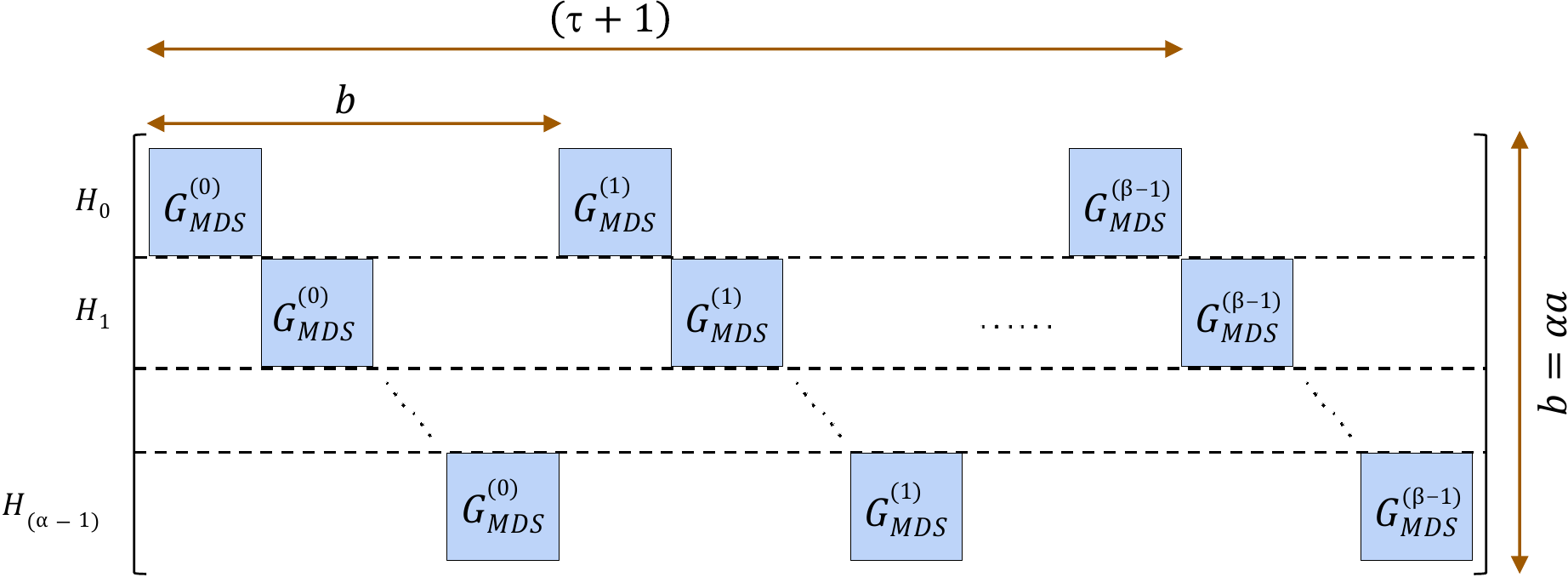}
       		\caption{Figure showing the partitioning of the p-c matrix shown in \ref{fig:interleaved_mds_code}.}
       		\label{fig:interleaved_mds_code_be_proof_1}
       	\end{figure}
       	Clearly, the $\beta a$ non-zero columns within an $H_i$ ($0\leq i\leq (\alpha-1)$) correspond to the $\beta a$ columns of \gmds. Fix $\ell\in[0:\delta-1]$. The row $H(\ell,:)$ naturally belongs to one of the submatrices $\{H_0.\ldots,H_{\alpha-2}\}$. In order to make this mapping explicit, we set $\ell=\mu a+\nu$, for some $\mu,\nu$ such that $0\leq \mu \leq (\alpha-2)$ and $0\leq \nu\leq (a-1)$. In this form, we have that $H(\ell,:)$ is a row of $H_\mu$. It can be verified from Fig. \ref{fig:interleaved_mds_code_be_proof_1} that the corresponding $H_\mu$ has the last $(\alpha-1-\mu)a\geq (\alpha-1-\mu)a-\nu=(\delta-\ell)$ coordinates as zero columns. Thus rows of $H_\mu(:,0:\ell+\tau)$ lie in the row space of \hsupl (see Fig. \ref{fig:interleaved_mds_code_be_proof_2} for an example). We note that any consecutive set of $b$ coordinates involve precisely $a$ non-zero columns of $H_i$ for any $i\in [0:\alpha-1]$. Moreover, these $a$ non-zero columns which correspond to each $H_i$ form an independent set as they are columns of an MDS matrix \gmds. It follows that $H_\mu(:,\ell)$ (which is a non-zero column of $H_\mu$) is linearly independent of the set of $(b-1)$ columns (only $(a-1)$ of them are non-zero columns):  
       	\bean
       	\left\{ H_\mu(:,j) \mid \ell+1 \leq j \leq \ell+b-1 \right\},
       	\eean
       	which implies condition B1.
     
       \begin{figure}[!htb]
       		\centering
       		\captionsetup{justification=centering}
       		\includegraphics[scale=0.65]{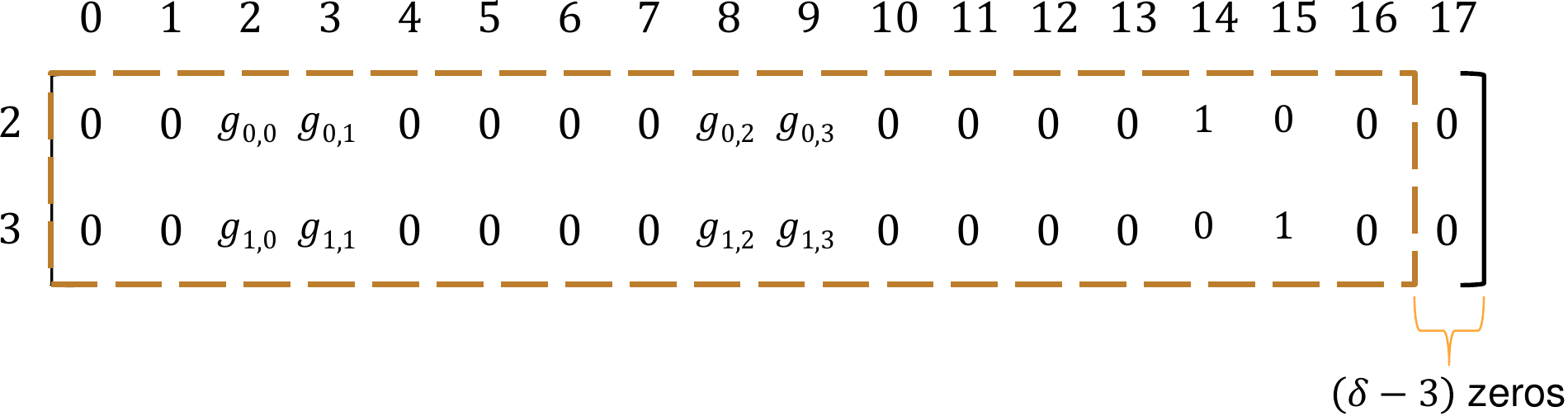}
       		\caption{Consider the Example \ref{eg:interleaved_mds_code}. Let $\ell=3$. Thus we have $\mu=1$ and $\nu=1$. In the figure, we have shown the submatrix $H_1$ (as $\mu=1$). The last $(\delta-3)=1$ columns of $H_1$ are zero columns. Hence rows of the matrix demarcated by dashed lines lie in the row space of $H^{(3)}$.  
       		}
       		\label{fig:interleaved_mds_code_be_proof_2}
       	\end{figure}

       	\item Condition B2:  Let $\ell\in[\delta: \tau-a+1]$. Consider any set of $b$ consecutive coordinates $[\ell:\ell+b-1]$, which are erased. Here, again, we note that any consecutive set of $b$ coordinates, $[\ell:\ell+b-1]$,  involve precisely $a$ non-zero coordinates of any $H_i$, $0\leq i\leq (\alpha-1)$ and all these $a$ non-zero columns of each $H_i$ form an independent set. Hence it follows that  	
       	\bean
       	\left\{ \underline{h}_{j} \mid \ell \leq j \leq \ell +b-1 \right\} 
       	\eean
        form an independent set (condition B2).
         \eit
        \item {\it Recovery from $\leq a$ arbitrary erasures}:  
        \bit
        \item Condition R1: Let $\ell\in[0:\delta-1]$ and recall the definition of $\mu,\nu$ such that: $\ell=a\mu+\nu$, where $0\leq \mu \leq (\alpha-2)$ and $0\leq \nu\leq (a-1)$. As seen in the condition B1 case, rows of $H_\mu(:,[0:\ell+\tau])$ lie in the row space of \hsupl. Also, we have already seen that $H_\mu(:,\ell)$ is non-zero and any $a$ non-zero columns of $H_\mu$ form an independent set. Clearly this implies condition R1.

        \item Condition R2: As any $a$ non-zero columns of $H_i$, $0\leq i\leq (\alpha-1)$ form an independent set, it
        follows that any $a$ columns of $H$ form an independent set. Hence condition R2 is satisfied.
        \eit
        \eit
        \end{proof}

  % 	\newpage

		\section{Linear Field-Size Construction D: for $\delta=(a-1)$ and $b\mid(\tau+1-\delta)$ }\label{sec:construction_D}
		%\section{Linear Field-Size Construction D:  Valid for $a\mid b \mid(t-a+1)$}

	In this section, we give a construction over \fq\ (will be referred to as Construction D) for parameters \params\ such that:
	\bit
	\item  $\delta\triangleq (b-a)=(a-1)$, $(\tau+1)=b+\delta+\gamma b$, for some $\gamma\in \{0,1,\ldots\}$ and $q\geq (\tau+2)$.
	\eit
	
	 As in the other three constructions, we set $n= (\tau+1+\delta)$ and $k=(\tau-a+1)$. The p-c matrix $H$ for Construction D is obtained by refining our description of $H$ in four steps. Initialize $H$ to be the zero matrix of size $(b\times n)$.
	 \bit
		
		\item {\it Step-a}: Consider a $[b+\delta,b]$ MDS code \cmds\ over \fq. For $i\in [0:\delta-1]$, choose a non-zero codeword $\underline{c}_i^T\in\mathcal{C}_\text{\tiny MDS}$ such that $c_{i,j}=0$ for $j\in [i+1:i+b-1]$. Also, choose a non-zero codeword $\underline{c}_\delta^T\in\cmds$ such that $c_{\delta,j}=0$ for $j\in [0:b-2]$. For $l\in [0:\delta]$, set $H(l,0:b+\delta-1)=\underline{c}_l^T$. As $q\geq (\tau+2)>(b+\delta)$, the required MDS code $\mathcal{C}_\text{\tiny MDS}$ can be obtained.

		\item {\it Step-b}: For $i\in [1:\gamma]$, set $H(0:\delta,b+\delta+(i-1)b:b+\delta-1+ib)=H(0:\delta,\delta:b+\delta-1)$ (elements of the submatrix on the RHS are already determined in Step-a).

		\item {\it Step-c}: Set all the zero entries of $H(\delta,\delta:\tau)$ to be $1$ (any non-zero value in place of $1$ would also work).

		\item {\it Step-d}: Let $C$ be a $(a\times (\tau+1-\delta))$ Cauchy-like matrix with $C(0,:)$ identical to $H(\delta,\delta:\tau)$. Set $H(\delta:b-1,\delta:\tau)=C$. Also, for $i\in[1:\delta]$, set $H(\delta+i,\tau+i)=1$ (any non-zero value instead of $1$ would also work). As $q\geq (a+\tau+1-\delta)=(\tau+2)$ (since $\delta=(a-1)$), the existence of such a $C$ is guaranteed. This completes the p-c matrix construction.
\eit

		\begin{example}\label{eg:constrn_D_example}\normalfont
		%\begin{figure}[!htb]
		%	\centering
		%	\captionsetup{justification=centering}
		%	\includegraphics[scale=0.55]{fig_construction_A_example}
		%	\caption{Parity check matrix using Construction A for parameters: $a=5,b=8,t=12,\delta=3$. Here $\alpha$ is a primitive  element of $\mathbb{F}_{2^8}$ satisfying $\alpha^8+\alpha^4+\alpha^3+\alpha^2+1=0$ and $\beta=\alpha^{17}$ is a primitive element of  $\mathbb{F}_{2^4}\subseteq \mathbb{F}_{2^8}$. Here $H(3:7,5:12)$ is a $5\times 8$ Cauchy matrix whose entries are over $\mathbb{F}_{2^4}$ and $v_{j,l}$'s are variables.}
		%	\label{fig:construction_A_example}
		%\end{figure}
		Consider the parameters $\{a=3,b=5,\tau=11\}$. Thus we have $\delta=2$ and $n=(\tau+1+\delta)=14$. We require a p-c matrix $H$ of size $(5\times 14)$. Set $q=2^4\geq (\tau+2)$. Initialize $H$ to be the zero matrix of size $(5\times 14)$. 
		\bit
		\item {\it Step-a}: Consider a $[7,5]$ MDS code $\mathcal{C}_\text{\tiny MDS}$ over $\mathbb{F}_{2^4}$. Set $H(0,0:6)$ to be a non-zero codeword $\underline{c}_0^T\in\mathcal{C}_\text{\tiny MDS}$ such that $c_{0,i}=0$ for $i\in [4]$, $H(1,0:6)$ to be a non-zero codeword $\underline{c}_1^T\in\mathcal{C}_\text{\tiny MDS}$ such that $c_{1,i}=0$ for $i\in [2:5]$, $H(2,0:6)$ to be a non-zero codeword $\underline{c}_2^T\in\mathcal{C}_\text{\tiny MDS}$ such that $c_{2,i}=0$ for $i\in [0:3]$ (see Fig. \ref{fig:construction_D_step1}). It can be verified that the matrix $H(0:2,0:6)$ has a rank of $3$ and  also $|\cup_{j=0}^{2}\text{supp}(H(j,0:6))|=5$. Hence applying Lemma \ref{lem:shortened_mds}, we make the observation that any $3$ non-zero columns of $H(0:2,:)$ form an independent set.
		
		\begin{figure}[!htb]
			\centering
			\captionsetup{justification=centering}
			\includegraphics[scale=0.55]{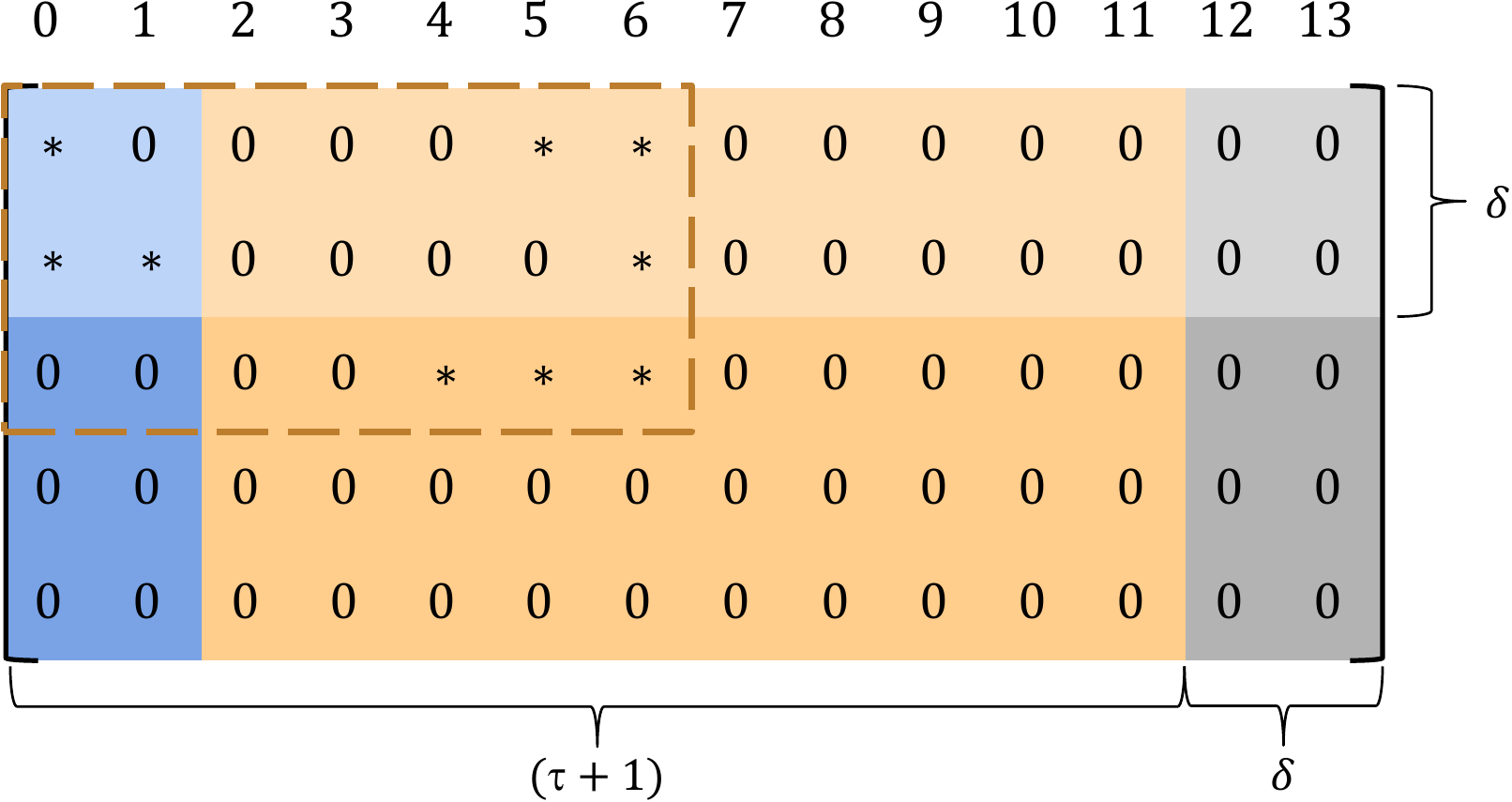}
			\caption{$H$ after Step-a for parameters \params\ $=\{3,5,11\}$.  Here rows of the submatrix demarcated by dashed lines, are codewords of a $[7,5]$ MDS code. After Step-a, any $3$ non-zero columns of $H(0:2,:)$ form an independent set.}
			\label{fig:construction_D_step1}
		\end{figure}
		
		\item {\it Step-b}: Set $H(0:2,7:11)=H(0:2,2:6)$ (see Fig. \ref{fig:construction_D_step2}). From the observation that we made in Step-a, it follows that there does not exist a set of $\leq 2$ columns of $H(0:2,[0:13]\setminus\{i\})$ which has $H(0:2,i)$ in its span, for $i=0,1$.
		
		\begin{figure}[!htb]
			\centering
			\captionsetup{justification=centering}
			\includegraphics[scale=0.55]{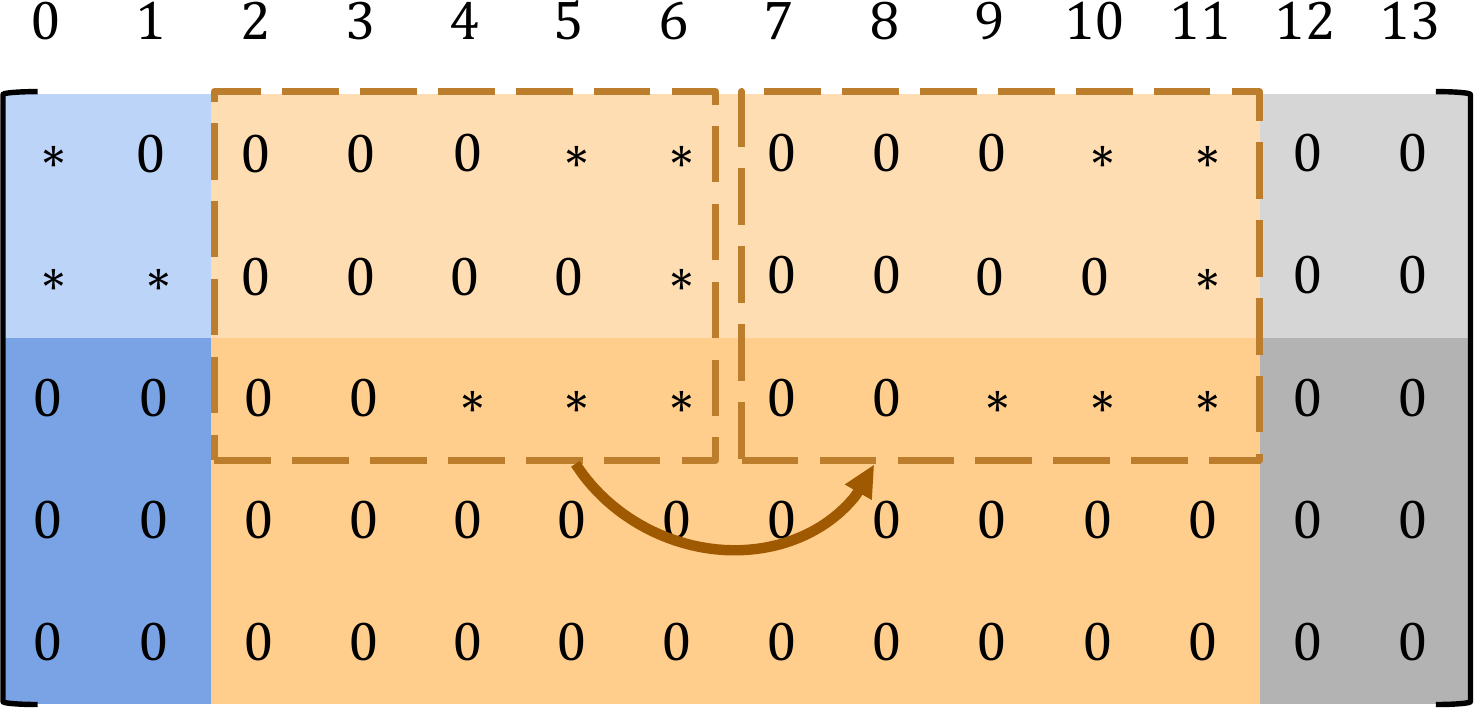}
			\caption{$H$ after Step-b for parameters \params\ $=\{3,5,11\}$. Here a portion of the p-c matrix is replicated.}
			\label{fig:construction_D_step2}
		\end{figure}
		
		\item {\it Step-c}: Set all the entries of $H(2,[2:3]\cup[7:8])$ to be $1$ (see Fig. \ref{fig:construction_D_step3}). Note that as each column of $H(0:2,[2:3]\cup[7:8])$ is a scalar multiple of $H(0:2,4)$, it is still true that there does not exist a set of $\leq 2$ columns of $H(0:2,[0:13]\setminus\{i\})$ which has $H(0:2,i)$ in its span, for $i=0,1$. From here on, the submatrix $H(0:2,0:13)$ will be unchanged.
		
		\begin{figure}[!htb]
			\centering
			\captionsetup{justification=centering}
			\includegraphics[scale=0.55]{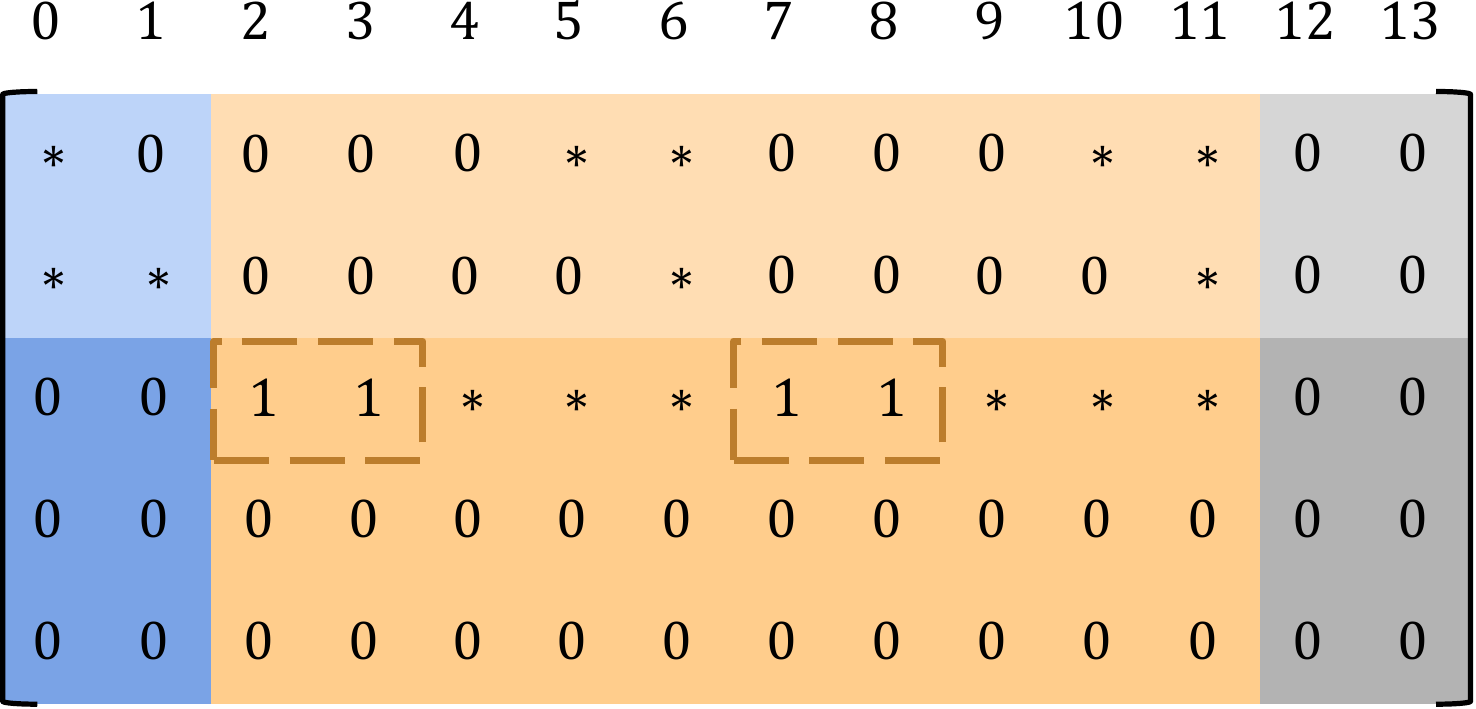}
			\caption{$H$ after Step-c for parameters \params\ $=\{3,5,11\}$. All the $0$'s of $H(2,2:11)$ are replaced with $1$'s.}
			\label{fig:construction_D_step3}
		\end{figure}
		
		\item {\it Step-d}: Let $C$ be a $(3\times 10)$ Cauchy-like matrix over $\mathbb{F}_{2^4}$ with $C(0,:)$ same as $H(2,2:11)$. Such a Cauchy-like matrix always exist, as columns of any given $(3\times 10)$ Cauchy-like matrix $C'$ can be scaled to arrive at $C$. Set $H(2:4,2:11)=C$. Also, set $H(3,12)=H(4,13)=1$ (see Fig. \ref{fig:construction_D_step4}). This completes the construction.
		
		\begin{figure}[!htb]
			\centering
			\captionsetup{justification=centering}
			\includegraphics[scale=0.55]{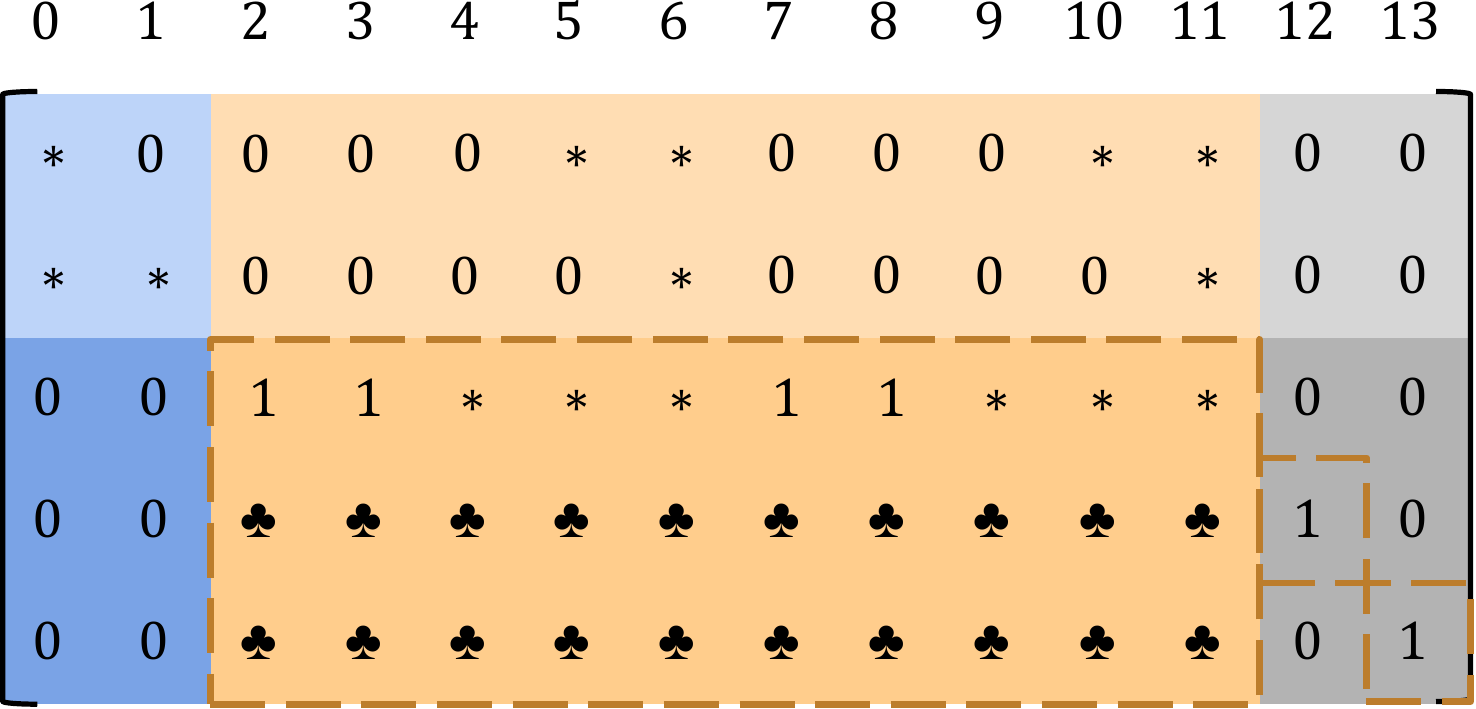}
			\caption{$H$ after Step-d for parameters \params\ $=\{3,5,11\}$. The submatrix $H(2:4,2:11)$ is set to be a Cauchy-like matrix $C$. Here $C(0,:)$ is such that $H(2,2:11)$ is unchanged from Step-c. Also, $0$'s at $H(3,12)$ and $H(4,13)$ are replaced with $1$'s.}
			\label{fig:construction_D_step4}
		\end{figure}
		\eit
	
		 \end{example}

		\begin{thm}\label{thm:constr_D}
				The code $\mathcal{C}$ over $\mathbb{F}_{q}$ having p-c matrix $H$ based on Construction D, when employed as the diagonally-embedded block code, will yield a rate-optimal streaming  code \cstr.
		\end{thm}
		\begin{proof}
		As in the case of other three constructions, we prove this by showing that $H$ meets all the conditions; B1, B2, R1 and R2.
		\bit
		\item[]
		\item {\it Recovery from burst erasure of length $\leq b$}: 
		\bit
		\item Condition B1: 
		Similar to all the previous constructions, for $\ell\in[0:\delta-1]$, the $\ell$-th row of the p-c matrix $H$ takes the form:
		
		\bean
		\hrowl\ & = & [\underbrace{\triangle \cdots \triangle}_{(\ell-1) \text{ symbols}} \ \ast \ \underbrace{0 \cdots 0}_{(b-1) \text{ symbols}} \ \triangle\cdots\triangle \ \underbrace{0 \cdots 0}_{\text{last } (\delta-\ell) \text{ symbols}} ],
		\eean
		 where $\triangle$'s indicate elements over \fq\ and $\ast$ indicates a non-zero element over \fq. Hence condition B1 follows.
		
		\item Condition B2: For $\delta \le \ell \le (\tau-a+1)$, let $\pl \triangleq H(:,\ell:\ell+b-1)$. We make the following two observations:
		\ben
		\item From Step-b, we have that $\pl(0:\delta-1,:)\triangleq H(0:\delta-1,\ell:\ell+b-1)$ is identical to $H(0:\delta-1,\delta:b+\delta-1 )$ up to column permutations.
		\item $\pl(\delta:b-1,:)\triangleq H(\delta:b-1,\ell:\ell+b-1)$ is an $a\times b$ MDS matrix.
		\een
		From these observations, it can be inferred that up to a permutation of columns, \pl\ has the form:
		\bean
		\left[ \begin{array}{cc} [0] & P_{\ell,3} \\ P_{\ell,2} & P_{\ell,4} \end{array} \right]. 
		\eean
		Using 1) we have that $P_{\ell,3}$ is a $\delta\times\delta$ matrix containing $\delta$ non-zero columns of $H(0:\delta-1,\delta:b+\delta-1)$. Consider the matrix $P\triangleq H(0:\delta-1,0:b+\delta-1)$, which is unchanged after Step-a. Recall that the rows of $P$ are codewords of a $[b+\delta,b]$ MDS code \cmds. Applying Lemma \ref{lem:shortened_mds}, any $\delta$ non-zero columns of $P$ (and thus of $H(0:\delta-1,\delta:b+\delta-1)$) form an independent set. Hence $P_{\ell,3}$ is invertible. From 2), we know that $P_{\ell,2}$ is an $a\times a$ matrix composed of $a$ columns of the MDS matrix $\pl(\delta:b-1,:)$, and hence is invertible. Thus \pl\ is invertible. In Fig. \ref{fig:construction_D_BE_proof_1}, we illustrate an example case of $\ell=6$.
		
		\begin{figure}[!htb]
			\centering
			\captionsetup{justification=centering}
			\includegraphics[scale=0.55]{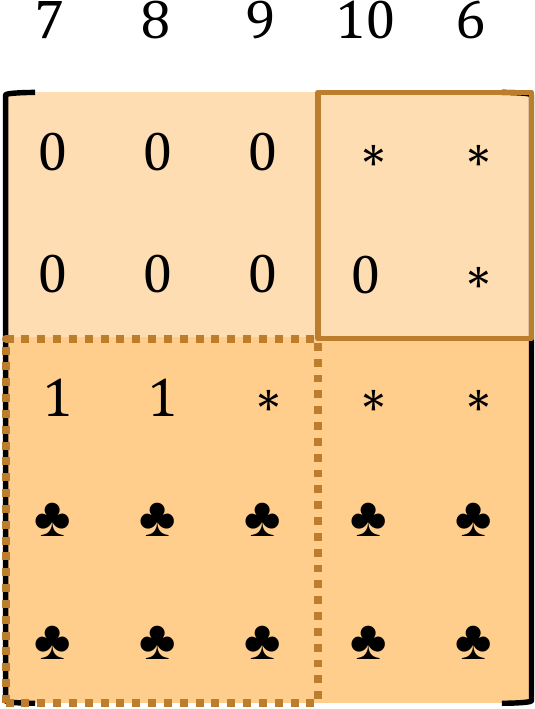}
			\caption{Consider the p-c matrix given in Example \ref{eg:constrn_D_example}. The figure shows a column-permuted version of \pl, where $\ell=6$. Here the submatrices demarcated by solid and dotted lines are both invertible. Thus \pl\ is invertible.}
			\label{fig:construction_D_BE_proof_1}
		\end{figure}

		\eit
		\item {\it Recovery from $\le a$ random erasures}:
		\bit
		\item Condition R1: Fix $\ell\in[0:\delta-1]$. For any $i$ from the set $R\triangleq [0:\delta]$, the $i$-th row $H(i,0:\ell+\tau)$ can be verified to belong to the row space of the shortened p-c matrix \hsupl. This is because each row $H(i,:)$ has a run of $(\delta-\ell)$ zeros across columns $[\ell+\tau+1:n-1]$.  Hence while discussing the recoverability of $\ell$-th code symbol, we will restrict ourselves to the rows $R$. In the following, we make three observations which help us show that condition R1 is satisfied by the p-c matrix $H$. These are generalizations of the observations that we make in Steps a, b and c of Example \ref{eg:constrn_D_example}. Consider the p-c matrix $H_a$ obtained after the refinement in Step-a. The first observation is as follows:
       \ben
	     \item[(i)] The matrix $H_a(R,0:b+\delta-1)$ has a rank of $(\delta+1)$ and  also $|\cup_{j=0}^{\delta}\text{supp}(H_a(j,0:b+\delta-1))|=2\delta+1$. Hence applying Lemma \ref{lem:shortened_mds}, we have that any $(\delta+1)$ non-zero columns of $H_a(R,:)$ form an independent set.
        \een 
		
		Now consider the p-c matrix $H_b$ obtained after the refinement in Step-b. The only change with respect to $H_a$ is that columns of $H_a(R,\delta:b+\delta-1)$ are replicated $\gamma$ times in $H_b$. We make a second observation as follows:
		\ben
		\item[(ii)]  From the observation (i), it follows that there does not exist a set of $\leq \delta=(a-1)$ columns of $H_b(R,[0:n-1]\setminus\{i\})$ which has $H_b(R,i)$ in its span, for any $i\in[0:\delta-1]$.
		\een
		Let $H_c$ be the p-c matrix obtained after Step-c, obtained from $H_b$ by replacing all the zero entries of $H_b(\delta,\delta:\tau)$ with $1$'s. For each entry $H_b(\delta,j)$ thus changed, $H_c(R,j)$ takes the form $[\underbrace{0\ \cdots\ 0}_{\delta}\ 1]^T$. This vector is a scalar multiple of $H_b(R,b-1)$. Hence observation (ii) can be extended to $H_c$ as well. Thus we have:
		 \ben
		\item[(iii)]  There does not exist a set of $\leq \delta=(a-1)$ columns of $H_c(R,[0:n-1]\setminus\{i\})$ which has $H_c(R,i)$ in its span, for any $i\in[0:\delta-1]$.
		\een
		As $H(R,:)$ is unchanged after Step-c, from observation (iii), condition R1 follows. 
		\item Condition R2: The condition R2 is clearly satisfied by the p-c matrix, as $H(\delta:b-1,\delta:n-1)$ is an $a\times (\tau+1)$ MDS matrix. 
		
		\eit
		\eit 
		
	\end{proof}
		
		\section{Rate-Optimal Convolutional Codes for Given Column Distance and Column Span}\label{sec:conv_codes_col_distance_col_span}
		In \cite{FongKhisti}, the authors observe that the rate-optimal streaming codes (which are convolutional codes) they construct for the DC-SW channel, are also rate-optimal convolutional codes with respect to column span and column distance. Here we extend this observation to the constructions presented in this paper.
		
		Consider a rate $\frac{k}{n}$ convolutional code with memory $m$. The relation between input vectors $\{\underline{s}(t)\}$  and $\{\underline{x}(t)\}$ is given by:
		 \beq \label{eq:G_conv}
		 \underline{x}^T(t) = \sum_{i=0}^{m} \underline{s}^T(t-i)G^\text{\tiny conv}_i,
		 \eeq
		 where $\underline{s}(t)\in \fq^{k\times 1}$, $\underline{x}(t)\in \fq^{n\times 1}$  and $G^\text{\tiny conv}_i\in\fq^{k\times n}$. 
		 
		 We borrow the following definitions from \cite{BadrPatilKhistiTIT17,FongKhisti}:
		
		 	\begin{align*}
		 	\text{Column distance},\ d_{\tau} & \triangleq \min\{\mathrm{wt}(\underline{x}(0), \underline{x}(1),\dotsc , \underline{x}(\tau)):\underline{s}(0)\ne \underline{0}\}\\
		 	\text{Column span},\ c_{\tau}&  \triangleq \min\{\mathrm{span}(\underline{x}(0), \underline{x}(1),\dotsc , \underline{x}(\tau)):\underline{s}(0)\ne \underline{0}\},
		 	\end{align*}
		 	where $\mathrm{wt}(\underline{x}(0), \underline{x}(1)\dotsc , \underline{x}(\tau))$ is the number of non-zero vectors in $\{\underline{x}(i)\}_{i=0}^{\tau}$ and
		 	$\mathrm{span}(\underline{x}(0), \underline{x}(1)\dotsc , \underline{x}(\tau))= \max\{i\mid\underline{x}(i)\ne\underline{0}\}-\min\{i\mid\underline{x}(i)\ne\underline{0}\}+1$. Clearly, $d_\tau\leq c_\tau\leq (\tau+1)$.

		   It is shown in \cite{BadrPatilKhistiTIT17} that a convolutional code with column distance, $d_\tau$ and column span, $c_\tau$ is a streaming code which can correct, with a delay $\tau$, all the erasure patterns of the DC-SW channel having parameters $\{a=d_\tau-1,b=c_\tau-1,\tau,w=\tau+1\}$. Thus, from \eqref{eq:rate_upper_bound}, it follows that:
		   \beq\label{eq:col_weight_span_bound}
		    \frac{k}{n}\leq \frac{\tau-d_\tau+2}{\tau-d_\tau+c_\tau+1}.
		   \eeq
		   
		   Conversely, for a streaming code (which also is a convolutional code, as in the case of our constructions) which can recover with a delay $\tau$ from all the erasure patterns of the $\{a,b,\tau,w=\tau+1\}$ DC-SW channel, it is shown in \cite{BadrPatilKhistiTIT17} that $d_\tau\geq (a+1)$ and $c_\tau\geq (b+1)$. Thus, for the rate-optimal streaming codes obtained via diagonally embedding Construction A, B, C or D, we have:
		   \bea \label{eq:col_weight_span_achievability}
		    \frac{k}{n}& = &\frac{\tau-a+1}{\tau-a+1+b}\nonumber\\
		    & \geq &\frac{\tau-d_\tau+2}{\tau-d_\tau+2+b}\nonumber\\
		    & \geq &\frac{\tau-d_\tau+2}{\tau-d_\tau+c_\tau+1} 
		   \eea 
		   A convolutional code having column distance, $d_\tau$ and column span, $c_\tau$ is defined to be rate-optimal, if it satisfies \eqref{eq:col_weight_span_bound} with equality. The following theorem is a direct consequence of Theorem \ref{thm:constr_A} and inequalities \eqref{eq:col_weight_span_bound}, 
		   \eqref{eq:col_weight_span_achievability}.  
		
		\begin{thm}
			For any $d_\tau,c_\tau$ and $\tau$ such that $d_\tau\leq c_\tau \leq (\tau+1)$, there exists a rate-optimal convolutional code \cstr\ with column distance $d_\tau$ and column span $c_\tau$, over $O(\tau^2)$ field-size. 
		\end{thm}

		\section{Numerical Evaluation}\label{sec:simulations}
		In this section, we study the performance of two of our proposed diagonally-embedded constructions (A and C), the random convolutional code appearing in \cite{FongKhisti}, the rate-optimal burst erasure correcting code of \cite{MartTrotISIT07} and a diagonally-embedded MDS code. Constructions A, C and the random convolutional code have parameters $\{a,b,\tau\} = \{4,8,11\}$. The optimal burst erasure correcting code has parameters $\{b,\tau\} = \{11,11\}$ and the MDS code has parameters $[n=12,k=6]$. Note that, with the chosen parameters, the code rate, $R=0.5$ and delay parameter, $\tau=11$ are the same for all the five code constructions. With regard to the field-size requirements, Construction A, random convolutional code and Construction C are over $\mathbb{F}_{2^8}$, $\mathbb{F}_{2^{10}}$ and  $\mathbb{F}_{2^3}$, respectively. The MDS code is over $\mathbb{F}_{2^4}$ and as $R=0.5$, the burst erasure correcting code turns out to be a repetition-code-based scheme needing just $\mathbb{F}_2$.  The simulations are performed over  Gilbert-Elliott and Fritchman channels. Each data point is a result of $10^8$ simulations. Note that the parameters $\{a,b,\tau\} = \{4,8,11\}$ that we have chosen, lie outside the permitted parameter ranges for Constructions B and D.
		\subsection{The Gilbert-Elliott Channel}
		The Gilbert-Elliott (GE) channel is a Markov model consisting of a good state and a bad state. The model is characterized by the tuple  $(\alpha,\beta,\epsilon)$. Here $\alpha$ and $\beta$ are the transition probabilities from the good state to the bad state and vice versa, respectively. In the good state, a packet is lost with probability $\epsilon$, i.e., the channel behaves as a binary erasure channel with erasure probability $\epsilon$. All packets transmitted while the channel is in the bad state, are lost. 
		
		We perform simulations over GE channels with $\alpha = 5\times 10^{-4},\beta =0.5$ and $\epsilon$ varying from 0.001 to 0.04.  Fig. \ref{fig:geplot} illustrates the performance of all the five coding schemes. As expected, the burst-erasure correcting code performs the best for very low $\epsilon\ (<0.003)$,  when the channel behavior is dominated by burst erasures. On the other end of the spectrum, the MDS code starts outperforming all the other codes when $\epsilon$ is large enough and random erasures become the dominant factor. In the intermediate range, despite smaller field-size requirements, Construction A and Construction C perform better than the random convolutional code.

		\begin{figure}
			\centering
			\includegraphics[scale=0.6]{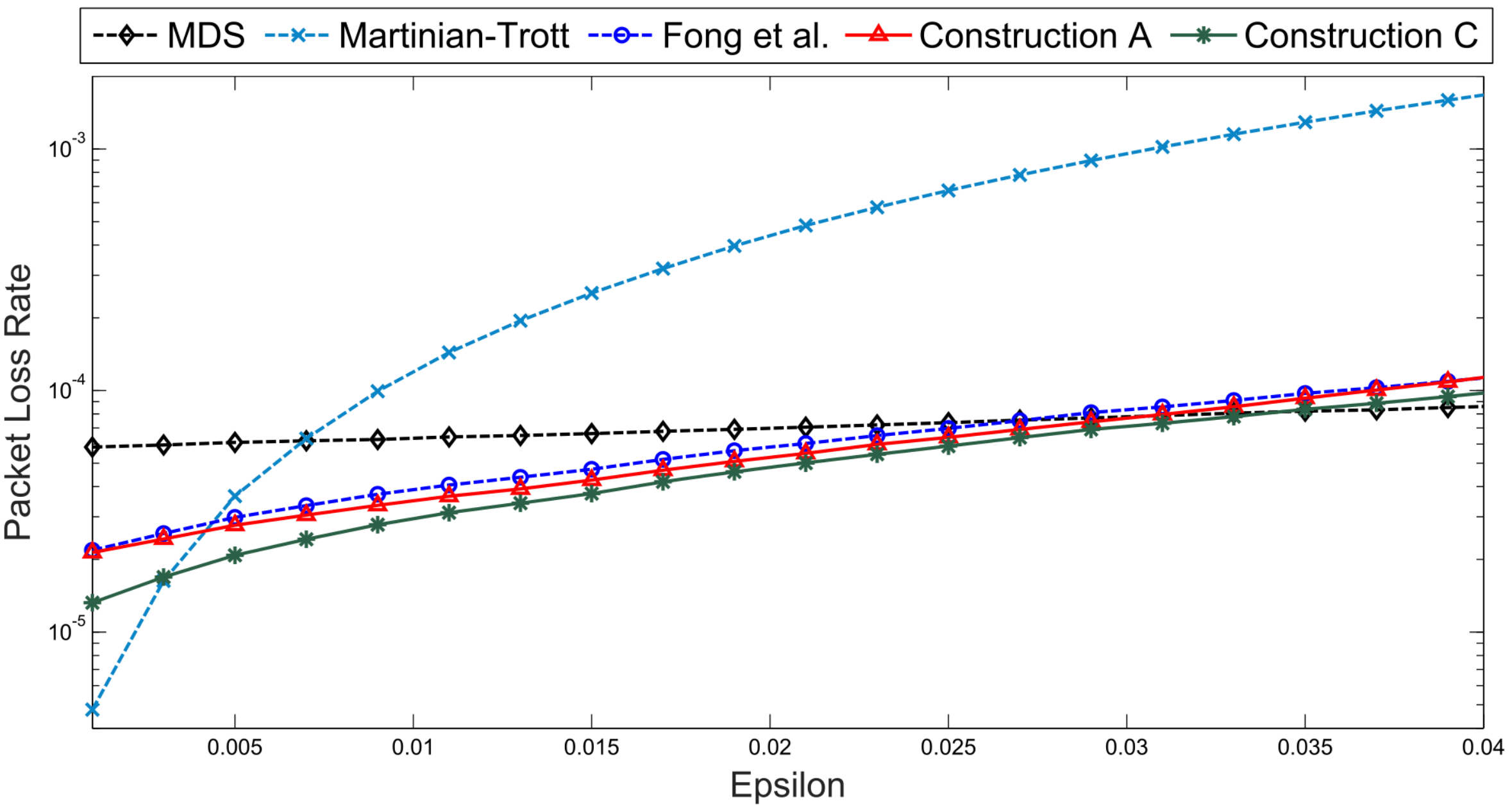}
			\caption{We consider the GE channel with parameters $\alpha = 5\times10^{-4},\beta = 0.5$ and $\epsilon$ varying from $0.001$ to $0.04$. The packet-loss probabilities of five coding schemes having rate, $R=0.5$ and delay parameter, $\tau=11$, are plotted; a diagonally-embedded $[12,6]$ MDS code, Martinian-Trott code \cite{MartTrotISIT07} for burst erasures having parameters $\{b=11,\tau=11\}$, Fong et al. code \cite{FongKhisti} for parameters  $\{a,b,\tau\} = \{4,8,11\}$ and two codes from the present paper; Constructions A and C, again, for parameters  $\{a,b,\tau\} = \{4,8,11\}$. }
			\label{fig:geplot}
		\end{figure}

		\subsection{The Fritchman Channel}
		The Fritchman Channel is a generalization of the two-state GE model. It is characterized by parameters $(\alpha,\beta,\epsilon,M)$. It consists of one good state $G$ and $M$ bad states, $E_1,\dotsc,E_M$. In the good state, the channel behaves as a Binary Erasure Channel with erasure probability $\epsilon$. All the packets transmitted while the channel is in any bad state, are lost with probability 1. Transitions between these states are governed by the following 3 rules:
		\ben
		\item If the channel is in the good state, it will remain in the same state with probability $(1-\alpha)$, or transition to $E_1$ with probability $\alpha$, in the next time slot. 
		\item If the channel is in a state $E_l,l\in [M-1]$, it will remain in the same state with probability $(1-\beta)$, or transition to $E_{l+1}$ with probability $\beta$, in the next time slot. 
		\item  If the channel is in state $E_M$, it will remain in the same state with probability $(1-\beta)$, or transition to $G$ with probability $\beta$, in the next time slot. 
		\een
		
		Figure \ref{fig:fritchman_state} shows the state transition probabilities for a Fritchman channel with $M = 4$. Note that the GE channel is a special case of Fritchman channel, where $M=1$. We perform simulations over Fritchman channels with $\alpha = 10^{-4},\beta =0.75, M=4$ and $\epsilon$ varying from $0.004$ to $0.05$. Fig. \ref{fig:fritchman_plot} illustrates the performance of these coding schemes, where the trends are similar to that of the GE channel.
		\begin{figure}[!htb]
			\centering
			\captionsetup{justification=centering}
			\includegraphics[scale=0.7]{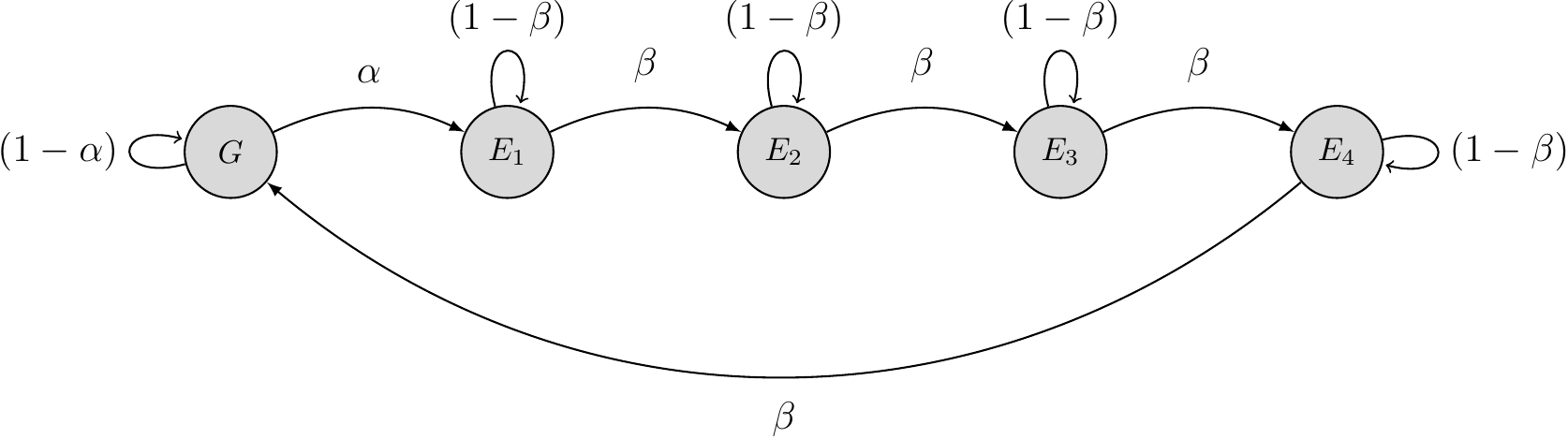}
			\caption{State transition diagram for the Fritchman channel with $M=4$ bad states. }
			\label{fig:fritchman_state}
		\end{figure}

	\begin{figure}
		\centering
		\includegraphics[scale=0.6]{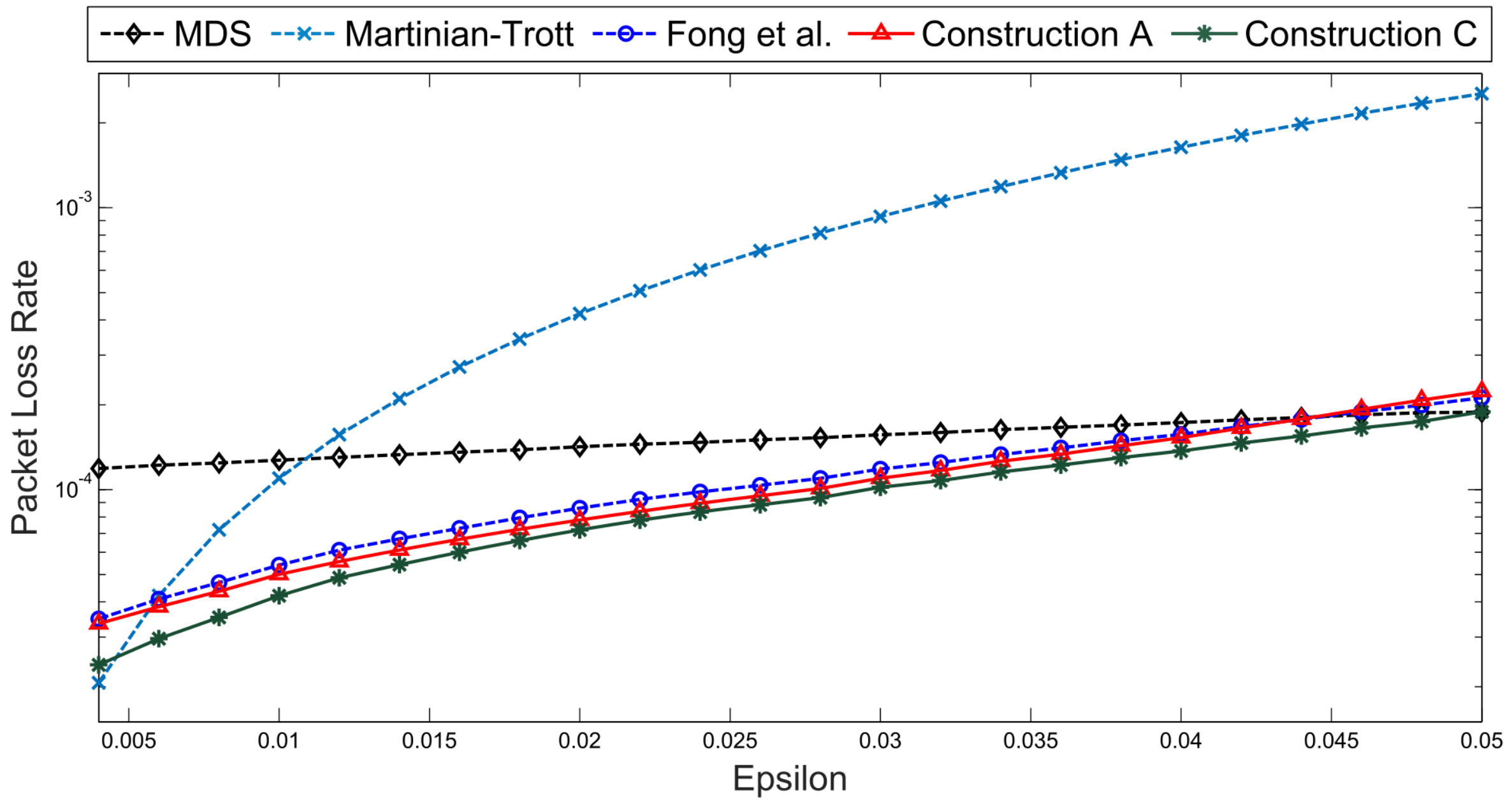}
		\caption{Here we consider the Fritchman channel with $5$ states, which is a generalization of the GE channel. The channel parameters are $\alpha =  10^{-4},\beta = 0.75, M=4$ and $\epsilon$ varying from $0.004$ to $0.05$. The five code constructions that are being compared, are the same as that we use in the case of GE channel. }
		\label{fig:fritchman_plot}
	\end{figure}
		
	\begin{remark}\normalfont
		As Construction A and the random convolutional code are not completely explicit, these simulations have been done on particular realizations of these constructions. The complete set of  Matlab codes used for these simulations, are available at: \texttt{https://github.com/deeptanshu04/StreamingCodes}. 
	\end{remark}	
		\appendices
		\section{Proof of Theorem \ref{thm:constr_A} (Construction A)}\label{app:proof_constrn_a} 
		\bit
		\item[]
		\item {\it Recovery from burst erasures of length $\leq b$}:  
		\bit
		\item Condition B1:   In order to establish that condition B1 is satisfied, we have to show that if there is a burst of $b$ erasures, where the burst begins with code symbol $\ell$, $\ell \in [0:\delta-1]$, we can recover the $\ell$-th symbol by accessing symbols belonging to the coordinates $[0:\ell-1]\cup[\ell+b:\ell+\tau]$. The $\ell$-th row \hrowl\ of the p-c matrix $H$ takes the form:
		
		\bean
		\hrowl\ & = & [\underbrace{0 \cdots 0}_{(\ell-1) \text{ symbols}} \ \alpha \ \underbrace{0 \cdots 0}_{(b-1) \text{ symbols}} \ \triangle\cdots \triangle\ \underbrace{0 \cdots 0}_{\text{last } (\delta-\ell) \text{ symbols}} ],
		\eean
		where $\triangle$'s denote field elements drawn from \fq.	It is clear that using this row, one can recover the $\ell$-th code symbol, when there is a burst of $b$ erasures starting with symbol $\ell$.
		%		\bean
		%		\hrowi\ & = & [\underbrace{0 \cdots 0}_{(\ell-1) \text{ symbols}} \ \alpha \ \underbrace{0 \cdots 0}_{(b-1) \text{ symbols}} \triangle \cdots \triangle \underbrace{0 \cdots 0}_{\text{last } (\delta-\ell) \text{ symbols}} ].
		%		\eean
		%	\begin{figure}[!htb]
		%		\centering
		%		\captionsetup{justification=centering}
		%		\includegraphics[scale=0.45]{fig_construction_A_BE_1}
		%		\caption{The $\ell$-th row of $H$ takes the form shown in this figure. Here $\triangle$'s denote field elements drawn from $\mathbb{F}_q$.}
		%		\label{fig:construction_A_h_row}
		%	\end{figure}
		
		\item Condition B2:  For $\delta \leq \ell \leq (\tau-a+1)$, let \pl\ denote the $(b\times b)$ square submatrix of $H$ corresponding to the column set $[\ell:\ell+b-1]$, i.e., 
		\bean
		\pl\ & \triangleq & H(:, [\ell:\ell+b-1]). 
		\eean
		In order to establish that the code is capable of recovering from any burst erasure of $b$ erasures occurring within coordinates $[\delta:n-1]$, it is sufficient to show that there exists an assignment of the variables \vij\ for which all the \plset\ have non-zero determinants.   For showing this, we begin by noting that $\det(\pl)$ for each $\ell$, is a multivariate polynomial in the variables \vij, which is of degree at most $1$ in each of the variables.  It follows that the product $\prod_{\ell} \det(\pl)$ is a multivariate polynomial that is of degree at most $ (\tau-a+1-\delta)+1= (\tau-b+2)$ in each variable.  If one can prove that each determinant is a non-zero polynomial, we will have established that the product is also a non-zero polynomial and we can then apply the Combinatorial Nullstellensatz (Lemma \ref{lem:comb_null}) to show that over a field of size $q > (\tau-b+2)$, there exists an assignment of the variables for which the product determinant is non-zero.  
		
		For showing that each individual determinant $\det(\pl)$ is a non-zero polynomial, we will make an  assignment of the \vij, specific to each $\ell$, for which $\det(\pl)$ is non-zero.  The specific assignment which we will use, is to set a selected `diagonal band' of the \vij\ to have the value $1$ and assign the value zero to the remaining \vij.  This is illustrated in Fig.~\ref{fig:construction_A_example_2}, for the cases $\ell=3,7,8$, with respect to the example p-c matrix $H$ shown in Fig. \ref{fig:construction_A_example}.  As can be seen from the Fig. \ref{fig:construction_A_example_2}, this assignment causes the resultant $(b \times b)$ matrix \pl, for $\ell\in[\delta:(\tau-a)]$, to be of the form: 
		\bea\label{eq:pl_case1}
		P_{\ell } & = & \left[ \begin{array}{cc} [0] & P_{\ell,3} \\ P_{\ell,2} & P_{\ell,4} \end{array} \right], 	
		\eea
		where both submatrices $P_{\ell,2}$ and $P_{\ell,3}$ are non-singular. It follows that:  
		\bean
		\det(P_\ell) & = & \det(P_{\ell,2}) \det(P_{\ell,3}) \ \neq \ 0. 
		\eean
		
		For the remaining case $\ell=(\tau-a+1)$, $P_{\ell}$ initially takes on the form: 
		\bea\label{eq:pl_case2}
		\left[ \begin{array}{cc} P_{\ell,1} & P_{\ell,3} \\ P_{\ell,2} & P_{\ell,4} \end{array} \right], 	
		\eea
		which can be reduced using elementary column operations, to the form:
		\bea\label{eq:pl_col_reduced_form}
		\left[ \begin{array}{cc} [0] & Q_{\ell,3} \\ Q_{\ell,2} & Q_{\ell,4} \end{array} \right], 	
		\eea
		where both $Q_{\ell,2}$ and $Q_{\ell,3}$ are non-singular. In Fig. \ref{fig:construction_A_example_3}, we illustrate the resultant matrix corresponding to $P_8$ (since $\ell=(\tau-a+1)=8$ in this example) in Fig.\ref{fig:construction_A_example_2}, following the elementary column operations.  
		
		Having outlined the assignment with respect to the example p-c matrix $H$ given in Fig. \ref{fig:construction_A_example}, we provide below the explicit assignments for the general case, followed by a justification as to why the relevant submatrices appearing in \eqref{eq:pl_case1} and \eqref{eq:pl_col_reduced_form} are non-singular. 
		\begin{figure}[!htb]
			\centering
			\captionsetup{justification=centering}
			\includegraphics[scale=0.6]{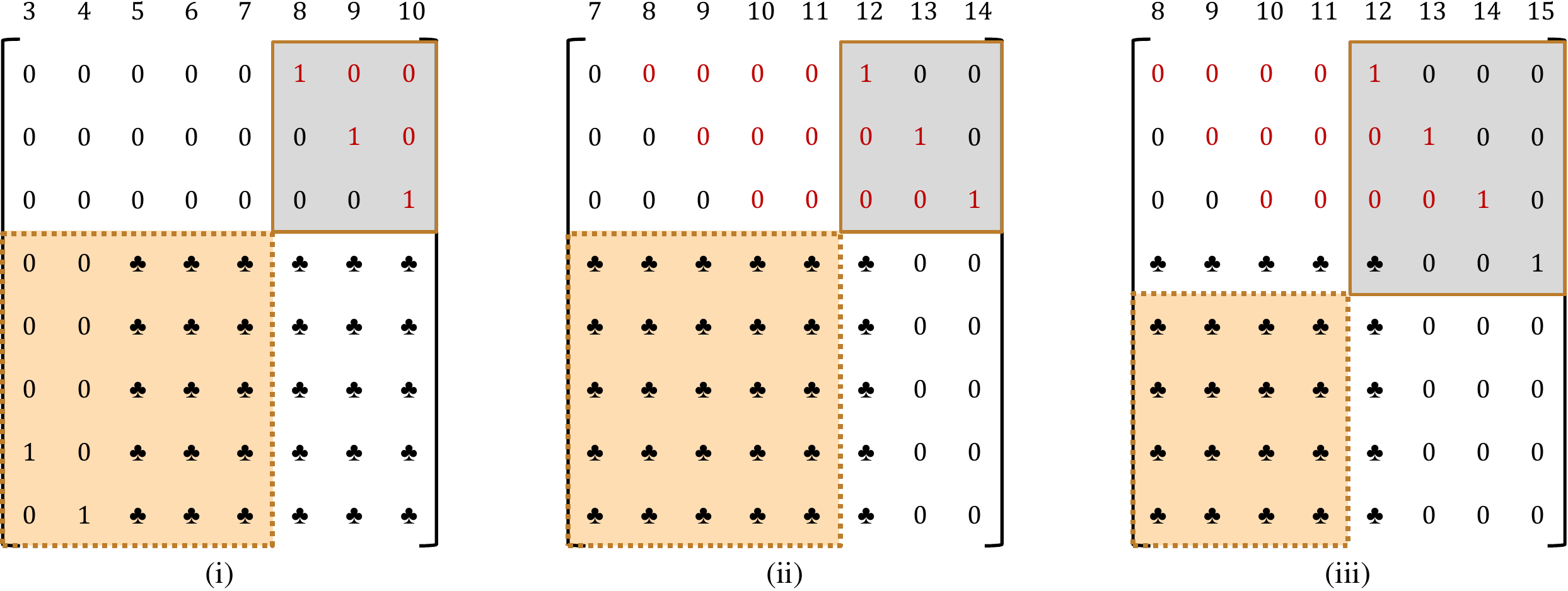}
			\caption{Three $(b \times b)$ square submatrices of the p-c matrix in Figure \ref{fig:construction_A_example}. The red-colored entries denote an assignment of values to the \vij. The assignments ensure respectively, that all determinants, $\det(P_3),\det(P_7)$ and $\det(P_8)$ are non-zero. The submatrix demarcated by dotted lines in each of the figures (i)-(iii) is invertible, as it is a square submatrix of the matrix $H(3:7,0:12)$.  This is because the submatrix $H(3:7,0:12)$ is of the form $[I_5\ C]$, where $C$ is a Cauchy-like matrix. Each submatrix demarcated by solid lines is also invertible, as it corresponds to a lower-triangular matrix with $1$'s along the diagonal.  }
			\label{fig:construction_A_example_2}
		\end{figure}
		\begin{figure}[!htb]
			\centering
			\captionsetup{justification=centering}
			\includegraphics[scale=0.6]{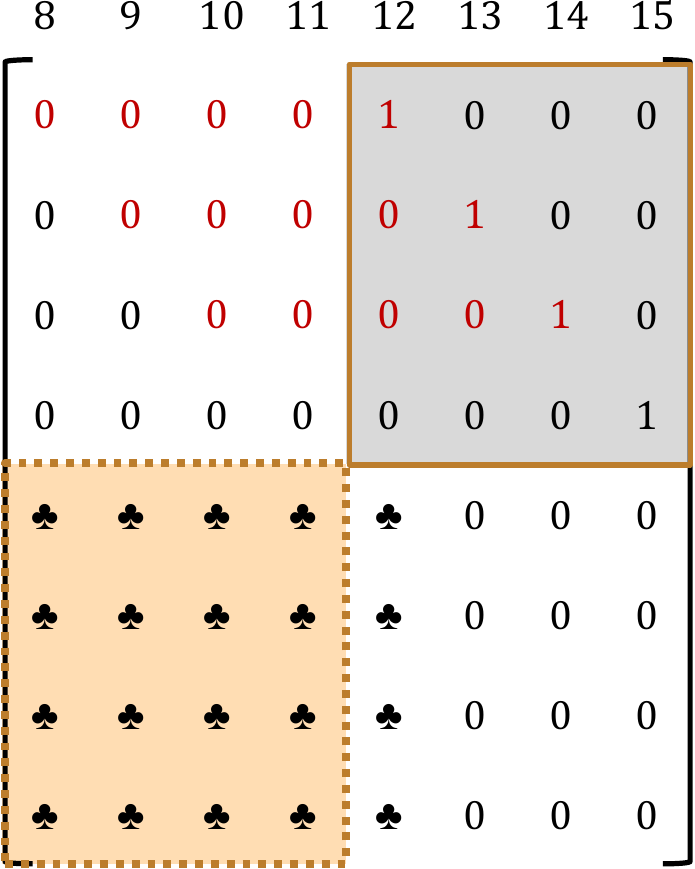}
			\caption{The resultant matrix after performing column operations on $P_8$.  }
			\label{fig:construction_A_example_3}
		\end{figure}
		When $\ell\in[\delta:(\tau-a)]$, for $0\leq i\leq (\delta-1)$ and $b+i\leq j\leq \tau+i$, we make the assignment:
		\bea\label{eq:assign_1}
		v_{i,j} & = & \left\{ \begin{array}{rl} 1, & \text{if } j=(a+i+\ell),\\
			0, & \text{otherwise}. \end{array} \right.
		\eea
		
		For the remaining $\ell=(\tau-a+1)$ case, we make the assignment:
		\bea\label{eq:assign_2}
		v_{i,j} & = & \left\{ \begin{array}{rl} 1, & \text{if } j=(a-1+i+\ell), \\
			0, & \text{otherwise}. \end{array} \right.
		\eea
		
		\begin{figure}[!htb]
			\centering
			\captionsetup{justification=centering}
			\includegraphics[scale=0.6]{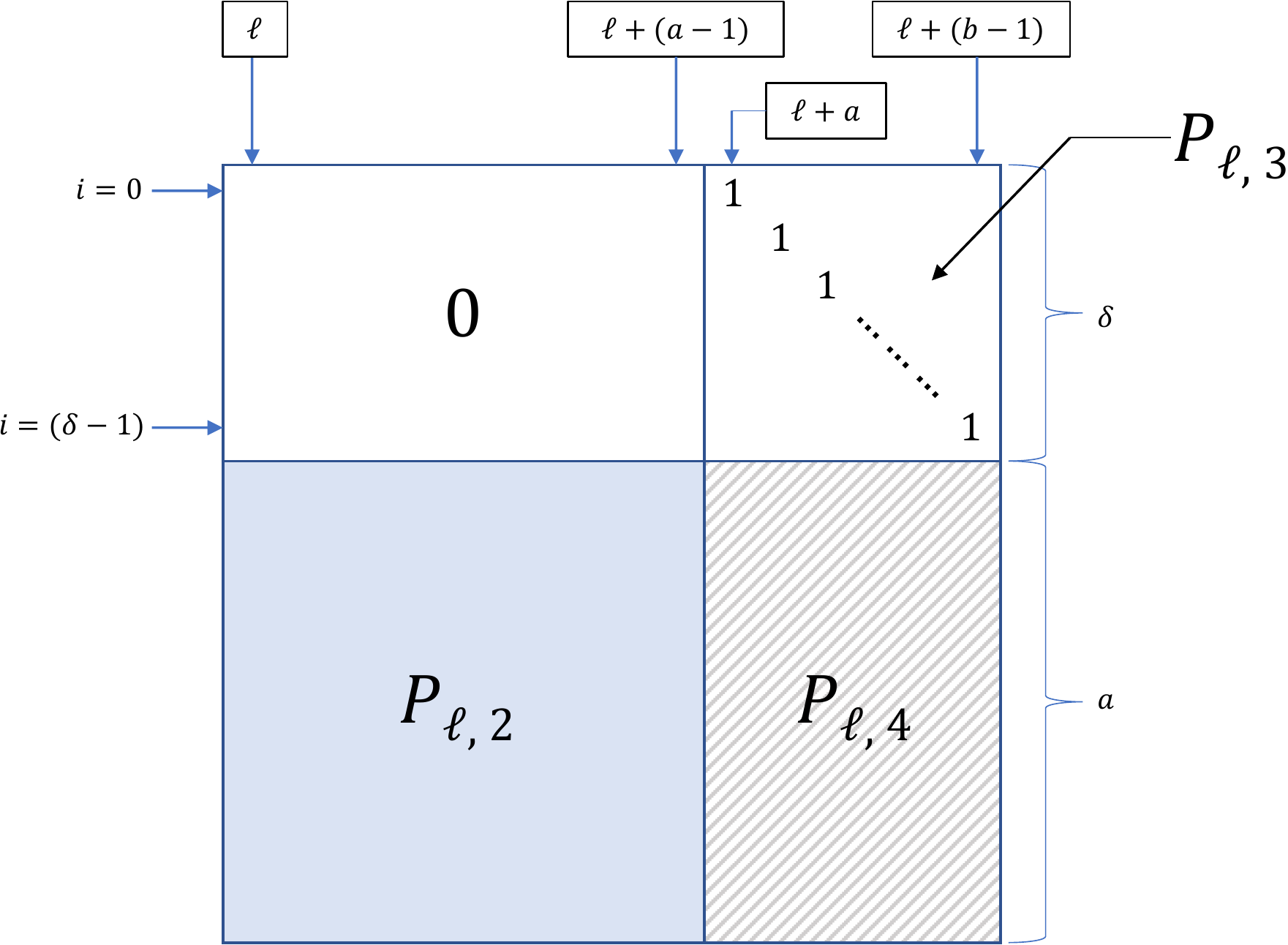}
			\caption{In this figure, we illustrate the structure of the matrix $H(:,[\ell:\ell+b-1])$ after the assignment \eqref{eq:assign_1}. Here $P_{\ell,2}$ is an $(a\times a)$ submatrix of an $(a\times(\tau+1))$ MDS matrix and hence is invertible.}
			\label{fig:construction_A_nz_poly_proof_1}
		\end{figure}

		\begin{figure}[!htb]
			\centering
			\captionsetup{justification=centering}
			\includegraphics[scale=0.6]{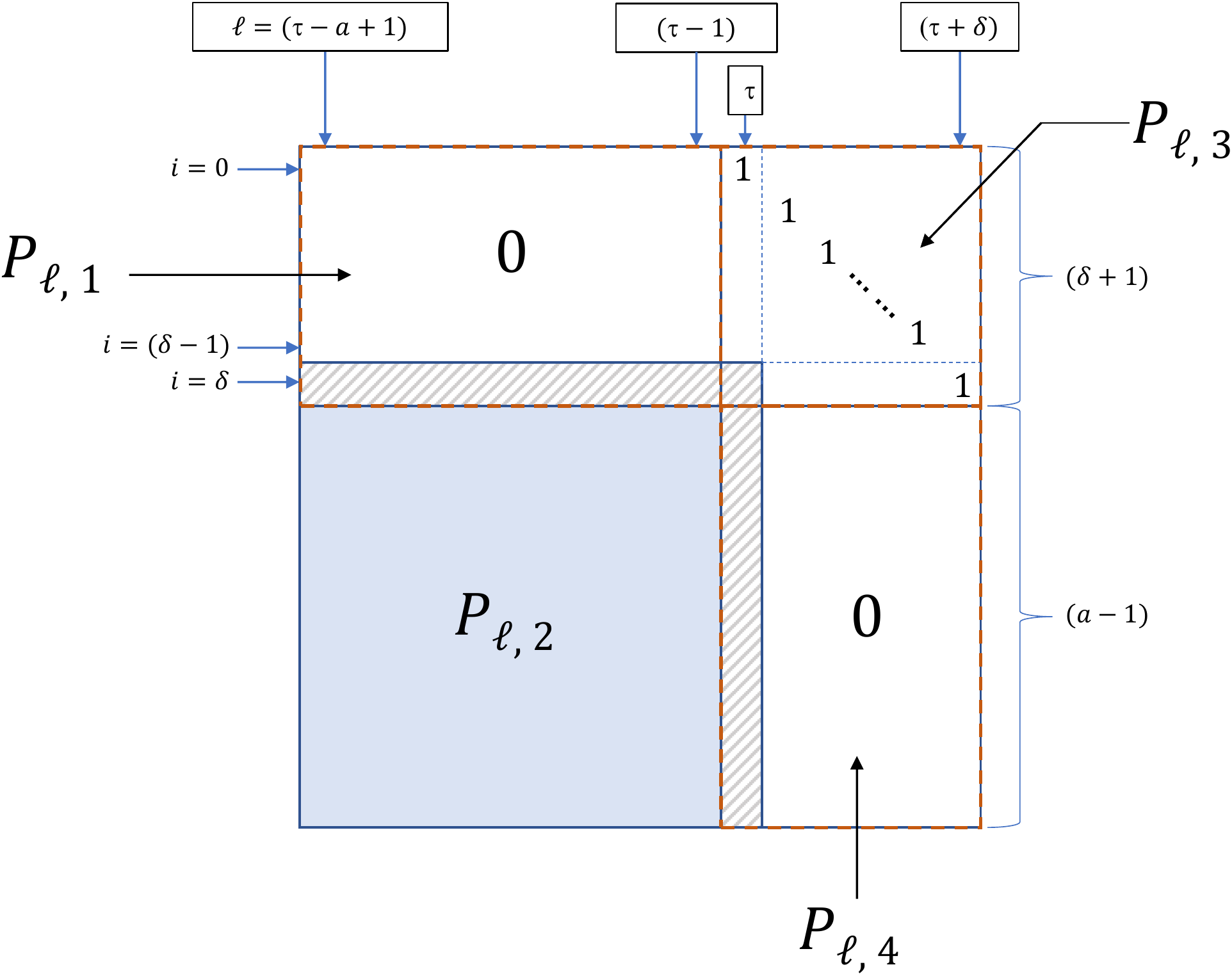}
			\caption{In this figure, we illustrate the structure of $H(:,[\ell:\ell+b-1])$, for $\ell=(\tau-a+1)$, after the assignment \eqref{eq:assign_2}. Here $P_{\ell,2}$ is an $(a-1)\times (a-1)$ submatrix of the $(a-1)\times \tau$ MDS matrix $H([\delta+1:b-1],[1:\tau])$ and is hence invertible.  We note that the non-zero entries appearing in $P_{\ell,1}$ and the non-zero entry in the bottom left corner of $P_{\ell,3}$ can be made $0$ via elementary column operations without altering any of the other entries in $P_\ell$.}
			\label{fig:construction_A_nz_poly_proof_2}
		\end{figure}
		For the general case, the submatrices $\{P_{\ell,i}\}_{i=1}^4$ appearing in \eqref{eq:pl_case1}, \eqref{eq:pl_case2} are as shown in Fig. \ref{fig:construction_A_nz_poly_proof_1} and Fig. \ref{fig:construction_A_nz_poly_proof_2}, for the cases $\ell\in[\delta:\tau-a]$ and $\ell=(\tau-a+1)$, respectively. It remains to formally prove invertibility of the matrices $\{P_{\ell,2},P_{\ell,3}\}$ appearing in \eqref{eq:pl_case1} and of the matrices $\{Q_{\ell,2},Q_{\ell,3}\}$ appearing in \eqref{eq:pl_col_reduced_form}.   We do this below. 
		
		For the case $\ell\in[\delta:\tau-a]$, it can be verified that $P_{\ell,2}\triangleq H(\delta:b-1,\ell:\ell+a-1)$ is an $(a\times a)$ submatrix of the $(a\times (\tau+1))$ MDS matrix $H(\delta:b-1,0:\tau)$ and is hence invertible. The matrix $P_{\ell,3}$, in this case, is the identity matrix $I_\delta$ and is thus invertible.
		
		Similarly, for $\ell=(\tau-a+1)$, $P_{\ell,2}\triangleq H(\delta+1:b-1,\tau-a+1:\tau-1)$ is an $(a-1)\times (a-1)$ submatrix of the  MDS matrix $H(\delta+1:b-1,1:\tau)$ of size $(a-1)\times \tau$ and is hence invertible. For the case $\ell=(\tau-a+1)$, the matrix $P_{\ell,3}$, as can be observed in Fig. \ref{fig:construction_A_nz_poly_proof_2}, takes on the form $I_{\delta+1}$ with an extra non-zero element in the bottom left corner. However applying a series of elementary column operations, one can make all the non-zero entries of the row $H(\delta,\tau-a+1:\tau+\delta)$ with the exception of the last entry, $H(\delta,\tau+\delta)$, to be $0$, without perturbing any of the other entries of \pl. Here the column operations can be described as follows. The last column of $P_{\tau-a+1}$  has a $1$ in row $\delta$ and zeros elsewhere.  We zero out all the other non-zero entries of the matrix $P_{\tau-a+1}$ in row $\delta$ by subtracting an appropriate scalar multiple of the last column. This will cause  \pl\ to be of the form \eqref{eq:pl_col_reduced_form}, where $Q_{\ell,2}=P_{\ell,2}$ and $Q_{\ell,3}=I_{\delta+1}$, both of which are invertible. 
		\eit
		\item {\it Recovery from $\leq a$ random erasures}:
		\bit
		\item Condition R1: Fix $\ell\in[0:\delta-1]$. For any $i$ from the set $\rl\triangleq [0:\ell]\cup [\delta+1:b-1]$, the $i$-th row $H(i,0:\ell+\tau)$ can be verified to belong to the row space of the shortened p-c matrix \hsupl. This is because each row $H(i,:)$ has a run of $(\delta-\ell)$ zeros across columns $[\ell+\tau+1:n-1]$ (see Fig. \ref{fig:construction_A_RE_1} for an example). Hence while discussing the recoverability of $\ell$-th code symbol, we will restrict ourselves to the rows $R_\ell$ (or its subsets) of the p-c matrix $H$. We divide the proof into two cases.
		\begin{figure}[!htb]
			\centering
			\captionsetup{justification=centering}
			\includegraphics[scale=0.6]{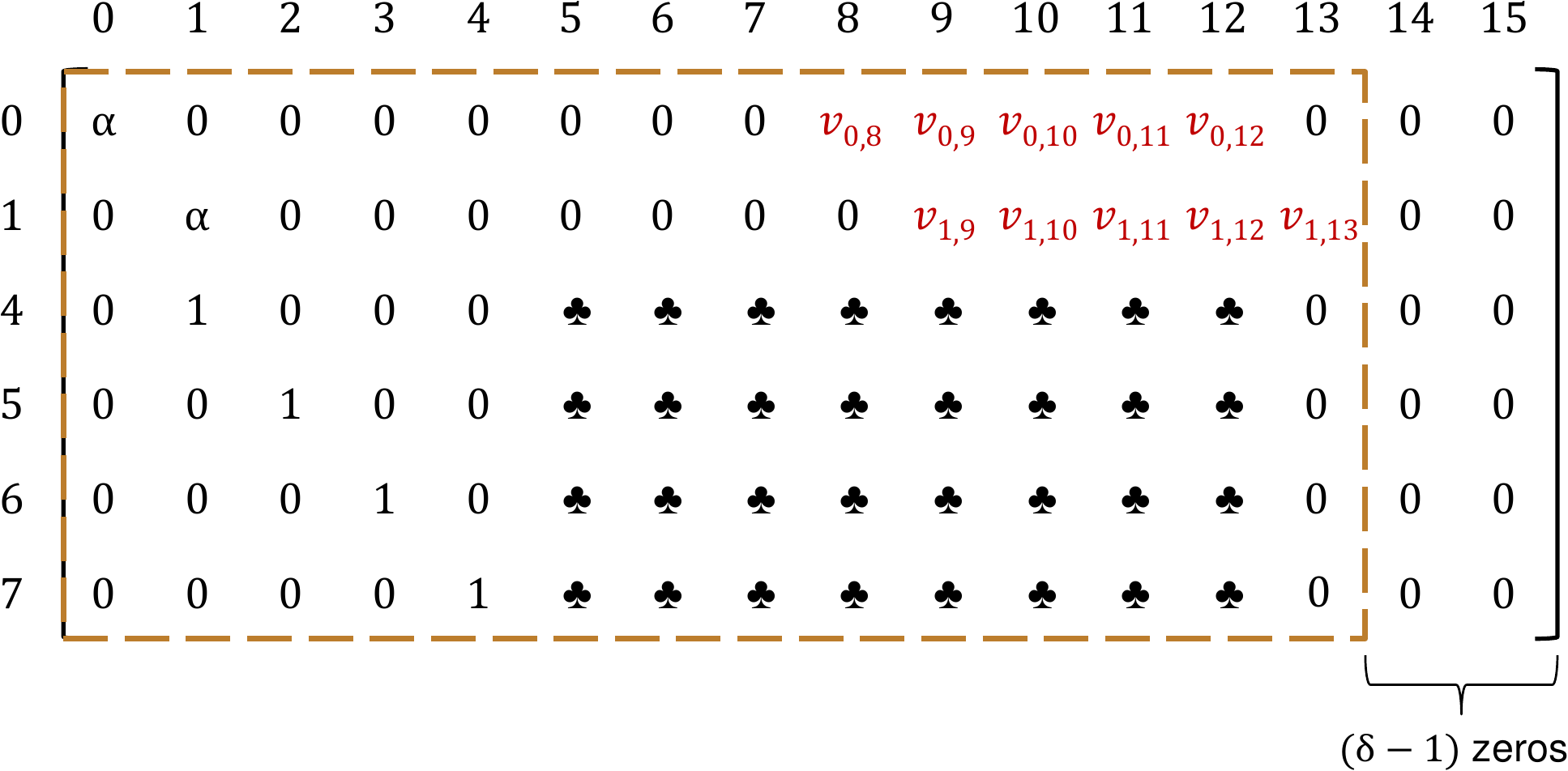}
			\caption{Consider the p-c matrix $H$ given in \ref{fig:construction_A_example}. Let $\ell=1$. All the rows $[0:1]\cup[4:7]$ of $H$ have a run of zeros among the last $(\delta-\ell)=2$ coordinates. Hence all the rows of the matrix demarcated here in dashed lines lie in the row space of $H^{(1)}$.}
			\label{fig:construction_A_RE_1}
		\end{figure}
		
		\bit
		\item (\textit{Case I}: $\ell=0$): Here we consider a subset of the rows $R\triangleq [\delta+1:b-1]\subseteq R_0$. The $(a-1)\times (\tau+1)$ matrix $H(R,0:\tau)$ has a zero-column at column $0$ and rest of the columns of $H(R,0:\tau)$ form an $(a-1)\times \tau$ MDS matrix (see Fig. \ref{fig:construction_A_RE_2}). Hence R1 is satisfied for $\ell=0$. 
		
		\begin{figure}[!htb]
			\centering
			\captionsetup{justification=centering}
			\includegraphics[scale=0.6]{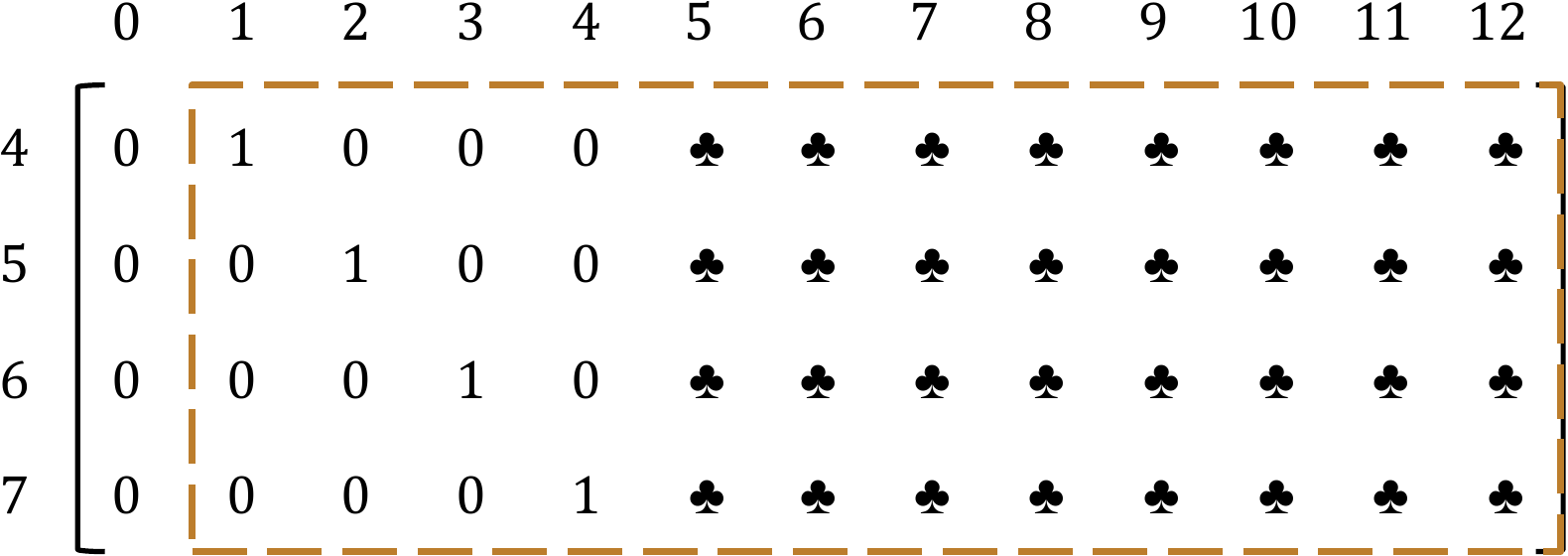}
			\caption{An illustration of $H(R,0:\tau)$ for the p-c matrix given in \ref{fig:construction_A_example}. The submatrix of size $4\times 12$ demarcated by dashed lines, is an MDS matrix.}
			\label{fig:construction_A_RE_2}
		\end{figure}
		\item (\textit{Case II}: $\ell\in[1:\delta-1]$): Here we consider the rows $S_\ell\triangleq\{\ell\}\cup[\delta+1:b-1]\subseteq R_\ell$ during the proof. Partition the matrix $H(S_\ell,0:\ell+\tau)$ in the following form:
		\bean
		\left[ \begin{array}{c} U_{\ell}\\ L_{\ell} \end{array} \right], 	
		\eean
		where  $U_\ell\triangleq H(\ell,0:\ell+\tau)$ and $L_\ell\triangleq H([\delta+1:b-1],0:\ell+\tau)$ (see Fig. \ref{fig:construction_A_RE_3} for an example case of $\ell=2$).   Note that the rows of $H(S_\ell,0:\ell+\tau)$ belong to the row space of \hsupl\ and so we are justified in working with the submatrix $H(S_\ell,0:\ell+\tau)$.  Note also that $U_\ell=[U_{\ell,0}\ \ U_{\ell,1}\ \ \cdots\ \ U_{\ell,\ell+\tau}]$ is a row vector.
		
		\begin{figure}[!htb]
			\centering
			\captionsetup{justification=centering}
			\includegraphics[scale=0.6]{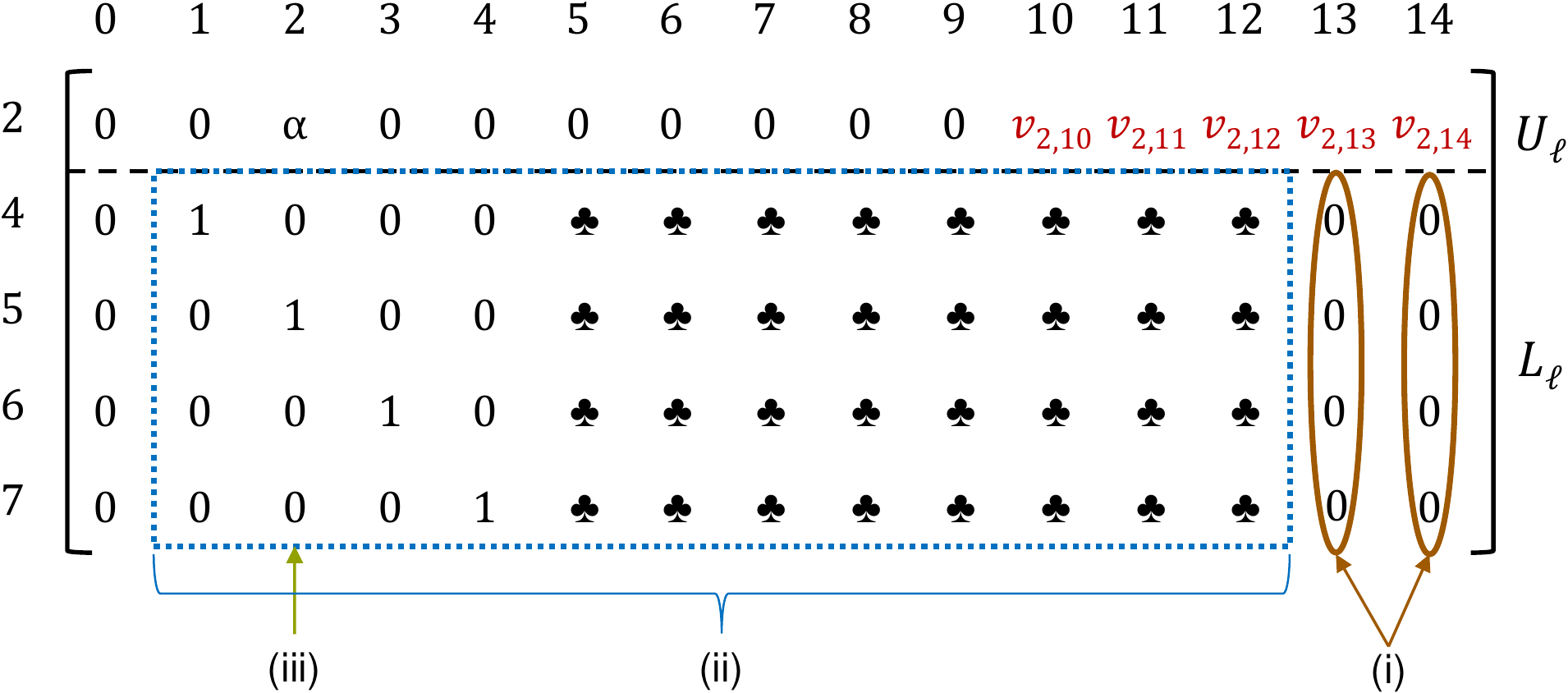}
			\caption{Let $\ell=2$, $S_\ell=\{2\}\cup[4:7]$ for the p-c matrix given in \ref{fig:construction_A_example}. The matrix $H(S_\ell,0:\ell+\tau)$ is partitioned into $U_\ell$ and $L_\ell$. Here the matrix, which is a part of $L_\ell$ and demarcated by dotted lines, is an MDS matrix.}
			\label{fig:construction_A_RE_3}
		\end{figure}
		
		An inspection of $L_\ell$ reveals the following facts: 
  	    \begin{enumerate}[label=(\roman*)]
		\item the columns of $L_\ell$ corresponding to the coordinates $[1+\tau:\ell+\tau]$ are all zero; in other words, the entries of  $H(\delta+1:b-1,1+\tau:\ell+\tau)$ are all zero; 
		
		\item $L_\ell(:,1:\tau)$ is an MDS matrix of size $(a-1)\times \tau$ over the subfield $\mathbb{F}_q$ of $\mathbb{F}_{q^2}$. 
		
		\item the $\ell$-th column of the matrix $L_{\ell}$  is within the MDS matrix $L_\ell(:,1:\tau)$ since $\ell\in[1:\delta-1]$ and $\tau\geq b>\delta$.
		\een
		In Fig. \ref{fig:construction_A_RE_3}, we illustrate these facts with respect to the example case of $\ell=2$.
		
		Now assume contrary to the condition R1, that there exists a set of $a$ coordinates, $A_\ell\subseteq [\ell:\ell+\tau]$, with $|A_\ell|=a$ and $\ell\in A_\ell$ such that $\underline{h}^{(\ell)}_{\ell}$ lies in the span of columns $\{\underline{h}^{(\ell)}_{i}\mid i\in A_\ell\setminus\{\ell\} \}$, i.e., 
		\beq\label{eq:constrn_a_re_lc}	
		\underline{h}^{(\ell)}_{\ell}=\sum_{i\in A_\ell\setminus\{\ell\}}a_i\underline{h}^{(\ell)}_{i}, 
		\eeq
		where $a_i\in\mathbb{F}_{q^2}$.  Equation \eqref{eq:constrn_a_re_lc} implies the following:
		
		\beq\label{eq:constrn_a_l_lc}	
		L_\ell(:,\ell)=\sum_{i\in A_\ell\setminus\{\ell\}}a_iL_\ell(:,i), 
		\eeq
		
		\beq  \label{eq:constrn_a_u_lc}
		U_{\ell,\ell}=\sum_{i \in A_\ell \setminus \{\ell \}} a_i U_{\ell,i}. 
		\eeq

		As $A_\ell\subseteq [\ell:\ell+\tau]$ by assumption, (i), (ii) and (iii) would then imply that in order for \eqref{eq:constrn_a_l_lc} to happen, $A_\ell\subseteq [\ell:\tau]$ and the coefficients $\{a_i\}$ are uniquely determined, must be all-non-zero and must all belong to the subfield $\mathbb{F}_q$. 
		
		Now consider the equation \eqref{eq:constrn_a_u_lc}. We have $U_{\ell,\ell}=\alpha\in\fqs\setminus\fq$ and $U_{\ell,i}\in\fq$ for $i\neq \ell$. Together with the constraint imposed by \eqref{eq:constrn_a_l_lc} that $\{a_i\}\subseteq \fq$, we have $LHS\in\fqs\setminus\fq$ and $RHS\in \fq$ for \eqref{eq:constrn_a_u_lc}. This clearly contradicts \eqref{eq:constrn_a_re_lc}. Thus our assumption as to the existence of $A_\ell$ is invalid, which proves that $H$ satisfies the condition R1. In Fig. \ref{fig:construction_A_RE_4}, we consider an example case where $\ell=2$ and  $A_\ell=\{2,4,7,10,12\}$.

		\begin{figure}[!htb]
			\centering
			\captionsetup{justification=centering}
			\includegraphics[scale=0.6]{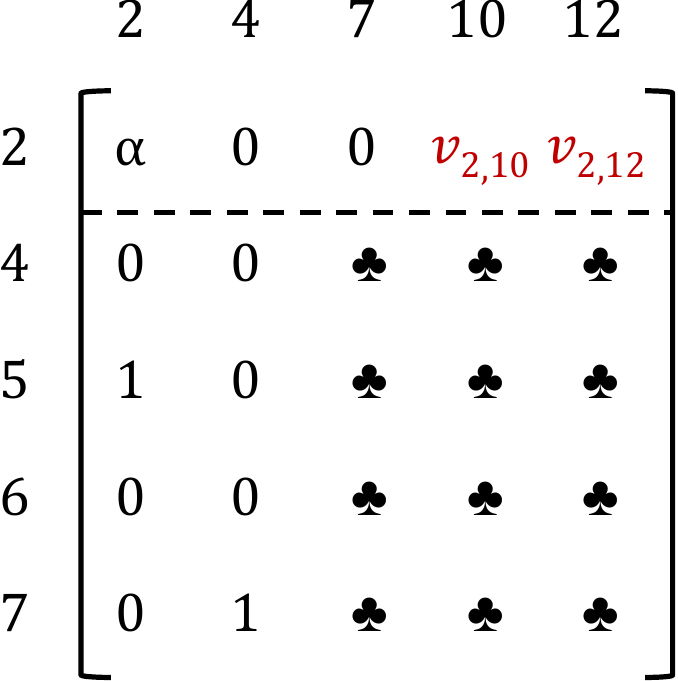}
			\caption{Let $\ell=2,A_\ell=\{2,4,7,10,12\}$. Here any $4$ columns of the $4\times 5$ submatrix indexed by rows $[4:7]$ form an independent set. The linear combination shown in \eqref{eq:constrn_a_re_lc} when restricted to this submatrix would then imply that the choice of coefficients $\{a_4,a_7,a_{10},a_{12}\}$ is unique, with all the coefficients drawn from \fq$\setminus\{0\}$. However $\alpha\in\mathbb{F}_{q^2}\setminus\mathbb{F}_q$ and $\{v_{2,10},v_{2,12}\}\subseteq\mathbb{F}_q$. Hence we arrive at a contradiction to \eqref{eq:constrn_a_re_lc}.}
			\label{fig:construction_A_RE_4}
		\end{figure}
		\eit
		\item Condition R2:  Let $S\triangleq [\delta:b-1]$. We make the following three observations:
		\ben
		\item  The columns $\left\{H(S,i)\mid i\in [\tau+1:\tau+\delta-1]\right\}$ are all-zero columns.
		
		\item The matrix formed using the remaining columns of $H(S,\delta:n-1)$, i.e., $H(S,[\delta:n-1]\setminus[\tau+1:\tau+\delta-1])$, forms an MDS matrix of size $a\times (a+(\tau-b)+2)$.
		
		\item  The collection of columns $\left\{ \underline{h}_{j} \mid j\in[\tau+1:\tau+\delta-1] \right\}$ forms a linearly independent set as required by condition B2, since it is a subset of the larger linearly independent set of $b$ columns, $\left\{ \underline{h}_{j} \mid j\in[\tau-a+1:\tau+\delta] \right\}$.
		\een 
		In Fig. \ref{fig:construction_A_RE_5}, we illustrate these observations with the help of an example. 
		
		In order to establish Condition R2, assume there exists a set $\al\subseteq [\delta:n-1]$ with $|\al|\leq a$ such that:
		\beq\label{eq:constrn_a_re2_lc}	
		\sum_{i\in A_\ell}a_i\underline{h}_{i}=0, 
		\eeq
		where the $\{a_i\}$ are all $\neq 0$. From the above observations 1) and 2) above, we have that \al\ cannot contain any of the columns in $[\delta:n-1]\setminus[\tau+1:\tau+\delta-1]$. However on the other hand, observation 3) states that the remaining columns form an independent set. Hence \eqref{eq:constrn_a_re2_lc} is not feasible and thus condition R2 is satisfied.
		\begin{figure}[!htb]
			\centering
			\captionsetup{justification=centering}
			\includegraphics[scale=0.6]{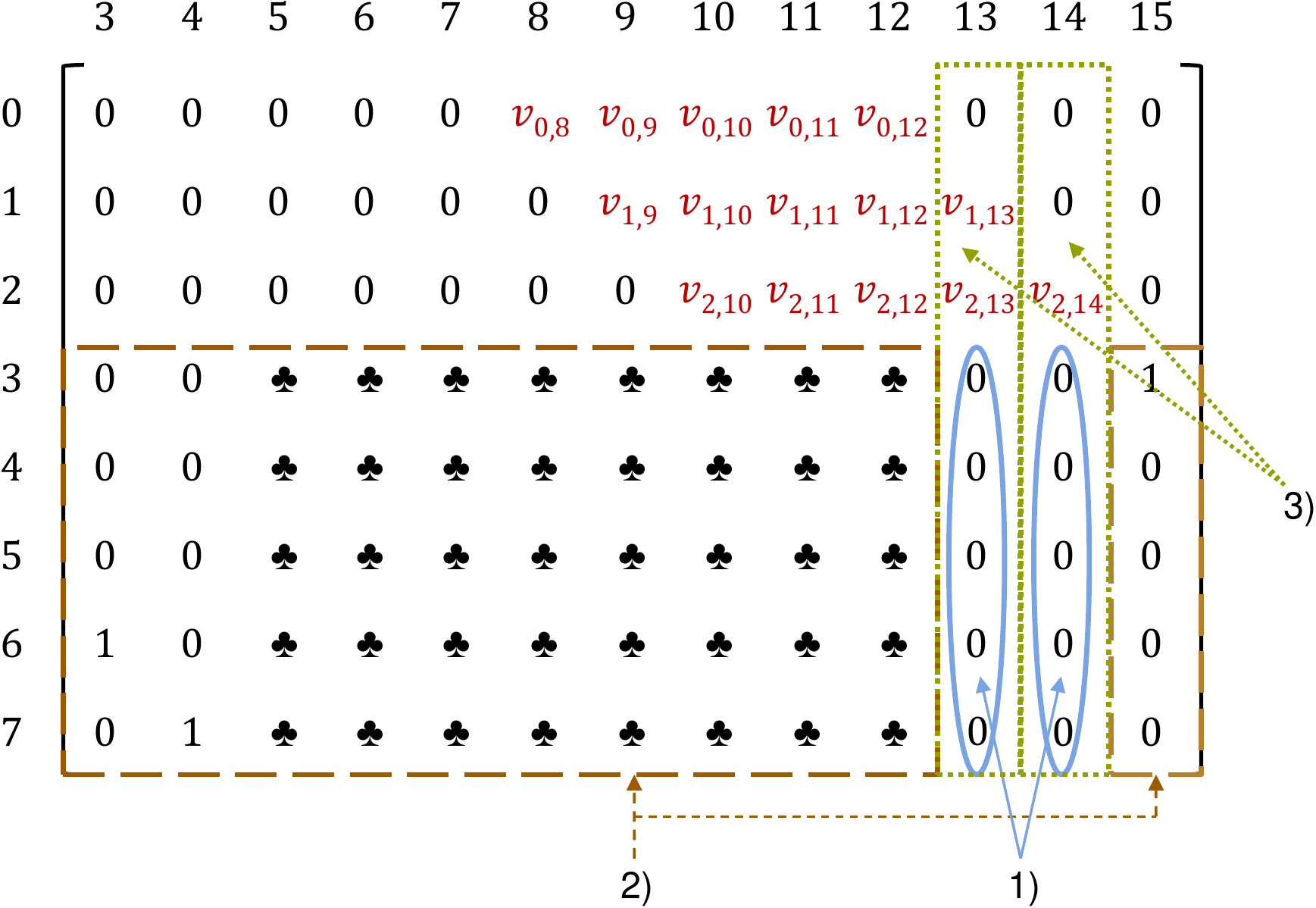}
			\caption{An illustration of $H(:,\delta:n-1)$ for the p-c matrix given in \ref{fig:construction_A_example}. In the figure, the three observations we make in the proof of condition R2 are identified.}
			\label{fig:construction_A_RE_5}
		\end{figure}
		\eit
		\eit
		%%%%%%%%%%%%%%%%%%%%%%%
		%%%%%%%%%%%%%%%%%%%%%%%
		\section{Proof of Theorem \ref{thm:constr_B} (Construction B)}\label{app:proof_constrn_b}
        	\bit
        \item[]
        \item {\it Recovery from burst erasure of length $\leq b$}: 
        \bit
        \item Condition B1: For $\ell\in[0:\delta-1]$, the $\ell$-th row \hrowl\ of the p-c matrix $H$ takes the form:
        
        \bean
        \hrowl\ & = & [\underbrace{\triangle \cdots \triangle}_{(\ell-1) \text{ symbols}} \ \ast \ \underbrace{0 \cdots 0}_{(b-1) \text{ symbols}} \ \triangle\cdots\triangle \ \underbrace{0 \cdots 0}_{\text{last } (\delta-\ell) \text{ symbols}} ],
        \eean
        where $\triangle$'s indicate elements over \fq\ and $\ast$ indicates a non-zero element over \fq. Similar to the corresponding case appearing in the proof of Theorem \ref{thm:constr_A}, using the structure of \hrowl, one can argue that $H$ satisfies condition B1.

        \item Condition B2: For $\delta \leq \ell \leq (\tau-a+1)$, let \pl\ denote the $(b\times b)$ square submatrix of $H$ corresponding to the column set $[\ell:\ell+b-1]$ as in Construction A, i.e., 
        \bean
        \pl\ & \triangleq & H(:, [\ell:\ell+b-1]). 
        \eean

        Consider the partition of $H(:,\delta:\delta+\tau)$ into three parts; $H(:,\delta:\tau)$, $H(:,1+\tau:a-1+\tau)$ and $H(:,a+\tau:\delta+\tau)$ (see Fig. \ref{fig:construction_B_BE_proof_1}). All possible burst erasure patterns involving coordinates $[\delta:n-1]$ can be classified into five cases. Let $\bl\triangleq [\ell:\ell+b-1]$.
        
        \begin{figure}[!htb]
        	\centering
        	\captionsetup{justification=centering}
        	\includegraphics[scale=0.6]{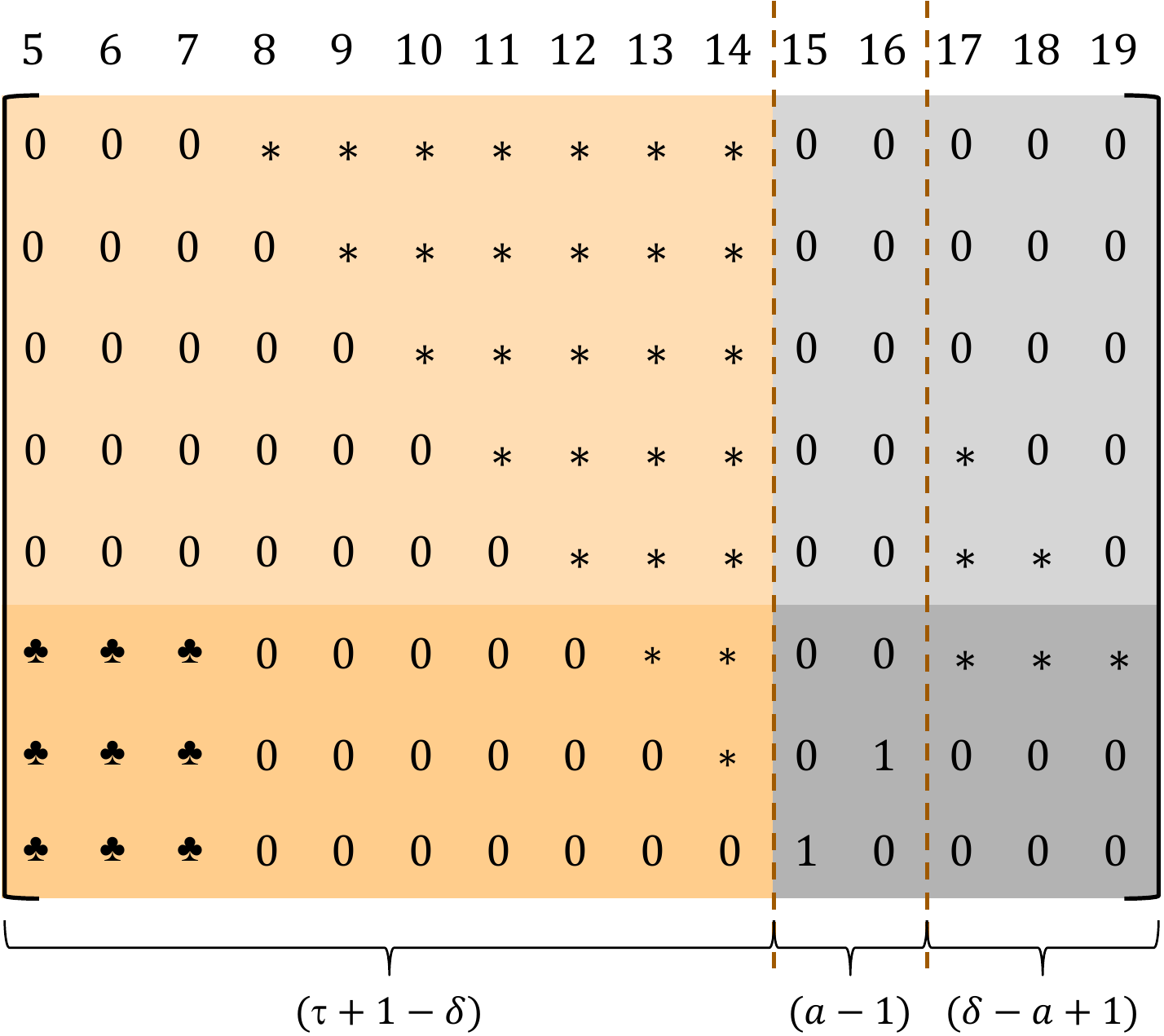}
        	\caption{Consider the p-c matrix given in Fig. \ref{fig:delta_geq_a_example_step4}. The $8 \times 15$ matrix $H(:,5:19)$ is partitioned into three regions; $H(:,5:14)$, $H(:,15:16)$ and $H(:,17:19)$. }
        	\label{fig:construction_B_BE_proof_1}
        \end{figure}
        %		One can formally indicate this scenario as $\ell\in[\delta:\min\{b,(t-b+1)\}]$.
        \bit
        \item	 (\textit{Case I}: $\bl\subseteq [\delta:\tau]$ and $ \ell \leq b$):  In words, this is the scenario where  \pl\ is contained within $H(:,\delta:\tau)$ and \pl\ contains the column $\underline{h}_b$. Here, it can be observed that \pl\ takes the form:
        \bea \label{eq:constrn_B_p_l_form}
        P_{\ell } & = & \left[ \begin{array}{cc} P_{\ell,1} & P_{\ell,3} \\ P_{\ell,2} & [0] \end{array} \right], 	
        \eea
        where $P_{\ell,3}\triangleq H(0:\ell-1,b:\ell+b-1)$  is an invertible upper triangular matrix and $P_{\ell,2}\triangleq H(\ell:b-1,\ell:b-1)$ is an invertible matrix. Applying Lemma \ref{lem:mds_consec_rows}, the $\ell\times\ell$ matrix $P_{\ell,3}$ is invertible as it consists of $\ell$ non-zero columns chosen from $\gmds(0:\ell-1,:)$. Note that  \gmds\ is the ZB generator matrix  used in Step-a of the construction. $P_{\ell,2}$ is invertible as it is a square submatrix of the Cauchy-like matrix $H(\delta:b-1,\delta:b-1)$. Hence \pl\ is invertible. In Fig. \ref{fig:construction_B_BE_proof_2}, we illustrate case I when $\ell=5$ and $6$, for the p-c matrix shown in Fig. \ref{fig:delta_geq_a_example_step4}.
        
        \begin{figure}[!htb]
        	\centering
        	
        	%add desired spacing between images, e. g. ~, \quad, \qquad, \hfill etc. 
        	%(or a blank line to force the subfigure onto a new line)
        	\begin{subfigure}[b]{0.4\textwidth}
        		\centering
        		\includegraphics[scale=0.65]{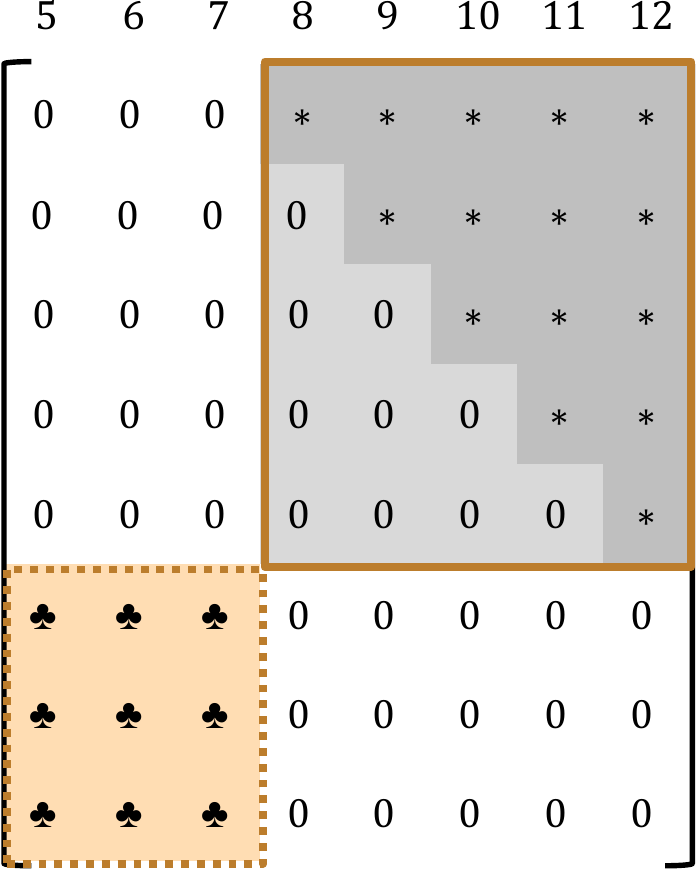}
        		\caption{}
        	\end{subfigure}
        	%add desired spacing between images, e. g. ~, \quad, \qquad, \hfill etc. 
        	%(or a blank line to force the subfigure onto a new line)
        	\begin{subfigure}[b]{0.4\textwidth}
        		\centering
        		\includegraphics[scale = 0.65]{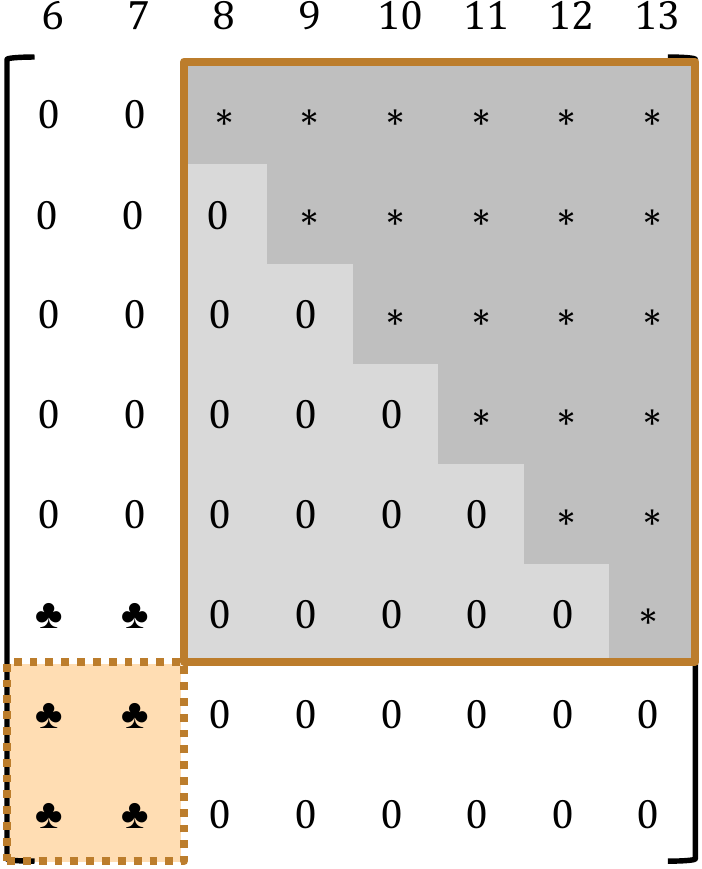}
        		\caption{}
        	\end{subfigure}
        	\caption{Two $(b \times b)$ square submatrices of the p-c matrix in Fig. \ref{fig:delta_geq_a_example_step4}. The submatrices demarcated by dotted lines in (a) and (b) are invertible, as they are square submatrices of the Cauchy-like matrix $H(5:7,5:7)$. The $5\times 5$ submatrix shown in (a) using solid lines is invertible as it consists of $5$ non-zero columns chosen from $5$ adjacent rows of the ZB generator matrix \gmds\ used in Step-a. Similarly, the $6\times 6$ matrix marked in (b) is invertible as it consists of $6$ non-zero columns chosen from $6$ adjacent rows of \gmds.}
        	\label{fig:construction_B_BE_proof_2}
        \end{figure}
        
        \item  (\textit{Case II}: $\bl\subseteq [\delta:\tau]$ and $ \ell \ge b$): Note that this scenario arises only if $b+(b-1)\leq \tau$. In this case, \pl\ contains $b$ columns from \gmds\ and hence is invertible. For the example p-c matrix in Fig. \ref{fig:delta_geq_a_example_step4} that we have been using, we have $b=8$ and $\tau=14$, and hence case II will not be encountered. In order to specifically illustrate this case, we consider an example p-c matrix for parameters $\{a=3,b=8,\tau=16\}$ in Fig. \ref{fig:construction_B_BE_proof_3}.

        \begin{figure}[!htb]
        	\centering
        	\captionsetup{justification=centering}
        	\includegraphics[scale=0.65]{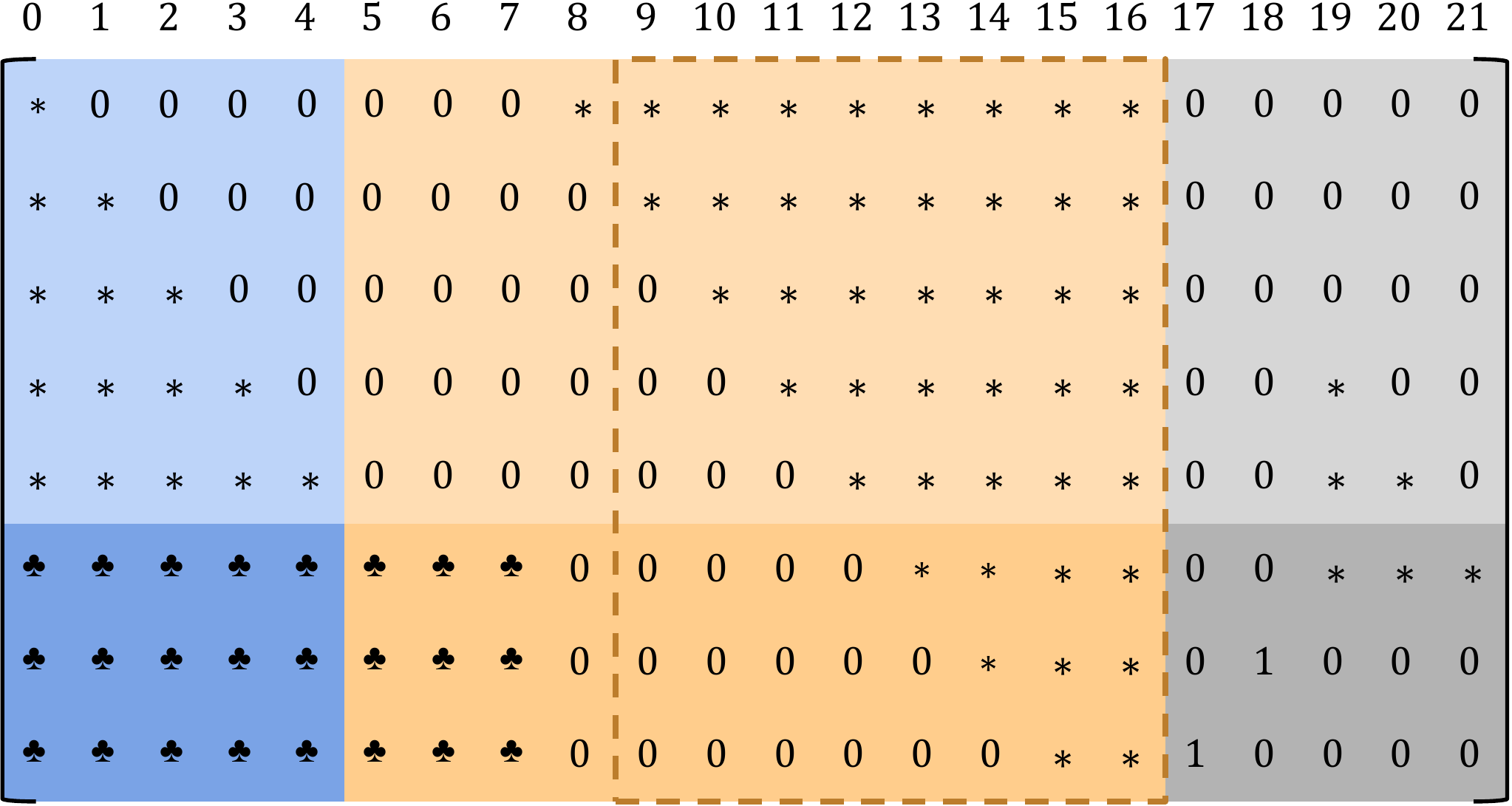}
        	\caption{An example p-c matrix $H$ for parameters $\params=\{3,8,16\}$. The submatrix $P_9$ demarcated by dashed lines consists of $8$ columns of the $8\times 17$ MDS matrix \gmds\ introduced in Step-a. Hence $P_9$ is invertible. }
        	\label{fig:construction_B_BE_proof_3}
        \end{figure}	
        
        \item (\textit{Case III}: $\bl\not\subseteq [\delta:\tau]$, $\bl\subseteq[\delta:a-1+\tau]$ and $\ell\leq b$): This corresponds to the scenario where  $P_\ell$ includes some columns from $H(:,1+\tau:a-1+\tau)$ but not any columns from $H(:,a+\tau:\delta+\tau)$ and also column $\underline{h}_b$ is present in the matrix \pl.	Here \pl\ takes the form shown in Figure \ref{fig:construction_B_BE_proof_4_1}. By changing the ordering of columns from  $\{\ell,(\ell+1),\ldots,b,\ldots,\tau,(1+\tau)\,\ldots,(\ell+b-1)\}$ to $\{(1+\tau),\ldots,(\ell+b-1),\ell,\ldots,b,\ldots,\tau\}$, as shown in Fig. \ref{fig:construction_B_BE_proof_4_2}, we obtain a column-permuted \pl\ of the form \eqref{eq:constrn_B_p_l_form}. Here $P_{\ell,2}\triangleq[H(\tau-b+1:b-1,1+\tau:\ell+b-1)\ H(\tau-b+1:b-1,\ell:b-1)]$, as can be seen from Fig. \ref{fig:construction_B_BE_proof_4_2}, is an invertible matrix. Similar to case I, $P_{\ell,3}$ is an upper triangular matrix and is invertible as it consists of $(\tau-b+1)$ non-zero columns chosen from $(\tau-b+1)$ adjacent rows of the ZB generator matrix \gmds. Hence \pl\ is invertible.

        \begin{figure}[!htb]
        	\centering
        	
        	%add desired spacing between images, e. g. ~, \quad, \qquad, \hfill etc. 
        	%(or a blank line to force the subfigure onto a new line)
        	\begin{subfigure}[b]{0.4\textwidth}
        		\centering
        		\includegraphics[scale=0.59]{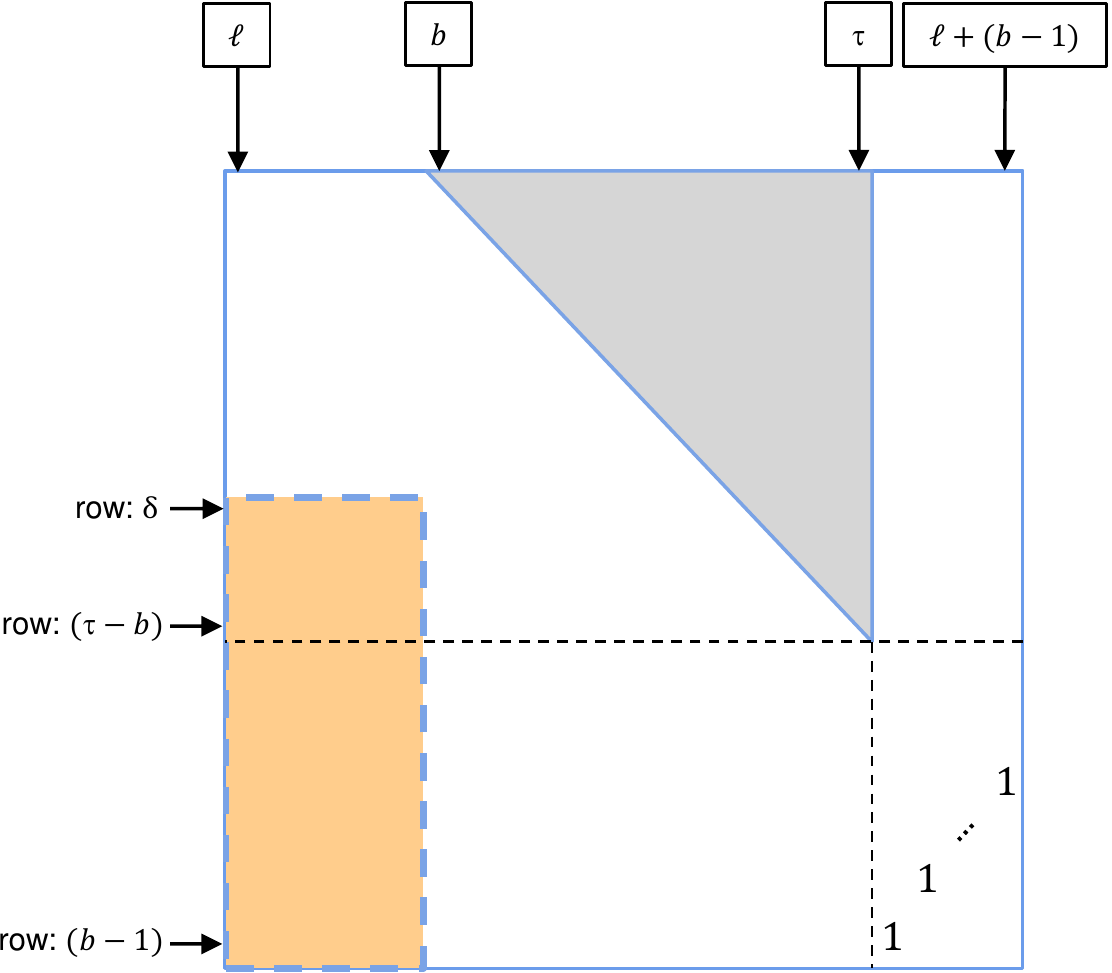}
        		\caption{As $(\tau+1)\geq (b+\delta)$, the matrix $H(\tau-b+1:b-1,\ell:b-1)$ is a submatrix of the Cauchy-like matrix $H(\delta:b-1,\delta:b-1)$.  }
        		\label{fig:construction_B_BE_proof_4_1}
        	\end{subfigure}
        	%add desired spacing between images, e. g. ~, \quad, \qquad, \hfill etc. 
        	%(or a blank line to force the subfigure onto a new line)
        	\begin{subfigure}[b]{0.4\textwidth}
        		\centering
        		\includegraphics[scale = 0.59]{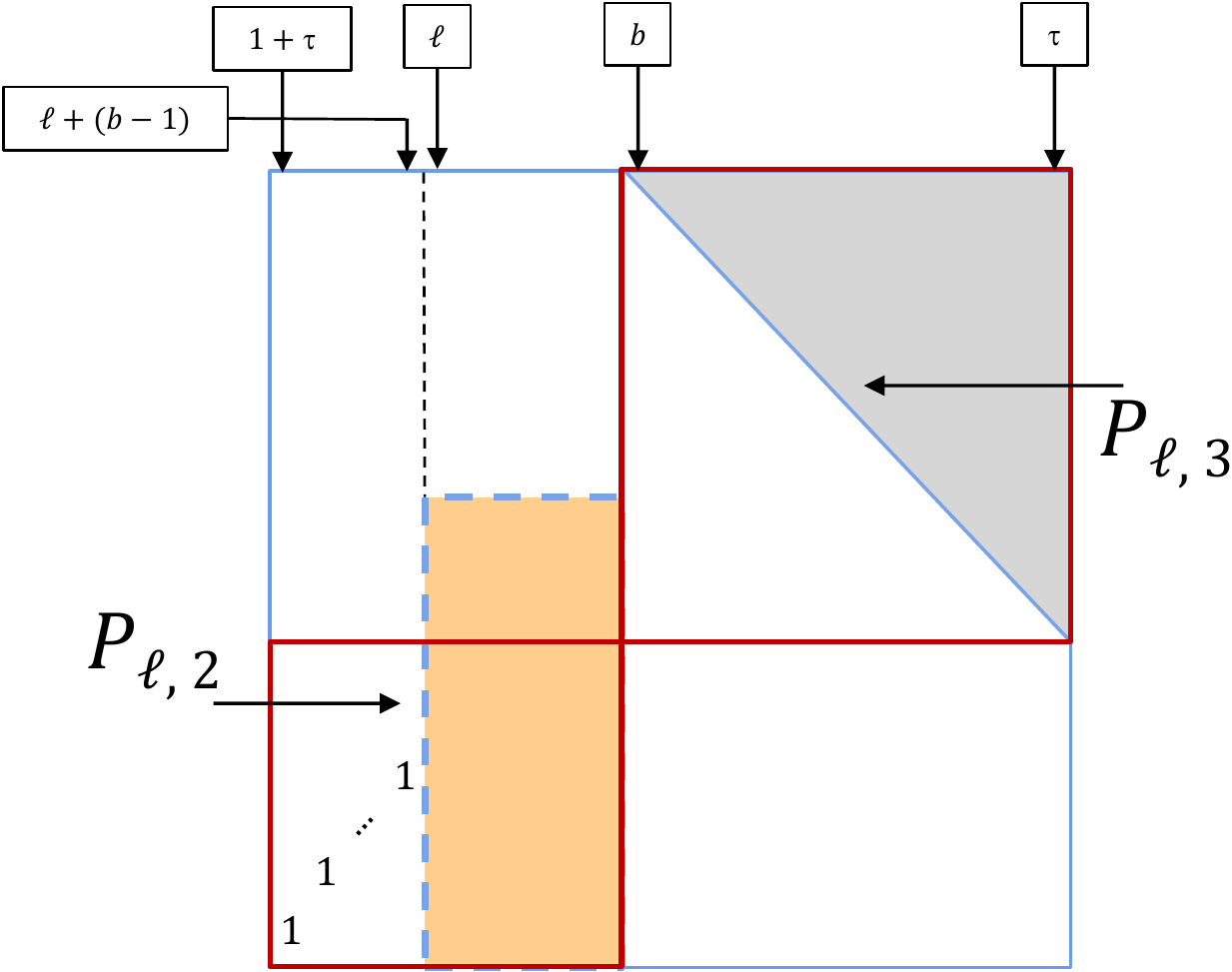}
        		\caption{An illustration of the column-permuted \pl. Here $P_{\ell,2}$ is invertible, as its columns consist of standard basis vectors and columns of the Cauchy-like matrix $H(\tau-b+1:b-1,\ell:b-1)$.} 
        		\label{fig:construction_B_BE_proof_4_2}
        	\end{subfigure}
        	\caption{An illustration of \pl\ for case III}
        \end{figure}
        
        \item (\textit{Case IV}:  $\bl\not\subseteq [\delta:\tau]$, $\bl\subseteq [\delta:a-1+\tau]$ and $\ell\geq b$): For this case, \pl\ takes the form:
        \bean \label{eq:constrn_B_p_l_form_2}
        P_{\ell } & = & \left[ \begin{array}{cc} P_{\ell,1} &[0] \\ P_{\ell,2} & P_{\ell,4} \end{array} \right], 	
        \eean
        where  $P_{\ell,1}$ and  $P_{\ell,4}$ are invertible matrices as can be seen from Fig. \ref{fig:construction_B_BE_proof_5} (see Fig. \ref{fig:construction_B_BE_proof_6} for an example). Hence \pl\ is invertible.
        \begin{figure}[!htb]
        	\centering
        	\captionsetup{justification=centering}
        	\includegraphics[scale=0.58]{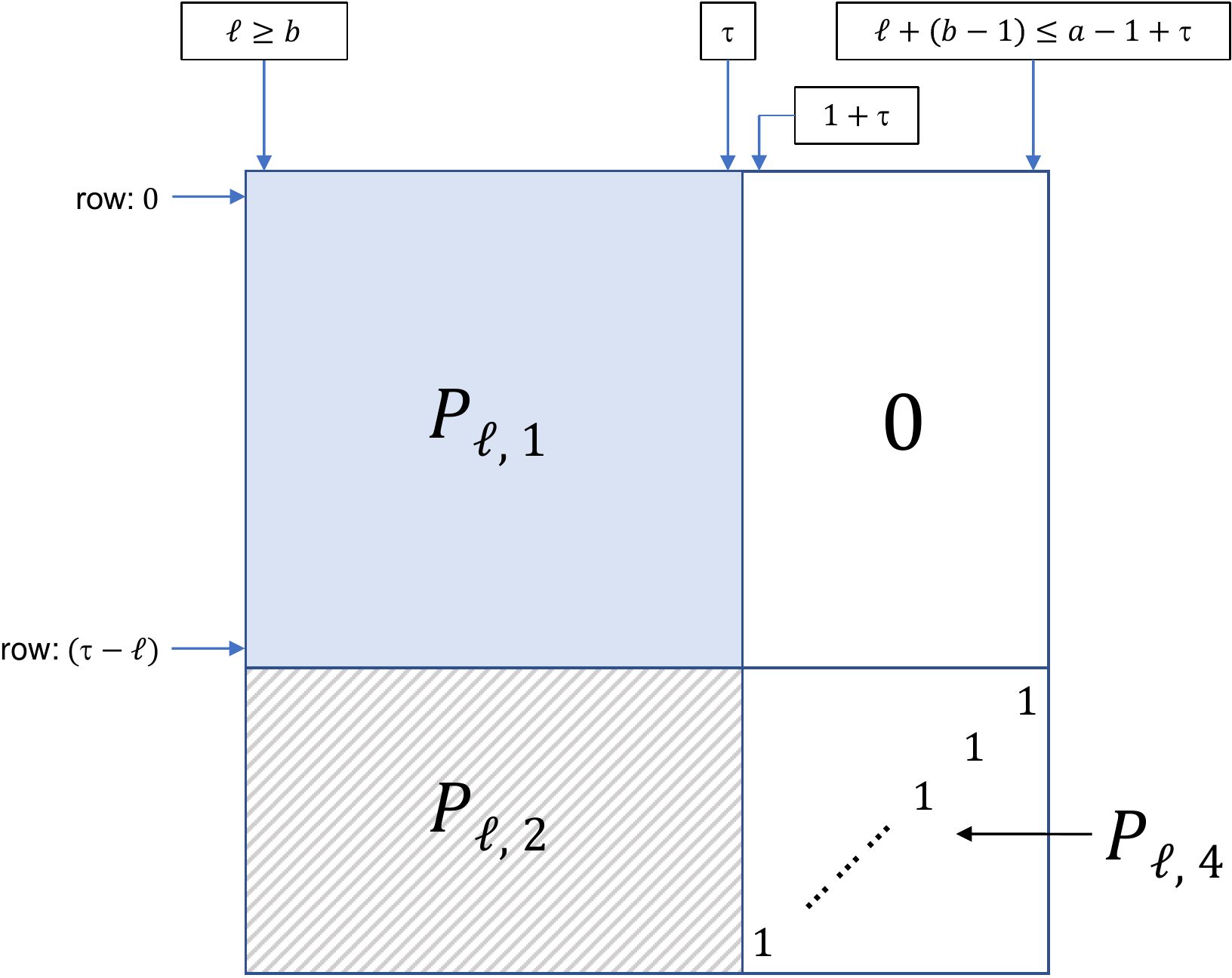}
        	\caption{In this figure, we illustrate the structure of \pl, for case IV. Here $P_{\ell,1}$ is a $(\tau-\ell+1)\times (\tau-\ell+1)$ submatrix consisting of $(\tau-\ell+1)$ non-zero columns from the first $(\tau-\ell+1)$ rows of the ZB generator matrix \gmds. Hence $P_{\ell,1}$ is invertible. The matrix $P_{\ell,4}$ is clearly invertible and hence invertibility of \pl\ follows. }
        	\label{fig:construction_B_BE_proof_5}
        \end{figure}
        
        \begin{figure}[!htb]
        	\centering
        	\captionsetup{justification=centering}
        	\includegraphics[scale=0.65]{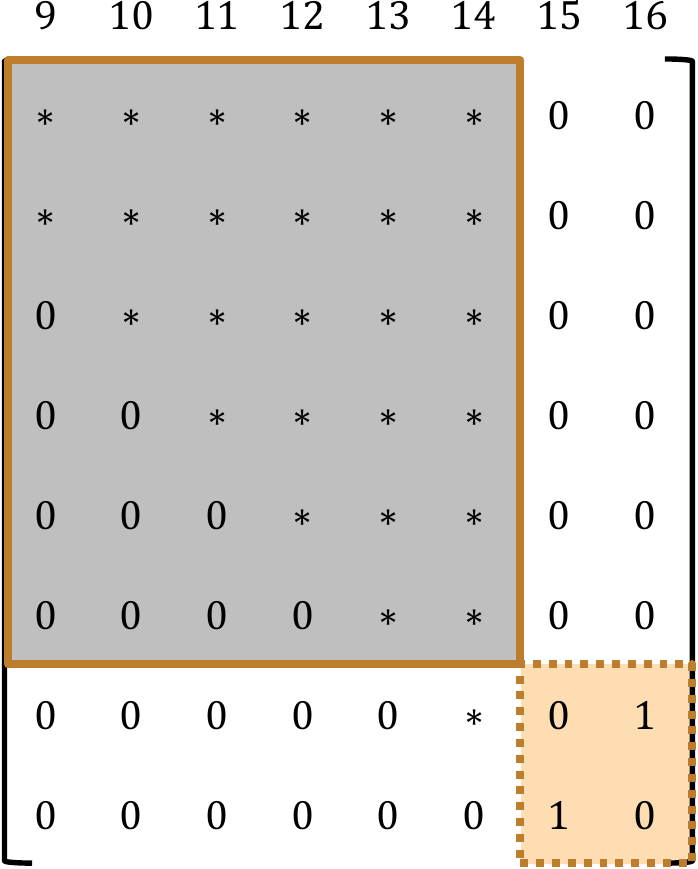}
        	\caption{Let $\ell=9$. In the figure, we illustrate the matrix $P_\ell$ corresponding to the p-c matrix given in Fig. \ref{fig:delta_geq_a_example_step4}. The two submatrices demarcated by solid and dotted lines, are  both invertible.}
        	\label{fig:construction_B_BE_proof_6}
        \end{figure}

        \item (\textit{Case V}:  $\bl\not\subseteq [\delta:a-1+\tau]$): Clearly this corresponds to the scenario where \pl\ includes some columns from $[a+\tau:\delta+\tau]$ and hence:
        \beq\label{eq:constrn_b_b2_case_v_1}
        (\ell+b-1)\geq (a+\tau).
        \eeq 
        
        Consider the ZB generator matrix \gmds\ of size $b\times (\tau+1)$ used in Step-a. In the following, with the help of three observations, we show that $b$ distinct columns of \gmds\ lie in the column space of \pl. As any $b$ columns of \gmds\ form an independent set, this would imply that \pl\ is invertible.

        As $\ell$ ranges from $\delta$ to $(\tau-a+1)$, using \eqref{eq:constrn_b_b2_case_v_1}, one can infer that the columns $\{\underline{h}_i\mid 1+\tau\leq i\leq a-1+\tau\}$ are part of \pl. As these $(a-1)$ columns are standard basis vectors covering the last $(a-1)$ rows, and the independent set of columns $\{\gmds(:,i)\mid \delta+1\leq i\leq b-1\}$ have non-zero entries only in the last $(a-1)$ rows, we make the following observation:
        
        (i) The space spanned by column vectors
        \beqn
        \{\underline{h}_i\mid 1+\tau\leq i\leq a-1+\tau\}
        \eeqn
        and
        \beqn
        \{\gmds(:,i)\mid \delta+1\leq i\leq b-1\}
        \eeqn
        is the same (see Fig. \ref{fig:construction_B_BE_proof_7}).

        Recall the constraint $(\tau+1)\geq(b+\delta)$ placed on the construction. Together with \eqref{eq:constrn_b_b2_case_v_1}, we have: $\ell\geq (a+\tau)-(b-1)=(\tau+1)-\delta\geq b$. As columns $[b:\tau]$ of $H$ are unchanged after Step-a, we also have:
        \beqn
        \underline{h}_j=\gmds(:,j),
        \eeqn
        where $b\leq j\leq \tau$. Hence we make a second observation:

        (ii)  As $\ell\geq b$, the  columns $\{\underline{h}_i\mid \ell\leq i\leq \tau\}$ of \pl\ are columns of the MDS matrix \gmds, where $\underline{h}_i=\gmds(:,i)$.

        Now consider the remaining columns $\{\underline{h}_i\mid a+\tau\leq i\leq \ell+b-1\}$ of \pl.  From Step-b, it follows that columns $\gmds(:,i-\tau)$ and $\underline{h}_i$ are identical except at the last $(a-1)$ rows, where $(a+\tau)\leq i\leq (\ell+b-1)$. As the standard basis corresponding to these last $(a-1)$ rows are part of the columns of \pl, we make a third observation as follows:

        (iii) The columns $\{\gmds(:,j)\mid j\in[a:\ell+b-1-\tau]\}$ lie in the span of 	$\{\underline{h}_i\mid 1+\tau\leq i\leq (\ell+b-1)\}$.

        Together observations (i), (ii) and (iii) imply that columns $[a:\ell+b-1-\tau]\cup[\delta+1:b-1]\cup[\ell:\tau]$ of \gmds\ lie in the column space of \pl. As $(\ell+b-1)-\tau\leq (n-1)-\tau=(\tau+\delta)-\tau=\delta$, clearly these are $b$ distinct columns of the $b\times(\tau+1)$ MDS matrix \gmds. This implies that \pl\ has rank $b$ and hence is invertible.
        
        \begin{figure}[!htb]
        	\centering
        	
        	%add desired spacing between images, e. g. ~, \quad, \qquad, \hfill etc. 
        	%(or a blank line to force the subfigure onto a new line)
        	\begin{subfigure}[b]{0.5\textwidth}
        		\centering
        		\includegraphics[scale=0.53]{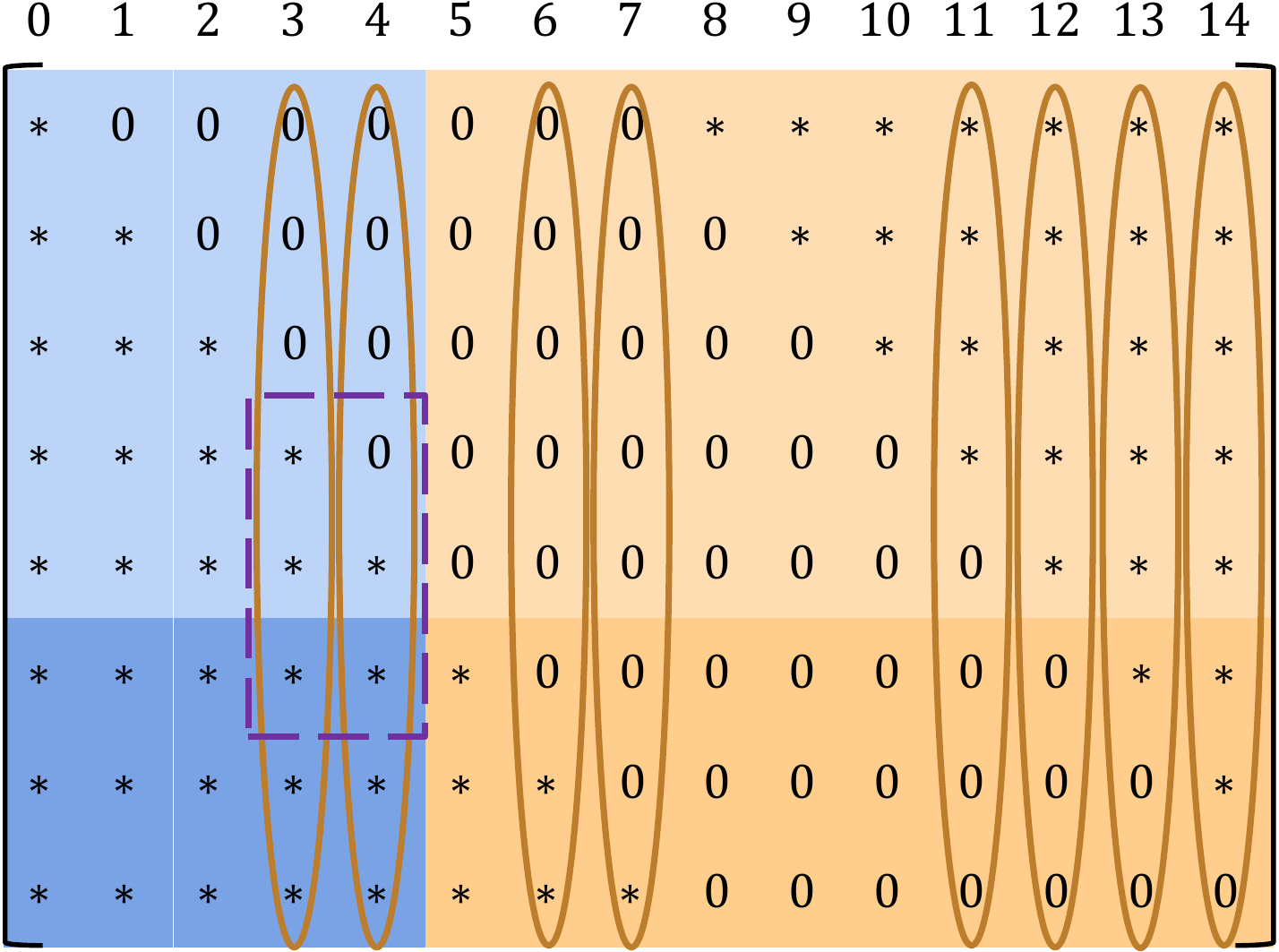}
        		\caption{}
        	\end{subfigure}
        	%add desired spacing between images, e. g. ~, \quad, \qquad, \hfill etc. 
        	%(or a blank line to force the subfigure onto a new line)
        	\begin{subfigure}[b]{0.3\textwidth}
        		\centering
        		\includegraphics[scale = 0.53]{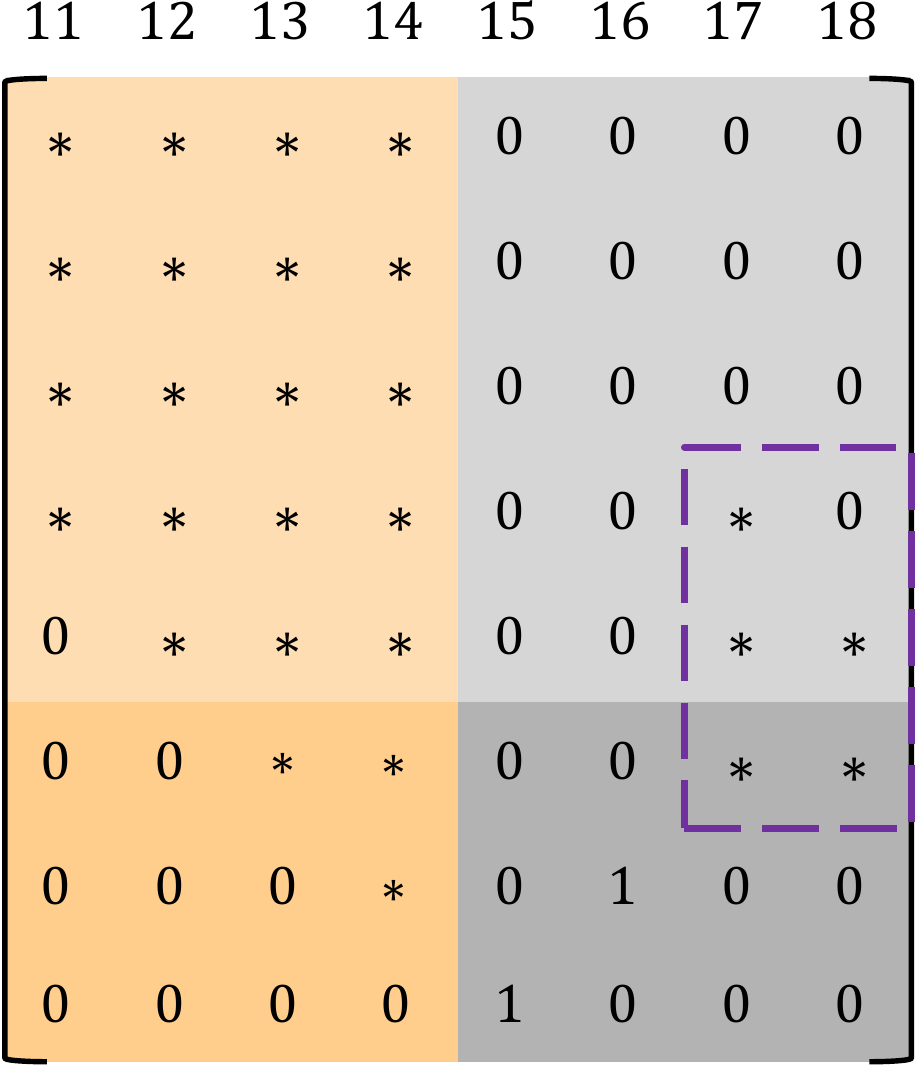}
        		\caption{} 
        	\end{subfigure}
        	\caption{An illustration of the three observations that we make in the proof of case V. (a) The ZB generator matrix \gmds\ corresponding to the p-c matrix shown in Fig. \ref{fig:delta_geq_a_example_step4}. (b) The submatrix $P_{11}\triangleq H(:,11:18)$. We make the following three observations; (i) The column space of $\gmds(:,[6:7])$ and $H(:,[15:16])$ is the same. (ii) $\gmds(:,i)=\underline{h}_i$ for $11\leq i\leq 14$. (iii) Columns $\gmds(:,3)$ and $\gmds(:,4)$ lie in the column space of $H(:,[15:18])$. Hence columns $[3:4]\cup[6:7]\cup[11:14]$ of \gmds\ lie in the column space of \pl. This implies that \pl\ is invertible.}
        	\label{fig:construction_B_BE_proof_7}
        \end{figure}

        \eit
        
        \eit
        
        %%%%%%%%%
        \item {\it Recovery from $\leq a$ random erasures}:
        \bit
        \item Condition R1: Let $0\leq \ell\leq (\delta-1)$, $R\triangleq[0:a-1]$,
        \bean\
        R_\ell & \triangleq & \left\{ \begin{array}{rl} R\cup[(b-\ell):(b-1)], & \text{if } \ell\in[0:(a-1)] \\ 
        	R\cup\left[(\delta+1):(b-1)\right]\cup\left[a:\ell\right]& \text{otherwise}. \end{array} \right.
        \eean

        For every $i\in\rl$, we have the $i$-th row $H(i,0:\ell+\tau)$ lying in the row space of the shortened p-c matrix \hsupl. This is because $H(i,:)$
        has a run of zeros across the last $(\delta-\ell)$ coordinates given by $[\ell+\tau+1:n-1]$. Thus during the discussion on recoverability of $\ell$-th code symbol, we will restrict our attention to the rows indexed by $R_\ell$ (or its subsets) of $H$. In Fig. \ref{fig_construction_B_RE_proof_1}, we illustrate the case of $\ell=3$, with respect to the p-c matrix provided in Fig. \ref{fig:delta_geq_a_example_step4}.
        
        \begin{figure}[!htb]
        	\centering
        	\captionsetup{justification=centering}
        	\includegraphics[scale=0.6]{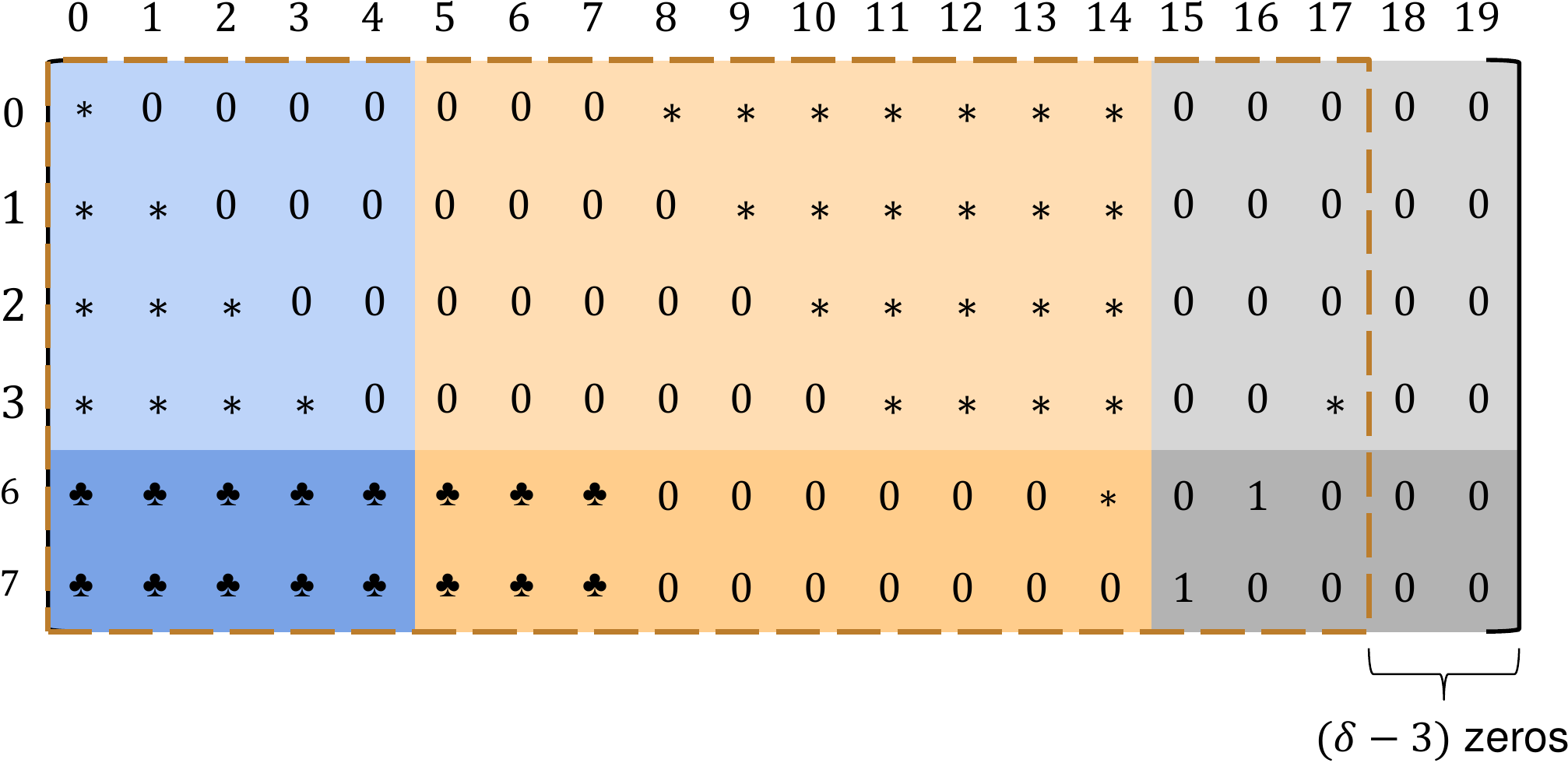}
        	\caption{Let $\ell=3$. The submatrix $H(R_\ell,:)$ for the $H$ given in Fig. \ref{fig:delta_geq_a_example_step4}. Here $R_\ell=[0:2]\cup[6:7]\cup\{3\}$. Each row of $H(\rl,0:3+\tau)$ (demarcated by dashed lines) is part of the row space of  $H^{(3)}$. }
        	\label{fig_construction_B_RE_proof_1}
        \end{figure}
        
        We divide the proof of condition R1 into two cases.
        \bit
        \item ({\it Case I}: $\ell\in[0:a-1])$: Consider the rows $R\triangleq [0:a-1]\subseteq R_\ell$ of $H$. Let $\ul\triangleq H(R,0:\ell+\tau)$. It can be verified that the columns $\mathcal{I}\triangleq[a:b-1]\cup[1+\tau:\ell+\tau]$ of \ul\ are zero columns. If we exclude the zero-columns $[1+\tau:\ell+\tau]$ of \ul, the submatrix $U_\ell(:,0:\tau)$ is composed of $a$ consecutive rows of a ZB generator matrix corresponding to a $[\tau+1,b]$ MDS code. Hence applying Lemma \ref{lem:mds_consec_rows}, we have that  any set of $\leq a$ non-zero columns of $U_\ell$ form an independent set. As $H(R,\ell)\triangleq \ul(:,\ell)$ is a non-zero column of \ul, it follows that condition R1 is satisfied for $\ell\in[0:a-1]$. In Fig. \ref{fig_construction_B_RE_proof_2}, we consider the example of $\ell=2$ with respect to the p-c matrix given in Fig. \ref{fig:delta_geq_a_example_step4}. Note that this proof idea does not extend to the case $\ell\in[a:\delta-1]$, as $\ul(:,\ell)$ is a zero-column for this range of $\ell$. 
        
        \begin{figure}[!htb]
        	\centering
        	\captionsetup{justification=centering}
        	\includegraphics[scale=0.6]{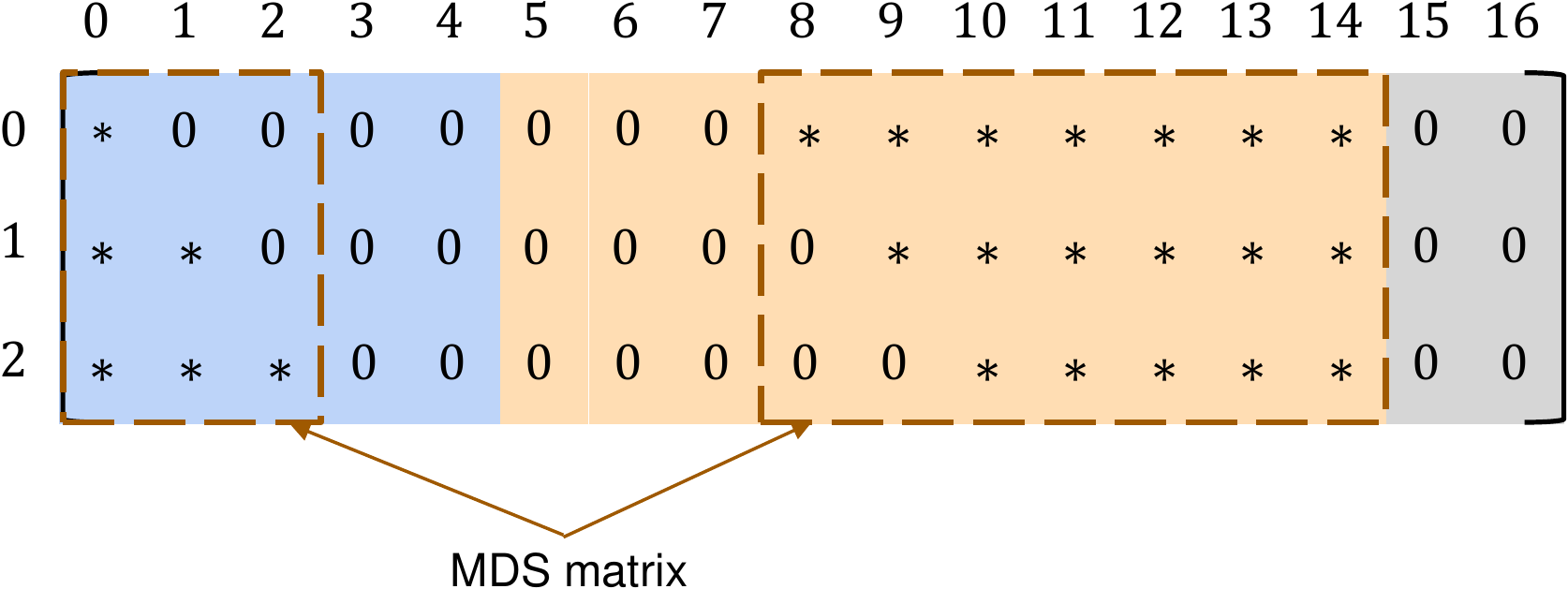}
        	\caption{Here $\ell=2,R=[0:2]$. In the figure, the submatrix $\ul\triangleq H(R,0:16)$ is shown,  which corresponds to the p-c matrix given in Fig. \ref{fig:delta_geq_a_example_step4}. The columns of \ul\ indexed by  $\mathcal{I}\triangleq[3:7]\cup[15:16]$ are zero columns. Any collection of $\leq 3$ non-zero columns form an independent set. In other words, $U_\ell(:,[0:2]\cup[8:14])$ is an MDS matrix. As a result, column $2$ cannot lie in the span of any set of $\leq 2$ other columns.}
        	\label{fig_construction_B_RE_proof_2}
        \end{figure}

        \item ({\it Case II}: $\ell\in[a:\delta-1])$: Consider the set of rows $S_\ell\triangleq R\cup \{\ell\}\cup[\delta+1:b-1]\subseteq \rl$. Partition the matrix $H(S_\ell,0:\ell+\tau)$ in the following form:
        \bean
        \left[ \begin{array}{c} U_{\ell}\\ \ml\\ L_{\ell} \end{array} \right], 	
        \eean
        where  $U_\ell$ is as defined in case I,  $\ml\triangleq H(\ell,0:\ell+\tau)$  and $L_\ell\triangleq H([\delta+1:b-1],0:\ell+\tau)$ (see Fig. \ref{fig_construction_B_RE_proof_3} for an example case of $\ell=4$). Note that $\ml=[M_{\ell,0}\ \ M_{\ell,1}\ \ \cdots\ \ M_{\ell,\ell+\tau}]$ is a row vector.
        
        \begin{figure}[!htb]
        	\centering
        	\captionsetup{justification=centering}
        	\includegraphics[scale=0.6]{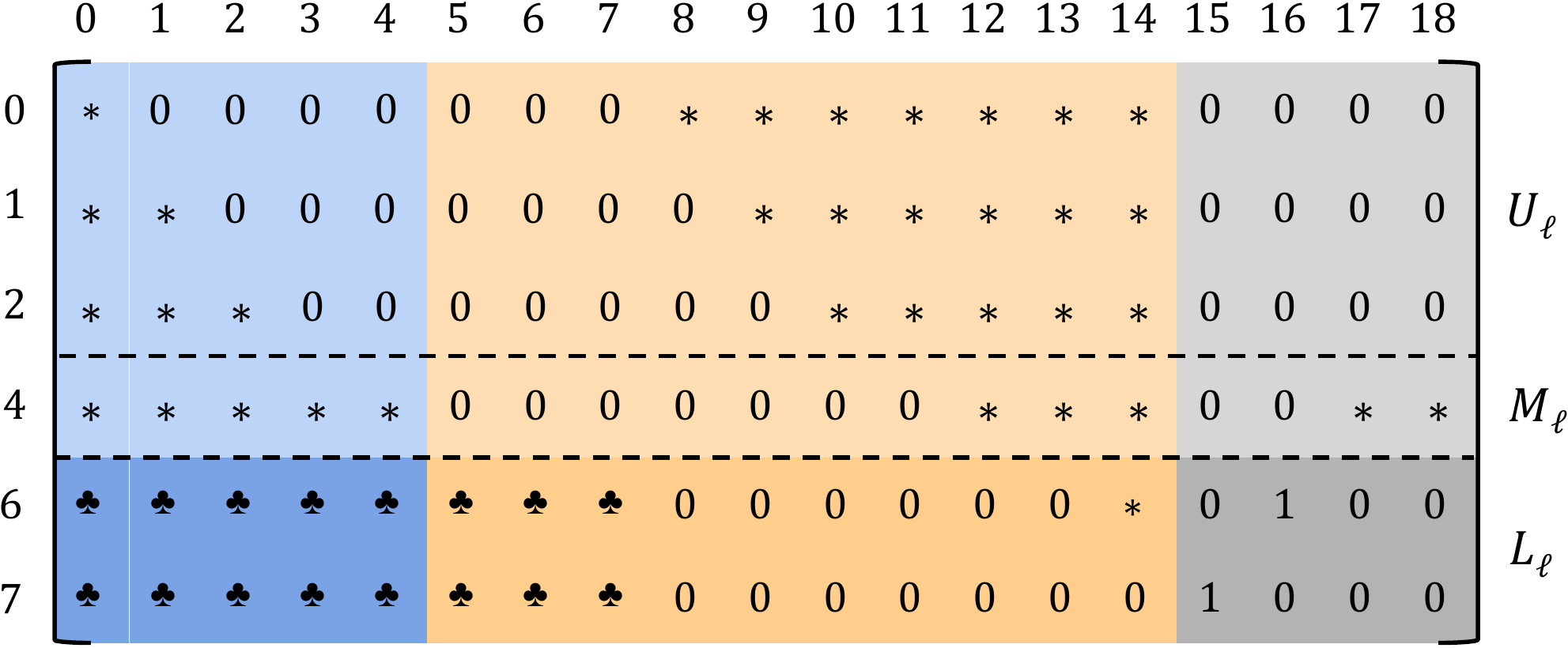}
        	\caption{Here $\ell=4$. In the figure, the submatrix $H(S_\ell,0:\ell+\tau)$ which corresponds to the p-c matrix shown in Fig. \ref{fig:delta_geq_a_example_step4}, is partitioned into \ul, \ml\ and $L_\ell$.}
        	\label{fig_construction_B_RE_proof_3}
        \end{figure}
        
        Now, contrary to the condition R1, assume  there exists a set $\al\subseteq [\ell:\ell+\tau]$ with $|\al|\leq a$ and $\ell\in \al$ such that:
        \beq\label{eq:constrn_b_lc}	
        \underline{h}^{(\ell)}_{\ell}=\sum_{i\in\al\setminus\{\ell\}}a_i\underline{h}^{(\ell)}_{i}, 
        \eeq 
        where $\{a_i\}$ are all non-zero. Equation \eqref{eq:constrn_b_lc} implies the following:
        
        \beq\label{eq:constrn_b_u_lc}	
        \ul(:,\ell)=\sum_{i\in A_\ell\setminus\{\ell\}}a_i\ul(:,i), 
        \eeq

        \beq\label{eq:constrn_b_m_lc}	
        M_{\ell,\ell}=\sum_{i\in A_\ell\setminus\{\ell\}}a_i M_{\ell,i}, 
        \eeq
        
        and 
        
        \beq\label{eq:constrn_b_l_lc}	
        L_\ell(:,\ell)=\sum_{i\in A_\ell\setminus\{\ell\}}a_iL_\ell(:,i). 
        \eeq

        As seen in case I, we have a set of zero columns for \ul\ at the columns indexed by $\mathcal{I}=[a:b-1]\cup[1+\tau:\ell+\tau]$ and any $\leq a$ non-zero columns of \ul\ form an independent set. Thus in order for \eqref{eq:constrn_b_u_lc} to hold, \al\ cannot include any of the non-zero coordinates of \ul. Thus we have:
        
        \beq \label{eq:constrn_b_re1_A_condition_1}
        A_\ell\subseteq [\ell:b-1]\cup[1+\tau:\ell+\tau].
        \eeq
        
        In Fig. \ref{fig_construction_B_RE_proof_4}, we illustrate this for $\ell=4$, with respect to the p-c matrix given in Fig. \ref{fig:delta_geq_a_example_step4}.
        \begin{figure}[!htb]
        	\centering
        	\captionsetup{justification=centering}
        	\includegraphics[scale=0.6]{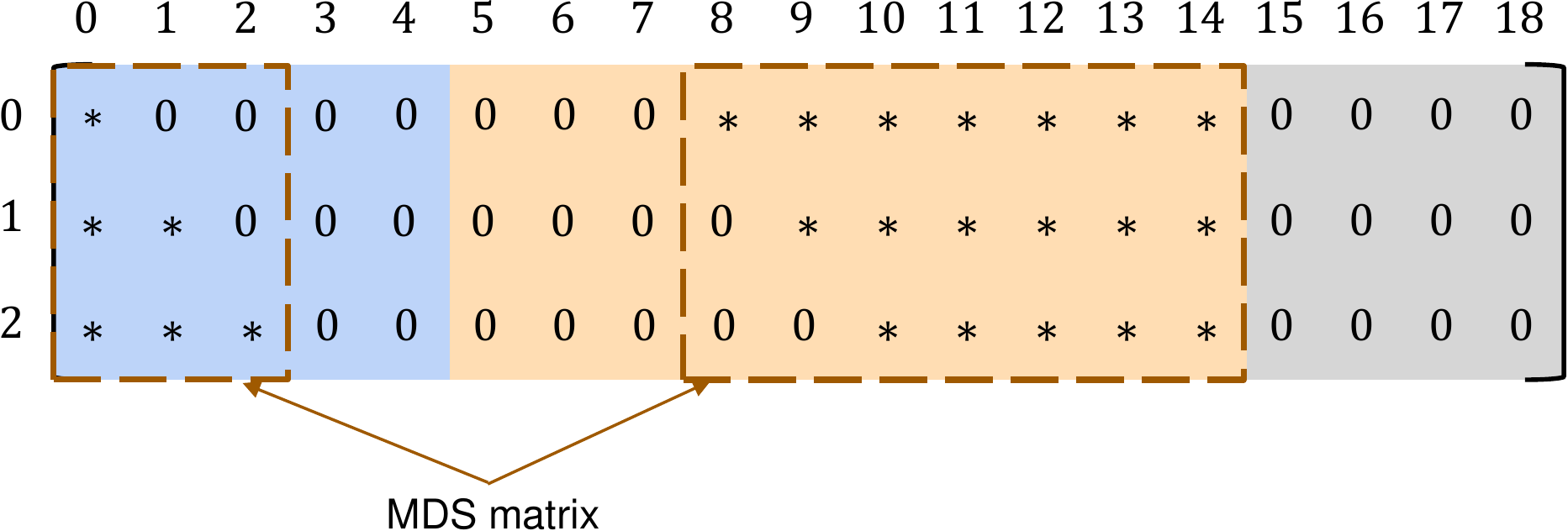}
        	\caption{Here $\ell=4$, $\mathcal{I}=[3:7]\cup [15:18]$. In the figure, we illustrate the submatrix $U_4$ which corresponds to Fig. \ref{fig_construction_B_RE_proof_3}. Any set of $\leq 3$ non-zero columns form an independent set. Thus, in order to satisfy \eqref{eq:constrn_b_u_lc}, $A_4\subseteq[4:18]$ cannot include any of the coordinates $[0:18]\setminus\mathcal{I}$. Hence $A_4\subseteq [4:7]\cup[15:18]$.}
        	\label{fig_construction_B_RE_proof_4}
        \end{figure}

        Now consider \eqref{eq:constrn_b_m_lc}. As $M_{\ell,\ell}\triangleq H(\ell,\ell)\neq 0$, for \eqref{eq:constrn_b_m_lc} to hold, we need at least a single coordinate $i\in \al\setminus\{\ell\}$ such that $M_{\ell,i}\triangleq H(\ell,i)\neq 0$. Because of the constraint \eqref{eq:constrn_b_re1_A_condition_1}, it can be inferred from the p-c matrix structure that, the only way this can be true is by including at least one coordinate from $[a+\tau:\ell+\tau]$ in \al\ (for instance, see Fig. \ref{fig_construction_B_RE_proof_5}).
        
        \begin{figure}[!htb]
        	\centering
        	\captionsetup{justification=centering}
        	\includegraphics[scale=0.6]{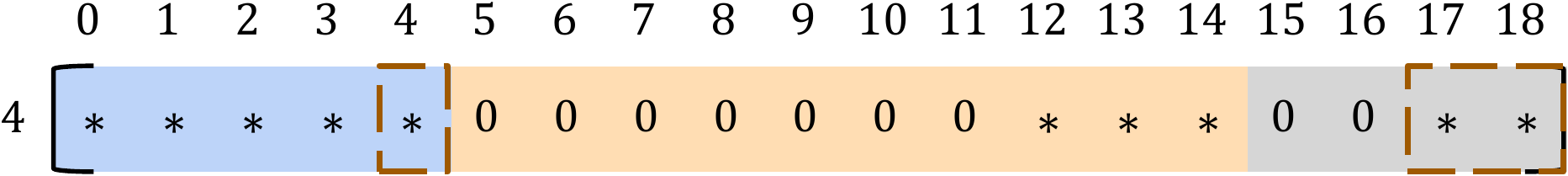}
        	\caption{Here $\ell=4$.  In the figure, we illustrate the submatrix $M_4$ which corresponds to Fig. \ref{fig_construction_B_RE_proof_3}. From \eqref{eq:constrn_b_re1_A_condition_1}, we have $A_4\subseteq [4:7]\cup[15:18]$. Here we have identified the non-zero entries among the coordinates $[4:7]\cup[15:18]$ to be $M_{4.4},M_{4,17}$ and $M_{4,18}$. Hence in order to satisfy \eqref{eq:constrn_b_m_lc}, we require at least one of the coordinates from $[17:18]$ to be in $A_4$.}
        	\label{fig_construction_B_RE_proof_5}
        \end{figure}

        We will turn our attention to \eqref{eq:constrn_b_l_lc} now. We make the following observations on $L_\ell$.
        
        (i) $L_\ell(:,[0:b-1]\cup[1+\tau:a-1+\tau])$ is of the form $[C'\ I_{a-1}]$, where $C'$ is a Cauchy-like submatrix of $C$ defined in Step-d of the construction. Hence $L_\ell(:,[0:b-1]\cup[\tau+1:\tau+a-1])$ is an $(a-1)\times (b+a-1)$ MDS matrix. Also, as $\ell \in[a:\delta-1]$ and $
        \delta<b$, $L_\ell(:,\ell)$ is a column of this MDS matrix.

        (ii) The columns of $L_\ell$ indexed by $[a+\tau:\ell+\tau]$ are all zero-columns.

        In Fig. \ref{fig_construction_B_RE_proof_6}, we identify these observations for the case $\ell=4$ (with respect to the p-c matrix  given in Fig. \ref{fig:delta_geq_a_example_step4}).
        
        \begin{figure}[!htb]
        	\centering
        	\captionsetup{justification=centering}
        	\includegraphics[scale=0.6]{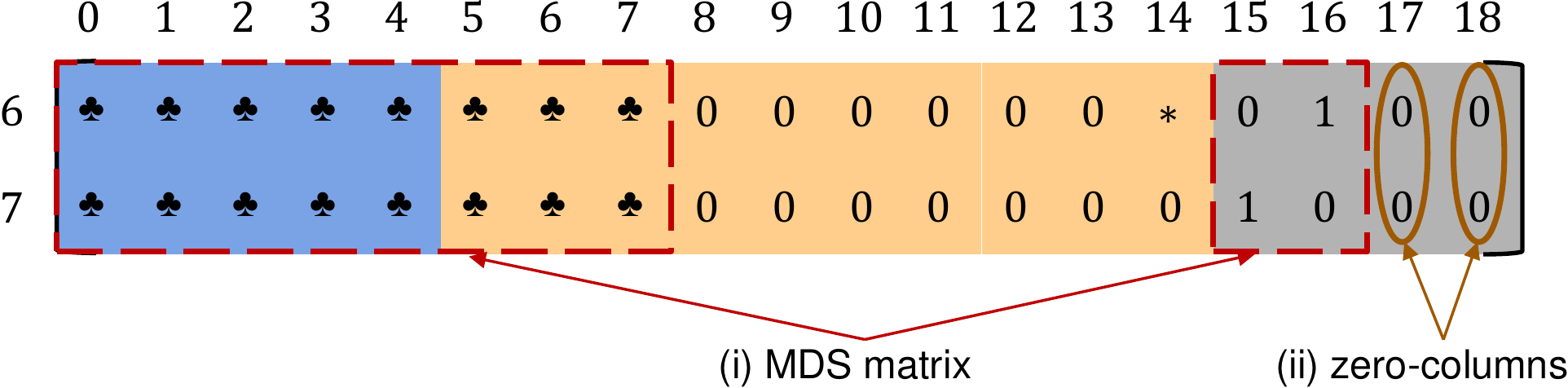}
        	\caption{Here $\ell=4$.  In the figure, we illustrate the submatrix $L_4$ which corresponds to Fig. \ref{fig_construction_B_RE_proof_3}. From \eqref{eq:constrn_b_re1_A_condition_1}, we have $A_4\subseteq [4:7]\cup[15:18]$ and from \eqref{eq:constrn_b_m_lc}, we require at least one of the coordinates from $[17:18]$ to be in $A_4$. However as $L_4(:,17)$ and $L_4(:,18)$ are zero-columns, and $L_4(:,[0:7]\cup[15:16])$ is an MDS matrix, \eqref{eq:constrn_b_l_lc} would result in a dependency across $\leq2$ columns of this MDS matrix, which is a contradiction.}
        	\label{fig_construction_B_RE_proof_6}
        \end{figure}
        
        Now consider the set $[\ell:b-1]\cup[1+\tau:\ell+\tau]$ appearing in \eqref{eq:constrn_b_re1_A_condition_1}. Consider the partition of this set into $A_{\ell,1}\triangleq [\ell:(b-1)]\cup[1+\tau:a-1+\tau]$ and $A_{\ell,2}\triangleq [a+\tau:\ell+\tau]$. By observation (i), all the columns of $L_\ell$ chosen from $A_{\ell,1}$, are columns of the MDS matrix described in observation (i). From observation (i), we also have $L_\ell(:,\ell)$ is a column of this MDS matrix. Similarly, using observation (ii), we have that all columns of $L_\ell$ chosen from $A_{\ell,2}$, are zero-columns. Hence if we include a coordinate (zero-column of $L_\ell$) from $A_{\ell,2}$ in \al, \eqref{eq:constrn_b_l_lc} will result in a dependency involving $\leq (a-1)$ columns of the MDS matrix $L_\ell(:,[0:b-1]\cup[\tau+1:\tau+a-1])$. Clearly, this is a contradiction (see \ref{fig_construction_B_RE_proof_6}, for an example). Thus our assumption on the existence of \al\ is not valid, and hence $H$ satisfies condition R1 for $\ell\in[a:\delta-1]$.

        \eit
        \item Condition R2: Assume there exists a set $\al\subseteq [\delta:n-1]$ with $|\al|\leq a$ such that:
        \beq\label{eq:constrn_b_re2_lc}	
        \sum_{i\in A_\ell}a_i\underline{h}_{i}=0, 
        \eeq
        where $\{a_i\}$ are all non-zero. Let $R\triangleq[0:a-1],S\triangleq[\delta:b-1],S'\triangleq[\delta+1:b-1]$. We make the following observations:
        
        (i) Consider the submatrix $H(R,\delta:\tau+\delta)$. All except the columns given by $\{H(R,j)\mid j\in[b:\tau] \}$ are zero columns. Also, $H(R,[b:\tau])$ is an $a\times (\tau-b+1)$ MDS matrix. 
        
        (ii) $H(S',[\delta:\tau+a-1]\setminus[b:\tau])$ is a $(a-1)\times(2a-1)$ MDS matrix and all the columns $\{H(S',i)\mid i\in [\tau+a:\tau+\delta]\}$ are zero-columns.
        
        (iii) $H(S,[\delta:\tau+a-1]\setminus[b:\tau])$ is a $a\times (2a-1)$ MDS matrix.
        
        (iv) The collection of columns $\left\{ \underline{h}_{j} \mid j\in[\tau+a:\tau+\delta] \right\}$ forms a linearly independent set as required by condition B2, since it is a subset of the larger linearly independent set of $b$ columns, $\left\{ \underline{h}_{j} \mid j\in[\tau-a+1:\tau+\delta] \right\}$.

        We illustrate these observations with respect to  an example in Fig. \ref{fig_construction_B_RE_proof_7}.
        
        \begin{figure}[!htb]
        	\centering
        	\captionsetup{justification=centering}
        	\includegraphics[scale=0.65]{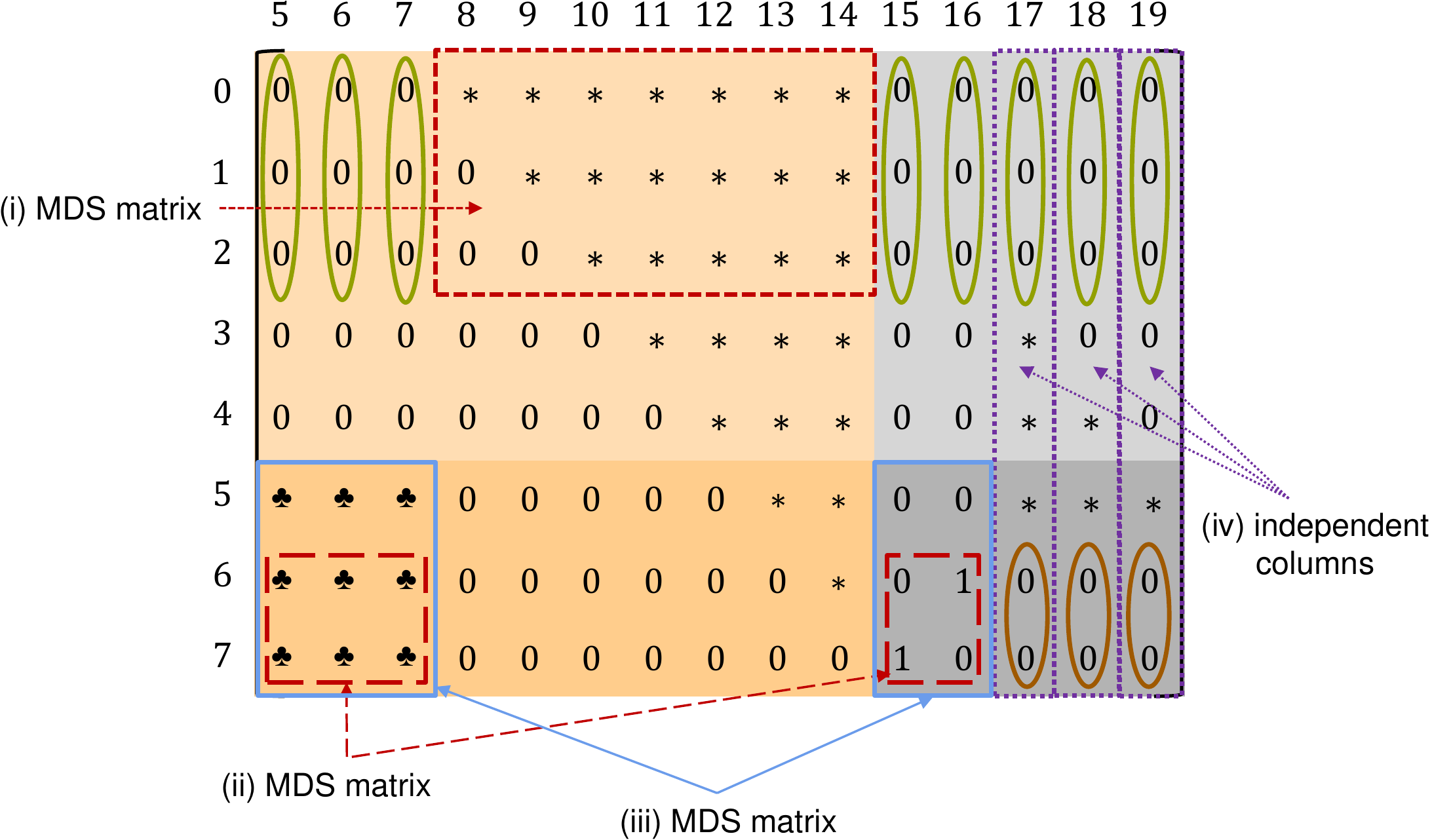}
        	\caption{In this figure,  we identify the observations (i)-(iv)  with respect to the submatrix $H(:,5:19)$, which corresponds to the p-c matrix in Fig. \ref{fig_construction_B_RE_proof_3}.}
        	\label{fig_construction_B_RE_proof_7}
        \end{figure}
        
        From (i), we have that none of the columns from $[b:\tau]$ can be a part of \al. Thus we have $\al\subseteq[\delta:b-1]\cup[\tau+1:\tau+\delta]$. Using observation (ii), we claim that \al\ cannot contain coordinates from both $[\delta:b-1]\cup[\tau+1:\tau+a-1]$ and $[\tau+a:\tau+\delta]$. This is so because if not, we will have a linear dependency amongst $\leq (a-1)$ columns of the MDS matrix $H(S',[\delta:\tau+a-1]\setminus[b:\tau])$, which clearly is not possible. Thus we have either:
        \beq\label{eq:construction_b_re2_al_condition_1}
        \al\subseteq [\delta:\tau+a-1]\setminus[b:\tau]
        \eeq
        
        or
        
        \beq\label{eq:construction_b_re2_al_condition_2}
        \al\subseteq [\tau+a:\tau+\delta].
        \eeq
        However observations (iii) and (iv), respectively, imply that \eqref{eq:construction_b_re2_al_condition_1} and \eqref{eq:construction_b_re2_al_condition_2} are not possible. This contradicts our assumption on the existence of \al\ and property R2 follows.
        \eit
        %%%%%%%%  
        \eit

		\bibliographystyle{IEEEtran}
		\bibliography{streaming_codes,deep_refs}
	\end{document}